\documentclass[12pt]{article}
\usepackage[utf8]{inputenc}
\usepackage[sectionbib]{natbib}

\title{Nonparametric Quantile Regression for Homogeneity Pursuit in Panel Data Models}

\author{Xiaoyu Zhang${}^a$, Di Wang${}^b$, Heng Lian${}^c$ and Guodong Li${}^a$
	\\ \textit{${}^a$Department of Statistics and Actuarial Science, University of Hong Kong}
		\\ \textit{${}^b$School of Mathematical Sciences, Shanghai Jiao Tong University}
	\\ \textit{and ${}^c$Department of Mathematics, City University of Hong Kong}}

\usepackage[margin=1in]{geometry}
\usepackage{graphicx}
\usepackage{mathrsfs}
\usepackage{multirow}
\usepackage{amsfonts}
\usepackage{amsmath}
\usepackage{comment}
\usepackage{amsthm}
\usepackage{color}
\usepackage{mathtools}

\usepackage{multirow}
\usepackage{ctable}
\usepackage{sectsty}
\usepackage[colorlinks]{hyperref}
\usepackage{amsmath}
\usepackage{amssymb}
\usepackage[toc,page]{appendix}
\usepackage[mathscr]{euscript}
\usepackage[toc]{appendix}

\let\counterwithin\relax
\usepackage{chngcntr}
\usepackage{apptools}
\usepackage{booktabs}
\usepackage{lscape}
\usepackage{arydshln}
\usepackage{soul}
\usepackage{caption}
\usepackage{subcaption}
\usepackage[ruled]{algorithm2e}
\AtAppendix{\counterwithin{lemma}{section}}

\newtheorem{lemma}{Lemma}

\newtheorem{remark}{Remark}
\newtheorem{theorem}{Theorem}

\hypersetup{
    colorlinks=true,
    linkcolor=blue,
    citecolor=blue,
    filecolor=magenta,
    urlcolor=cyan,
}

\DeclareMathOperator*{\argmin}{arg\,min}

\newcommand{\bm}{\boldsymbol}

\makeatletter
\newcommand*{\addFileDependency}[1]{% argument=file name and extension
  \typeout{(#1)}
  \@addtofilelist{#1}
  \IfFileExists{#1}{}{\typeout{No file #1.}}
}
\makeatother

\mathtoolsset{showonlyrefs}

\begin{document}

\setlength{\parindent}{16pt}

\date{\today}

\maketitle

\begin{abstract}
Many panel data have the latent subgroup effect on individuals, and it is important to correctly identify these groups since the efficiency of resulting estimators can be improved significantly by pooling the information of individuals within each group.
However, the currently assumed parametric and semiparametric relationship between the response and predictors may be misspecified, which leads to a wrong grouping result, and the nonparametric approach hence can be considered to avoid such mistakes. 	
Moreover, the response may depend on predictors in different ways at various quantile levels, and the corresponding grouping structure may also vary.
To tackle these problems, this article proposes a nonparametric quantile regression method for homogeneity pursuit in panel data models with individual effects, and a pairwise fused penalty is used to automatically select the number of groups.
The asymptotic properties are established, and an ADMM algorithm is also developed.
The finite sample performance is evaluated by simulation experiments, and the usefulness of the proposed methodology is further illustrated by an empirical example.

\end{abstract}


\newpage
\section{Introduction}

Panel data analysis is an important topic in econometrics and has been well studied in economics, medicine, climatology, etc.; see \cite{hsiao2005panel} and \cite{baltagi2008econometric} for a comprehensive exploitation. 
In the classical settings, a fixed or random effect is assumed to take into account the heterogeneity across individuals, while the covariates contribute to the responses in the same way for all individuals.
This makes the estimation more precise and efficient, especially when the number of time points is relatively small.
However, more and more recent evidence has shown that the involvement of covariates may vary for different individuals 
\citep{hsiao1997panel,browning2007heterogeneity,su2013testing}, 
and it usually comes in terms of subgroups. 
For example, in financial markets, stocks from the same sector may share common characteristics 
\citep{ke2015homogeneity}, 
and a spatial geographic grouping is commonly observed in the data from economic geography 
\citep{fan2011sparse,bester2016grouped}. 
As a result, panel data models with latent subgroup effects
\citep{wang2018homogeneity,lian2019homogeneity} have become more and more popular in the literature, 
and this article concentrates on this type of appealing models.

There are two important issues for panel data models with subgroup effects: identifying latent groups and modeling homogeneous structures within each group; see Section \ref{sec:modmeth} for details. 
First, to identify the groups, the K-means approach was proposed for panel data models \citep{lin2012estimation,bonhomme2015grouped,ando2016panel}, and the mixture model can also be used \citep{xu1995alternative,viele2002modeling,shen2015inference}.
However, the number of subgroups needs to be predetermined for both methods and, for mixture models, we are even required to specify the form of distributions for each latent group.
The binary segmentation \citep{ke2016structure} is another method to identify the latent groups in the literature. It is motivated by the idea of structural change detection \citep{bai1998estimating}, and an initial estimation is conducted for each individual separately. 
Recently, the classifier-Lasso method \citep{su2016identifying} has been proposed in the literature for simultaneous group classification and parameter estimation \citep{su2018identifying,su2019sieve}, while the number of subgroups needs to be predetermined.
Along this line, the concave fused penalty \citep{ma2017concave,zhu2018cluster,zhang2019robust} can be used to automatically select the number of subgroups, and the penalty functions are concave on $[0,\infty)$ such that the oracle property can be obtained \citep{fan2001variable,zhang2010nearly}.
This article adopts the concave fused penalty for our homogeneity pursuit.

Secondly, for panel data models with subgroup effects, it is also important to consider a suitable homogeneous structure for each latent group, and the possible misspecification may also lead to a wrong homogeneity pursuit.
Parametric and even linear models have been commonly used in this literature, where parameters are identical within the group and heterogeneous across groups; see \citet{su2016identifying}, \citet{ke2016structure}, \citet{wang2020identifying}, and references therein.
On the other hand, due to the appearance of more and more large-scale datasets in real application, semiparametric and even nonparametric models recently have gained great popularity in panel data analysis since they can provide more flexible modeling in empirical studies \citep{wu2006nonparametric}.
\citet{su2019sieve} proposed a semiparametric time-varying panel data model, and the coefficients are allowed to vary over both individuals and time points. \citet{lian2019homogeneity} considered a single index model with both nonparametric functions and index parameters being homogeneous within groups while heterogeneous across groups.
This article assumes a fully nonparametric structure to each latent group, and it is supposed to enjoy greater flexibility; see Section \ref{sec:real} for more empirical evidence.

In the meanwhile, quantile regression has attracted more and more attention from statisticians and economists since its appearance in \citet{koenker1978regression}. 
Compared with conditional mean models, it enjoys much more flexibility since it can be used to investigate the structures of the response at any quantile level \citep{Koenker2005}.
Quantile regression for panel or longitudinal data can be dated back to \citet{koenker2001quantile} and has been extensively studied in the literature; see, e.g., \citet{koenker2004quantile}, \citet{kato2012asymptotics}, and references therein.
For homogeneity pursuit, \cite{ZHANG201954} considered a composite quantile regression estimation in linear panel data models, while the grouping structure is assumed to be the same for all quantile levels.
It is noteworthy that, for models with latent subgroups, quantile regression can be used to explore different grouping patterns, as well as different structures of responses within a group, for various quantile levels. 
Moreover, it also has a long history to use nonparametric quantile regression to flexibly accommodate nonlinear effects in conditional quantiles \citep{hendricks1992hierarchical,he1994convergence,he1997quantile,yu1998local}.
This motivates us to conduct a nonparametric quantile regression for homogeneity pursuit in panel data models with individual effects.

This article has three main contributions.
First, Section \ref{sec:modmeth} proposes a new homogeneity pursuit method based on nonparametric quantile regression, and the concave fused penalty is employed to identify the unobserved group structure.
It is supposed to be able to capture the relationship between covariates and responses more flexibly. 
Secondly, Section \ref{sec:thm} establishes the oracle property of the concave fused penalized estimator under mild mixing condition, and the proposed subgroup selection procedure is also theoretically verified to be consistent. 
Finally, a novel and efficient alternating direction method of multipliers (ADMM) algorithm is introduced in Section \ref{sec:algo} to overcome the optimization challenge due to the non-differentiability of quantile loss and the large number of penalty terms with non-convexity. 

In addition, Section \ref{sec:sim} conducts simulation experiments to evaluate the finite-sample performance of the proposed methodology, and its usefulness is further demonstrated by a real example in Section \ref{sec:real}. The proofs of theorems are given in a separated  supplementary file.

\section{Model Settings and Methodology} \label{sec:modmeth}

\subsection{Quantile regression models for panel data with subgroups} \label{sec:models}
Consider the panel data of $\{y_{it},x_{it}\}$ with $1\leq i\leq n$ and $1\leq t\leq T$, where $n$ is the number of individuals, $T$ is the number of time points, and $y_{it}$ is the response variable. For the ease of demonstration, without loss of generality, we assume that $x_{it}$ is a univariate covariate of interest, and the methodology for the multivariate case is similar; see Remark \ref{rmk:2} for details.

For any level $0<\tau<1$, the $\tau$-th quantile of $y_{it}$ conditional on the covariate $x_{it}$ can be denoted by
\begin{equation}\label{eq:ConditionalQuantile}
	Q_\tau(y_{it}|x_{it}) = \mu_{i,\tau} + m_{i,\tau}(x_{it}),
\end{equation}
where $\mu_{i,\tau}$ is the individual effect, $m_{i,\tau}(\cdot)$ is an unknown smooth function, and both of them may vary for different individuals and quantile levels.
Note that model \eqref{eq:ConditionalQuantile} is not identifiable since, for any nonzero constant $c_{i,\tau}$,
\begin{equation}
Q_\tau(y_{it}|x_{it}) = (\mu_{i,\tau}+c_{i,\tau}) + (m_{i,\tau}(x_{it})-c_{i,\tau})\equiv \widetilde{\mu}_{i,\tau} + \widetilde{m}_{i,\tau}(x_{it}),
\end{equation}
and hence the following constraint on $m_{i,\tau}(\cdot)$ is added,
\begin{equation}\label{eq:constraint}
\int_{\mathcal{X}}m_{i,\tau}(x)dx=0 \hspace{3mm}\text{for any $\tau$ and $1\leq i\leq n$},
\end{equation}
where $\mathcal{X}$ is the support of $x_{it}$.
Moreover, define the random error $e_{it}(\tau)=y_{it}-\mu_{i,\tau} - m_{i,\tau}(x_{it})$, and it holds that $\mathbb{P}(e_{it}(\tau)\leq 0|x_{it})=\tau$ for all $1\leq i\leq n$ and $1\leq t\leq T$.

Individuals of panel data are usually sampled from different backgrounds and are hence with different individual characteristics so that an abiding feature of the data is its heterogeneity. 
Thus assuming homogeneity, that is, all $m_{i,\tau}(\cdot)$'s are equal, may be too strong. 
To solve the problem, one can estimate the functions for each individual separately. It is actually under the assumption of total heterogeneity, and we may encounter difficulty in obtaining consistent estimation when the number of time points $T$ is fixed or relatively small.
On the other hand, for some panel data, there is a subgroup effect although it is usually unobserved, and the individuals within one group may share the same structure up to an individual effect. 
This makes it possible to achieve a more precise and efficient estimation for $m_{i,\tau}(\cdot)$'s by pooling the information of all individuals within one group.

As a result, we assume that, at a quantile level $\tau$, all individuals belong to $K$ different latent groups, and the nonparametric function, $m_{i,\tau}(\cdot)$'s,  are heterogeneous across groups but homogeneous within a group.
Specifically, there exists an unknown partition $\{G_{k,\tau}:k=1,2,\dots,K\}$ of the individual index set $\{1,2,\dots,n\}$, which can be used to denote the subgroup membership. We impose the following group structure on the model:
\begin{equation}
	m_{i,\tau}(\cdot)=\begin{cases}
		m_{(1),\tau}(\cdot), & \text{when }i\in G_{1,\tau}\\
		m_{(2),\tau}(\cdot), & \text{when }i\in G_{2,\tau}\\
		\vdots & \vdots\\
		m_{(K),\tau}(\cdot), & \text{when }i\in G_{K,\tau}
	\end{cases}
\end{equation}
where $m_{(k),\tau}(\cdot)$'s are unknown functions to be estimated, and $K$ is an unknown integer depending on the quantile level $\tau$.

\begin{remark}\label{rmk:1}
	For each individual $i$, consider the case that response $y_{it}$ has the conditional quantile at model \eqref{eq:ConditionalQuantile} for all $\tau\in[0,1]$.
	Suppose that the data generating process for $y_{it}$ exists, and it then holds that $\mu_{i,\tau} + m_{i,\tau}(x)$ is an increasing function with respect to $\tau$ for all $x\in\mathcal{X}$.
	Moreover, the corresponding data generating process has the form of $y_{it}=\mu_{i,U_{it}} + m_{i,U_{it}}(x_{it})$, where the innovations $\{U_{it},1\leq t\leq T\}$ are independent and follow uniform distribution on $[0,1]$; see \cite{Koenker2005} and \cite{Koenker_Xiao2006}.
	In addition, consider the function $m_{i,\tau}(x)$ with respect to $\tau$ and $x$ for a specific individual $i$. It may share the same form with some individuals at one quantile level $\tau$, that is, they belong to the same subgroup, while its group mates may change for another quantile level.
	Refer to the data generating process in Section 5.2 for an illustrating example.
	As a result, compared with panel data models with subgroup structures in conditional mean \citep[e.g.][]{su2016identifying,su2018identifying}, model \eqref{eq:ConditionalQuantile} allows a more flexible framework for homogeneity pursuit.
	
\end{remark}

\subsection{Nonparametric estimation with a concave fused penalty} \label{sec:method}

This article uses polynomial splines to flexibly approximate the unknown smooth functions.
Without loss of generality, assume that the support of $x_{it}$ is $[0,1]$.
Let $0=\nu_0<\nu_1<\dots<\nu_{H'}<\nu_{H'+1}=1$ be a partition of $[0,1]$, and denote the $H'+1$ subintervals by $I_{h'}=[\nu_{h'},\nu_{h'+1})$ with $0\leq h'\leq H'$, where $H'\equiv H'(nT)$ increases with the value of $nT$.
Throughout the numerical studies, we simply take equally spaced knots, and it works reasonably well; see simulation results in Section \ref{sec:sim}. 

Denote by $S_q$ the space of polynomial splines of order $q$.
Note that any function $f(\cdot)$ from $S_q$ satisfies:
(i) on each $I_{h'}$, $0\leq h'\leq H'$, $f(\cdot)$ is a polynomial of degree $q-1$;
(ii) $f(\cdot)$ is globally $q-2$ times continuously differentiable on $[0,1]$.
For the knots sequences $0=\nu_0<\nu_1<\dots<\nu_{H'}<\nu_{H'+1}=1$, the $q$-th order B-splines are defined recursively,
\begin{equation}
    \begin{split}
        B_h^1(x)&=\begin{cases}
            1, & \nu_h\leq x<\nu_{h+1}\\
            0, & \text{otherwise}
        \end{cases},\\
        \text{and}~B_h^q(x)&=\frac{x-\nu_h}{\nu_{h+q-1}-\nu_h}B_h^{q-1}(x)+\frac{\nu_{h+q}-x}{\nu_{h+q}-\nu_{h+1}}B_{h+1}^{q-1}(x).
    \end{split}
\end{equation}
Consider the B-spline basis vector, $\bm{B}(s):=(B_1^q(s),\dots,B_H^q(s))^\top$, where $H$ is determined by the number of knots $H'$ and B-spline order $q$, and we may approximate $m_{i,\tau}(x)$ by
\begin{equation}
	m_{i,\tau}(x)\approx \sum_{h=1}^H B_h^q(x)\beta_{ih}(\tau) = \bm{B}(x)^\top\bm{\beta}_{i,\tau},
\end{equation}
where $\bm{\beta}_{i,\tau}=(\beta_{i1}(\tau),\dots,\beta_{iH}(\tau))^\top\in\mathbb{R}^H$ is a coefficient vector.
In the meanwhile, to tackle the identification restriction at \eqref{eq:constraint}, we further impose that
\begin{equation}\label{eq:contraint2}
	\int_0^1\bm{B}(x)^\top\bm{\beta}_{i,\tau}dx=\sum_{h=1}^H\beta_{ih}(\tau)\int_0^1B_h^q(x)dx=0.
\end{equation}

For equally spaced knots on the interval $[0,1]$, by the definition of the B-spline basis functions, it can be verified that all $H$ integrals $\int_0^1B_h^q(x)dx$ with $1\leq h\leq H$ are identical, and hence the constraint at \eqref{eq:contraint2} can be simplified to  $\sum_{h=1}^H\beta_{ih}(\tau)=\bm{1}_H^\top\bm{\beta}_{i,\tau}=0$, where $\bm{1}_H=(1,1,\dots,1)^\top\in\mathbb{R}^H$.
Similar to the existing literature on the panel data models with fixed effects \citep{baltagi2008econometric}, we first construct an orthogonal rotation matrix $\bm{O}\in\mathbb{R}^{H\times H}$ such that  $\bm{O}^\top=\left[H^{-1/2}\bm{1}_H,\widetilde{\bm{O}}^\top\right]$, where $\widetilde{\bm{O}}\in\mathbb{R}^{(H-1)\times H}$.
Denote by $\bm{\Pi}(x)=H^{1/2}\widetilde{\bm{O}}\bm{B}(x)\in\mathbb{R}^{H-1}$ and $\bm{\theta}_{i,\tau}=H^{-1/2}\widetilde{\bm{O}}\bm{\beta}_{i,\tau}\in\mathbb{R}^{H-1}$ the transformed basis function and transformed coefficient vector, respectively.
Then the smooth function $m_{i,\tau}(x)$ with the constraint at \eqref{eq:constraint} can be approximated by
\begin{equation}\label{eq:approximation}
	m_{i,\tau}(x)\approx\bm{B}(x)^\top\bm{\beta}_{i,\tau} = \bm{\Pi}(x)^\top\bm{\theta}_{i,\tau}.
\end{equation}
The normalization of $H^{1/2}$ in $\bm{\Pi}(x)$ is used to guarantee that the eigenvalues of $\int_0^1\bm{\Pi}(s)\bm{\Pi}(s)^\top ds$ are bounded away from zero and infinity.
Such normalization is only adopted for the convenience of theoretical derivations, and it is neither essential nor necessary in practice. 

For other choices of knots for B-spline approximation with $H$ basis functions, denote $\bm{b}=(b_1,\dots,b_H)^\top$ with $b_h:=\int_0^1B_h^q(x)dx$, and these $b_j$'s may not be identical. 
We alternatively consider another orthogonal rotation matrix $\bm{O}^\top=\left[\bm{b}/\|\bm{b}\|_2,\widetilde{\bm{O}}^\top\right]\in\mathbb{R}^{H\times H}$, and the estimation with constraints similarly can be transformed into an unconstrained problem.
Note that there are many different choices of $\widetilde{\bm{O}}\in\mathbb{R}^{(H-1)\times H}$, while they all lead to the same estimate of $m_{i,\tau}(x)$; see \cite{Wu_Li2014}.
We then can arbitrarily choose one in real applications.

\begin{remark}\label{rmk:2}
The univariate spline approximation can be naturally extended to its multivariate analogs. Consider $p$-dimensional covariates $\bm{x}_{it}=(x_{1it},\dots,x_{pit})^\top\in\mathbb{R}^p$, and we can use B-spline basis functions to represent functions for each coordinate in $\bm{x}_{it}$. Then, as discussed in \citet{friedman2001elements}, the $H_1 \times H_2 \times\cdots\times H_p$ dimensional tensor product basis can be used to approximate the $p$-dimensional smooth function
\begin{equation}
	m_{\tau}(\bm{x}_{it})\approx \sum_{h_1=1}^{H_1}\sum_{h_2=1}^{H_2}\cdots\sum_{h_p=1}^{H_p}B_{h_1}(x_{1it})B_{h_2}(x_{2it})\cdots B_{h_p}(x_p)\beta_{h_1h_2\dots h_p}(\tau),
\end{equation}
and the constraint due to \eqref{eq:constraint} and \eqref{eq:contraint2} can be settled down by a  reparameterization similar to that for the univariate case.
Note that the total number of parameters in the $p$-dimensional spline approximation will increase exponentially with $p$, and hence a larger sample size will be needed for the inference. 
\end{remark}

If one assumes heterogeneity without taking into account the subgroup structure, by using the B-spline approximation at \eqref{eq:approximation}, the estimators could then be obtained by minimizing the loss function,
\begin{equation}
	\frac{1}{nT}\sum_{i=1}^n\sum_{t=1}^T\rho_\tau(y_{it}-\mu_{i,\tau}-\bm{\Pi}(x_{it})^\top\bm{\theta}_{i,\tau}),
\end{equation}
where $\rho_\tau(v)=\tau v-vI\{v\leq 0\}$ is the classical check loss function for quantile regression. Secondly, if assuming homogeneity, i.e. the B-spline coefficient vectors for all individuals are the same or $\bm{\theta}_\tau=\bm{\theta}_{1,\tau}=\bm{\theta}_{2,\tau}=\dots=\bm{\theta}_{n,\tau}$, the estimated B-spline coefficient vector for all individuals can then be obtained by minimizing
\begin{equation}
	\frac{1}{nT}\sum_{i=1}^n\sum_{t=1}^T\rho_\tau(y_{it}-\mu_{i,\tau}-\bm{\Pi}(x_{it})^\top\bm{\theta}_\tau).
\end{equation}
Thirdly, if the group memberships $\{G_1,\dots,G_K\}$ were known, one could combine all individuals belonging to the same group and simultaneously estimate the B-spline coefficient vectors for all groups, 
\begin{equation}\label{eq:oracle}
	\begin{split}
		&\{\widehat{\mu}_{1,\tau}^\text{o},\dots,\widehat{\mu}_{n,\tau}^\text{o},\widehat{\bm{\theta}}_{(1),\tau}^\text{o},...,\widehat{\bm{\theta}}_{(K),\tau}^\text{o}\}\\
		=&\argmin \frac{1}{nT}\sum_{k=1}^K\sum_{i\in G_k}\sum_{t=1}^T\rho_\tau(y_{it}-\mu_{i,\tau}-\bm{\Pi}(x_{it})^\top\bm{\theta}_{(k),\tau}),
	\end{split}
\end{equation}
where $\bm{\theta}_{(k),\tau}$ is the coefficient vector for group $k$ at quantile $\tau$.
As the estimator is based on the true subgroup structure, we call $\{\widehat{\mu}_{1,\tau}^\text{o},\dots,\widehat{\mu}_{n,\tau}^\text{o},\widehat{\bm{\theta}}_{(1),\tau}^\text{o},\dots,\widehat{\bm{\theta}}_{(K),\tau}^\text{o}\}$ the oracle estimator.

This article considers the case that the grouping structure is unknown, and a concave pairwise fused penalized approach is employed to achieve B-spline coefficient estimation and subgroup identification simultaneously. The objective function can be formulated as
\begin{equation} \label{eq:objective}
	\frac{1}{nT}\sum_{i=1}^n\sum_{t=1}^T\rho_\tau(y_{it}-\mu_{i,\tau}-\bm{\Pi}(x_{it})^\top\bm{\theta}_{i,\tau})+\binom{n}{2}^{-1}\sum_{i<j}p_\lambda(\|\bm{\theta}_{i,\tau}-\bm{\theta}_{j,\tau}\|_2),
\end{equation}
where $p_\lambda(\cdot)$ is a concave penalty function on $[0,\infty)$ with a tuning parameter $\lambda$.
The pairwise fused penalty shrinks some of the pairs $\|\bm{\theta}_{i,\tau}-\bm{\theta}_{j,\tau}\|_2$ to zero, and the tuning parameter $\lambda$ controls the number of selected subgroups. As a result, we can partition the individuals into several subgroups based on the values of $\bm{\theta}_{i,\tau}-\bm{\theta}_{j,\tau}$, for all $i<j$. 
Specifically, for a fixed penalty parameter $\lambda$, denote the above penalized estimators by $\widehat{\bm{\theta}}_{i,\tau}(\lambda)$ with $1\leq i \leq n$. Suppose that there are $\widehat{K}(\lambda)$ or $\widehat{K}$ distinct values. Denoted by $\widehat{\bm{\theta}}_{(1),\tau}(\lambda),\dots,\widehat{\bm{\theta}}_{(\widehat{K}),\tau}(\lambda)$, and then the estimated groups can be defined as $\widehat{G}_{k,\tau}(\lambda)=\{i:\widehat{\bm{\theta}}_{i,\tau}(\lambda)=\widehat{\bm{\theta}}_{(k),\tau}(\lambda)\}$ for $1\leq k\leq \widehat{K}$.

There are many penalty functions available in the literature, and they include the Lasso proposed by \cite{tibshirani1996regression}, adaptive Lasso by \cite{zou2006adaptive}, smoothly clipped absolute deviation (SCAD) by \cite{fan2001variable}, and minimax concave penalty (MCP) by \cite{zhang2010nearly}.
This article considers the SCAD function, a commonly used penalty function in high-dimensional models, and it has the form of
\begin{equation}
	\begin{split}
		p_\lambda(u) = \begin{cases}
			\lambda u, & \text{if }0\leq u <\lambda,\\
			(a-1)^{-1}[a\lambda u-(u^2+\lambda^2)/2], & \text{if }\lambda\leq u\leq a\lambda,\\
			(a+1)\lambda^2/2, & \text{if }u>a\lambda,
		\end{cases}
	\end{split}
\end{equation}
where $a>2$ is a pre-specified constant, and the concaveness of $p_{\lambda}(x)$ on $[0,\infty)$ can guarantee the oracle property of many high-dimensional problems; see \citet{fan2001variable}.
The Lasso penalty with $p_{\lambda}(x)=\lambda |x|$ is expected to lead to biased estimates for variable selection problems \citep{zou2006adaptive} and may not be able to correctly recover subgroups. Generally speaking, a variety of penalties can be tried in the literature of penalized variable selection, and there is no clear winner. 
The SCAD pairwise fusion penalty will be used in this article to achieve accurate nonparametric estimation and subgroup identification.

\subsection{Tuning parameter selection} \label{sec:SIC}
The proposed concave pairwise fusion penalized estimation in the previous subsection relies on the choice of tuning parameter $\lambda$, which can heavily affect the performance on both nonparametric estimation and subgroup selection.
This article suggests to select $\lambda$ via Schwarz information criterion (SIC), and it has the form of
\begin{equation}\label{eq:SIC}
    \text{SIC}(\lambda)=\log\left(\sum_{i=1}^n\sum_{t=1}^T\rho_\tau(y_{it}-\widehat{\mu}_{i}(\lambda)-\bm{\Pi}(x_{it})^\top\widehat{\bm{\theta}}_{i,\tau}(\lambda))\right)+\widehat{K}(\lambda)H\log(nT)/(nT),
\end{equation}
where $\{\widehat{\mu}_i(\lambda),\widehat{\bm{\theta}}_{i,\tau}(\lambda)\}$'s and $\widehat{K}(\lambda)$ are the estimators and number of groups selected for the given $\lambda$, respectively. Note that SIC is widely known as Bayesian information criterion (BIC) outside the area of quantile regression; however, we follow the quantile regression literature \citep{kim2007quantile,wang2009quantile} to call it SIC. Its theoretical justification will be established later in Section \ref{sec:thm}.

\section{Asymptotic Theory} \label{sec:thm}

This section investigates the asymptotic properties of both oracle and pairwise fusion penalized estimators, and the tuning parameter selection based on the Schwartz information criterion is also justified theoretically. 

To measure the temporal dependency in the panel data, we adopt the $\alpha$-mixing condition \citep{bradley2005basic}. To be specific, for any stationary process $\{X_t\}$, we denote the $\alpha$-mixing coefficients, for $l\geq 1$,
\begin{equation}
    \alpha(l)=\alpha(X_{-\infty:t},X_{t+l:\infty}).
\end{equation}
Throughout this section, we impose the following assumptions.
\begin{enumerate}
	\item[(A1)] Let $e_{it}(\tau)=y_{it}-\mu_{i,\tau}-m_{i,\tau}(x_{it})$. For each $\tau\in(0,1)$,  random variables $\{e_{it}(\tau),x_{it}\}$ are independent across $i$, and are strictly stationary and $\alpha$-mixing for each $i$, with mixing coefficient $\alpha(l)=O(r^l)$ for some $0<r<1$.
	\item[(A2)] The distribution of $x_{it}$ is supported on $[0,1]$ with density bounded away from zero and infinity, for all $1\leq i\leq n$.
	\item[(A3)] Let $f_i(\cdot|x_{it})$ be the conditional density of $e_{it}(\tau)$. Assume that $f_i(\cdot|x_{it})$ is bounded and bounded away from zero in a neighborhood of zero, uniformly over the support of $x_{it}$, for all $1\leq i\leq n$. The derivative of $f_i(\cdot|x_{it})$ is uniformly bounded in a neighborhood of zero over the support of $x_{it}$.
	\item[(A4)] The functions $m_{i,\tau}(\cdot)$ are in the H\"{o}lder space of order $d\geq2$, for all $1\leq i\leq n$; that is, $|m_{i,\tau}^{(u)}(x)-m_{i,\tau}^{(u)}(y)|\leq C|x-y|^v$ for $d=u+v$ and $u$ is the largest integer strictly smaller than $d$, where $m_{i,\tau}^{(u)}(\cdot)$ is the $u$-th derivative of $m_{i,\tau}(\cdot)$.
	\item[(A5)] There is a partition $\{G_{1,\tau},\dots,G_{K,\tau}\}$ of $\{1,\dots,n\}$ such that $m_{i,\tau}(\cdot)= m_{(k),\tau}(\cdot)$ for all $i\in G_{k,\tau}$. For each $\tau\in(0,1)$, the number of groups $K$ is fixed, and $|G_{k,\tau}|/n\to c_{k,\tau}$ for some $c_{k,\tau}\in(0,1)$.
	\item[(A6)] Assume that $T$ and $H$ diverge to infinity, and $n$ possibly diverges to infinity. 
\end{enumerate}

Assumptions (A1) contains the mild geometric mixing condition for the data, which is common in panel data analysis. The boundedness assumption for $x_{it}$ in (A2) is tied to the estimation approach using B-splines whose basis functions are usually constructed on a bounded interval. In practice, one can always transform the predictors into $[0,1]$ before the analysis. (A3) contains some conditional density assumptions and (A4) contains smoothness conditions, which are standard in nonparametric quantile regression literature \citep{he1994convergence}. It is noteworthy that the distributional condition in  (A3) encompasses a large class of distributions and is more general than the Gaussian or sub-Gaussian error condition in the penalized least squares methods in \citet{ma2017concave}. Assumption (A5) requires that the number of subgroups is fixed and none of the subgroups vanishes asymptotically, and it agrees with many applications, in which $K$ is expected to be small, and thus substantial reduction of the unknown parameters can be achieved.
The asymptotic theory in this section is in the sense of (A6). 

To derive the asymptotic theory for the proposed estimator, we first need to specify the population true values of B-spline coefficient vectors. For any subgroup $k$, let $f_{(k)}(\cdot|x_{it})=|G_{k,\tau}|^{-1}\sum_{i\in G_{k,\tau}}f_i(\cdot|x_{it})$ be the average conditional density of the error term for all $i\in G_{k,\tau}$. Throughout the article, $\bm{\theta}_{0(k),\tau}$ is defined as the spline coefficients in the best B-spline approximation of $m_{(k),\tau}(\cdot)$,
\begin{equation}
    \bm{\theta}_{0(k),\tau}=\argmin_{\bm{\theta}\in\mathbb{R}^{H-1}}\mathbb{E}\{f_{(k)}(0|x_{it})[m_{(k),\tau}(x_{it})-\bm{\Pi}(x_{it})^\top\bm{\theta}]^2\}.
\end{equation}
By the approximation theory of B-splines \citep{de1978practical}, it holds that
\begin{equation}
    \sup_{x\in[0,1]}|m_{(k),\tau}(x)-\bm{\Pi}(x)^\top\bm{\theta}_{0(k),\tau}|\leq CH^{-d},
\end{equation}
for some constant $C$ independent of $n$, $T$ and $H$.

We first state asymptotic properties of the oracle estimator $\widehat{\bm{\theta}}_{(k),\tau}^\text{o}$, assuming that the group memberships are known. Denote by $\{x_{it}\}$ the collection of $x_{it}$ for all $1\leq i\leq n$ and $1\leq t\leq T$.

\begin{theorem} \label{thm:1}
	Suppose that $H^2\log(n)^2/T\to0$, $H^3\log(T)/T\to0$, $Hn^2\log(n)^3/T\to0$, and Assumptions (A1)-(A6) hold. For the oracle estimator $\widehat{\bm{\theta}}^\textup{o}_{(k),\tau}$ with $1\leq k\leq K$ at \eqref{eq:oracle}, it holds that
	\begin{itemize}
		\item[(i)] $\max_{1\leq i\leq n}|\widehat{\mu}_{i,\tau}-\mu_{i,\tau}|=O_p(\sqrt{\log(n)/T}+H^{-d})$ and $\|\widehat{\bm{\theta}}_{(k),\tau}^\textup{o}-\bm{\theta}_{0(k),\tau}\|_2=O_p(\xi(n,T))$, where $\xi(n,T)=\sqrt{H/(nT)}+H^{-d}$.
		\item[(ii)] For any $x\in(0,1)$,
		\begin{equation}
			\frac{\bm{\Pi}(x)^\top(\widehat{\bm{\theta}}_{(k),\tau}^\textup{o}-\bm{\theta}_{0(k),\tau})}{\sqrt{\widehat{\textup{Var}}(\widehat{m}_{(k),\tau}(x)|\{x_{it}\})}}\to N(0,1),
		\end{equation}
		where $\widehat{\textup{Var}}(\widehat{m}_{(k),\tau}(x)|\{x_{it}\})=\tau(1-\tau)\bm{\Pi}^\top(x)(\bm{Z}_k^\top \bm{f}_k\bm{Z}_k)^{-1}(\bm{Z}_k^\top\bm{Z}_k)(\bm{Z}_k^\top \bm{f}_k\bm{Z}_k)^{-1}\bm{\Pi}(x)$, $\bm{Z}_k=[\bm{\Pi}(x_{11})I\{1\in G_k\},\dots,\bm{\Pi}(x_{1T})I\{1\in G_k\},\bm{\Pi}(x_{21})I\{2\in G_k\},\dots,\bm{\Pi}(x_{nT})I\{n\in G_k\}]^\top$ and $\bm{f}_k=\textup{diag}(f_k(0|x_{11})I\{1\in G_k\},\dots,f_k(0|x_{1T})I\{1\in G_k\},\dots,f_k(0|x_{nT})I\{n\in G_k\})$.
		Moreover, if $H(nT)^{-1/(2d+1)}\to \infty$,
		\begin{equation}\label{eq:asymp_norm}
		\frac{\widehat{m}_{(k),\tau}(x)-m_{(k),\tau}(x)}{\sqrt{\widehat{\textup{Var}}\left(\widehat{m}_{(k),\tau}(x)|\{x_{it}\}\right)}}\to N(0,1),
		\end{equation} 
		where $\widehat{m}_{(k),\tau}(x)=\bm{\Pi}(x)^\top\widehat{\bm{\theta}}_{(k),\tau}^\textup{o}$.
	\end{itemize}
\end{theorem}

For panel data models, the conditional quantile model with individual effects is well known for the incidental parameter problem, as the number of parameters for individual effects increases with the number of individuals $n$. To consistently estimate the parameter $\bm{\theta}_{(k),\tau}$, as in \citet{kato2012asymptotics}, it is necessary to assume that 
$T$ and $n$ satisfy that $H^2\log(n)^2/T\to0$.
The estimation error rate in the argument (i) of Theorem \ref{thm:1} has a standard form as in nonparametric regression, and the two terms correspond to bias and variance, respectively. Note that, by taking advantage of group structures, the effective sample size becomes $nT$, while it is $T$ for the case of full heterogeneity. In order to minimize the estimation error rate, the optimal choice of $H$ obviously satisfies $c_1 \leq H(nT)^{-1/(2d+1)} \leq c_2 $ for two constants $0<c_1<c_2<\infty$. 

The condition for mean-zero asymptotic normality is more stringent. To make the asymptotic bias of the fixed effects negligible, we further require that $T$ is much larger than $n$, namely $Hn^2\log(n)^3/T\to0$.
Moreover, to remove the B-spline approximation error asymptotically and make inference on the smooth function, we need to set a large number of knots such that $H(nT)^{-1/(2d+1)}\to\infty$, and the asymptotic normality at \eqref{eq:asymp_norm} enables us to construct point-wise confidence intervals for the estimated smooth function $\widehat{m}_{(k),\tau}(x)$. To fit the conditional density functions in the asymptotic variance at \eqref{eq:asymp_norm}, we may consider the method in \cite{hendricks1992hierarchical} or a standard kernel conditional density estimation such as the re-weighted Nadaraya-Watson method \citep{de2003conditional}.

The following theorem shows that the oracle estimator is a local minimizer of the proposed penalized objective function with probability approaching one.
\begin{theorem}[Oracle property] 
\label{thm:2}
Suppose that the 
conditions in Theorem \ref{thm:1} hold. If (a) $\lambda^{-1}[\sqrt{H}\xi(n,T)+\sqrt{(H/T)\log(T)\log(nH\log(T))}]\to0$ and (b) $\lambda^{-1}\cdot\min_{k_1\neq k_2}\|\bm{\theta}_{0(k_1),\tau}-\bm{\theta}_{0(k_2),\tau}\|_2\to\infty$, then there exists a local minimizer in \eqref{eq:objective} such that $\widehat{\bm{\theta}}_{i,\tau}(\lambda)$ can be clustered into $K$ groups, denoted by $\widehat{\bm{\theta}}_{(k),\tau}(\lambda)$ with $1\leq k\leq K$, and, up to a permutation of subgroups and as $T\to\infty$,
	\begin{itemize}
		\item[(i)] $\mathbb{P}(\widehat{\bm{\theta}}_{(k),\tau}(\lambda)=\widehat{\bm{\theta}}_{(k),\tau}^\textup{o})\to 1$, and
		\item[(ii)] $\mathbb{P}(\widehat{G}_{k,\tau}(\lambda)=G_{k,\tau},\text{ for }k=1,\dots,K)\to1$.
	\end{itemize}
\end{theorem}

Theorem \ref{thm:2} presents the oracle property for pairwise fusion penalized estimators. This result, together with the asymptotic properties of the oracle estimator in Theorem \ref{thm:1}, directly leads to the convergence rate and asymptotic normality of $\widehat{\bm{\theta}}_{i,\tau}(\lambda)$. The two additional assumptions (a) and (b) imply that $\sqrt{H}\xi(n,T)+\sqrt{(H/T)\log(T)\log(nH\log(T))} \ll \lambda \ll \min_{k_1\neq k_2}\|\bm{\theta}_{0(k_1),\tau}-\bm{\theta}_{0(k_2),\tau}\|_2$. 
Specifically, the condition (a) indicates the minimum rate of the tuning parameter $\lambda$, and the condition (b) requires that the minimum difference of the B-spline coefficients in different groups cannot be too small.

\begin{remark}\label{rmk:3}
    If all smooth functions $m_{(k),\tau}(\cdot)$ remain the same as $T$ and $n$ diverge, the normalization of B-spline basis implies that $\min_{k_1\neq k_2}\|\bm{\theta}_{0(k_1),\tau}-\bm{\theta}_{0(k_2),\tau}\|_2$ is bounded away from zero. Hence, the conditions (a) and (b) require that $\lambda\to0$ and $\sqrt{(H/T)\log(T)\log(nH\log(T))}\to 0$. In other words, the number of time points $T$ is required to diverge to infinity, and if $n$ also diverges to infinity, we further require that $H\log(nH)=o(T/\log(T))$. Otherwise, when $T$ is fixed and $n$ increases to infinity, the pairwise fusion penalized estimator may not be able to identify the group structure or consistently estimate the smooth function.
\end{remark}

\begin{remark}\label{rmk:4}
    Theorem \ref{thm:2} implies that, with a high probability, the oracle estimator is a local minimizer, but not necessarily the global minimizer, of objective function \eqref{eq:objective}. In addition, since the objective function is non-convex, the estimator $\widehat{\bm{\theta}}_{i,\tau}(\lambda)$ will heavily depend on algorithms, as well as parameter initialization. An efficient iterative algorithm is introduced in Section \ref{sec:algo}, and a warm-start initialization procedure is also provided.
\end{remark}

We next investigate asymptotic properties of the SIC in \eqref{eq:SIC}. For technical reasons, we only consider $\lambda$ such that the number of estimated groups is not larger than $K_{\max}$, a pre-specified upper bound which is fixed. Since the partition will be unchanged if one just relabels by a permutation of $\{1,\dots,K\}$, without loss of generality, we assume that $\min{\{i:i\in G_k\}}>\min\{i:i\in G_{k+1}\}$, $k=1,\dots,K-1$. Under the additional assumption on the conditional densities within each subgroup, we have the following subgroup selection consistency of the SIC.

\begin{theorem}\label{thm:3}
	Suppose that $f_i(0|x_{it})=f_{(k)}(0|x_{it})$ for all $i\in G_k$ and $k=1,\dots,K$.
    Let $\widehat{K}$ be the number of estimated groups and $\widehat{\mathcal{G}}$ be the corresponding partition with the penalty parameter $\lambda$ selected by \textup{SIC}. If the conditions in Theorem \ref{thm:2} hold, then
	\begin{equation}
		\mathbb{P}(\widehat{K}=K,\widehat{\mathcal{G}}=\mathcal{G})\to 1 \hspace{3mm}\text{as}\hspace{3mm} T\to\infty.
	\end{equation}
\end{theorem}

This theorem shows that, with probability approaching one, both the number of subgroups and group memberships can be estimated correctly by the proposed penalized estimation and SIC tuning parameter selection procedure.
The additional homogeneous condional density assumption in this theorem is not too stringent, since we only require that all individual's conditional density are identical at $0$ within each subgroup, rather than assuming the conditional density functions $f_i(\cdot|x_{it})$ are homogenous almost everywhere.

\begin{remark}\label{rmk:4}
	Consider multiple quantile levels of $0<\tau_1<\dots<\tau_Q<1$, and it is of interest to test whether the latent group structures at the $Q$ levels are the same. First, under the null hypothesis of the same group structures, we can extend the pairwise fusion penalized estimation at \eqref{eq:objective} to the case with multiple quantiles,
	\begin{equation}
	\min\frac{1}{nT}\sum_{q=1}^Q\sum_{i=1}^n\sum_{t=1}^T\rho_{\tau_q}(y_{it}-\mu_{i,\tau_q}-\bm{\Pi}(x_{it})^\top\bm{\theta}_{i,\tau_q})+\binom{n}{2}^{-1}\sum_{i<j}p_\lambda(\|\bm{\theta}_i-\bm{\theta}_j\|_2),
	\end{equation}
	where $\bm{\theta}_i=(\bm{\theta}_{i,\tau_1}^\top,\dots,\bm{\theta}_{i,\tau_Q}^\top)^\top$ is the stacked parameter vector, and its asymptotic properties can be established similarly. On the other hand, under the alternative hypothesis, we can conduct the penalized estimation at \eqref{eq:objective} to the $Q$ quantiles separately. As a result, the test statistic can be set to the difference between the loss functions under null and alternative hypotheses. However, this actually is a high-dimensional testing problem and, not like the case with fixed dimensions, the null distribution usually has a complicated form. We leave it for future research.
\end{remark}

\section{ADMM Algorithm} \label{sec:algo}

It is nontrivial to search for the proposed estimator, $\{\widehat{\mu}_{1,\tau}(\lambda),\cdots,\widehat{\mu}_{n,\tau}(\lambda),\widehat{\bm{\theta}}_{1,\tau}(\lambda),\cdots,\widehat{\bm{\theta}}_{n,\tau}(\lambda)\}$, since there are $O(n^2)$ non-convex penalty terms in the objective function \eqref{eq:objective} and the quantile loss $\rho_{\tau}(\cdot)$ is also not differentiable.
This section overcomes the difficulty by introducing an alternating direction method of multipliers (ADMM) algorithm \citep{boyd2010distributed}.

We first introduce some notations. Denote the matrix form of the panel data by $\bm{y}=(y_{11},\dots,y_{1T},y_{21},\dots,y_{2T},\dots,y_{nT})^\top\in\mathbb{R}^{nT}$,  $\bm{\Pi}_i=(\bm{\Pi}(x_{i1}),\bm{\Pi}(x_{i2}),\dots,\bm{\Pi}(x_{iT}))^\top\in\mathbb{R}^{T\times (H-1)}$, $\widetilde{\bm{\Pi}}_i=(\bm{1}_T,\bm{\Pi}_i)\in\mathbb{R}^{T\times H}$, and
$\bm{\Pi}=\text{diag}( \widetilde{\bm{\Pi}}_1,\widetilde{\bm{\Pi}}_2,\dots,\widetilde{\bm{\Pi}}_n)\in\mathbb{R}^{nT\times nH}$. For the pairwise difference of coefficient vectors, define the pairwise difference indicator matrix by  $\bm{\Delta}=\{(\bm{e}_i-\bm{e}_j),~i<j\}^\top\in\mathbb{R}^{\binom{n}{2}\times n}$, where $\bm{e}_i$ is the coordinate vector with the $i$-th element equal to one and all the other elements equal to zero. Let $\bm{w}_\tau=(\mu_{1,\tau},\bm{\theta}_{1,\tau}^\top,\mu_{2,\tau},\bm{\theta}_{2,\tau}^\top,\dots,\mu_{1,\tau},\bm{\theta}_{n,\tau}^\top)^\top\in\mathbb{R}^{nH}$ be the vectorized coefficients, and $\bm{A}=\bm{\Delta}\otimes\bm{I}_H\in\mathbb{R}^{\binom{n}{2}H\times nH}$, where $\otimes$ refers to the Kronecker product of matrices.

The objective function of the pairwise fusion penalized estimator in \eqref{eq:objective} can be rewritten into
\begin{equation}\label{eq:obj2}
	\frac{n-1}{2T}\rho_\tau(\bm{y}-\bm{\Pi} \bm{w}_\tau)+\sum_{i=1}^{\binom{n}{2}}p_\lambda(\|(\bm{A}\bm{w}_\tau)^{(i)}\|_2),
\end{equation}
where $\rho_\tau(\bm{u})=\sum_{j=1}^{nT}\rho_\tau(u_j)$ for any vector $\bm{u}=(u_1,\dots,u_{nT})^\top\in\mathbb{R}^{nT}$,  $(\bm{v})^{(i)}$ refers to the sub-vector extracted from the $[(i-1)H+2]$-th to $(iH)$-th elements of $\bm{v}\in\mathbb{R}^{\binom{n}{2}H}$ and $1\leq i\leq \binom{n}{2}$. Note that the sub-vector extraction will skip the term $\mu_{i,\tau}$ and enable us to calculate the pairwise difference between $\bm{\theta}_{i,\tau}$ and $\bm{\theta}_{j,\tau}$.

By introducing the dummy variable $\bm{v}_\tau=\bm{A}\bm{w}_\tau\in\mathbb{R}^{\binom{n}{2}H}$, we can rewrite the objective function \eqref{eq:obj2} into the constrained form
\begin{equation}
	\frac{n-1}{2T}\rho_\tau(\bm{y}-\bm{\Pi} \bm{w}_\tau)+\sum_{i=1}^{\binom{n}{2}}p_\lambda(\|\bm{v}_\tau^{(i)}\|_2),~~\text{subject to }\bm{A}\bm{w}_\tau-\bm{v}_\tau=0,
\end{equation}
and further the augmented Lagrangian form
\begin{equation}
	\mathcal{L}_\gamma(\bm{w}_\tau,\bm{v}_\tau;\bm{u}_\tau)=\frac{n-1}{2T}\rho_\tau(\bm{y}-\bm{\Pi}\bm{w}_\tau)+\sum_{i=1}^{\binom{n}{2}}p_\lambda(\|\bm{v}_\tau^{(i)}\|)+\frac{\gamma}{2}\left\|\bm{A}\bm{w}_\tau-\bm{v}_\tau+\frac{\bm{u}_\tau}{\gamma}\right\|_2^2,
\end{equation}
respectively, where $\bm{u}_\tau$ is the Lagrangian multiplier and $\gamma$ is the penalty parameter. Given the fixed tuning parameter $\lambda$, the augmented Lagrangian form can be solved by the ADMM algorithm, as summarized in Algorithm \ref{alg:admm}.\\

In the outer loop of Algorithm \ref{alg:admm}, the $\bm{w}_\tau$-update takes the form of
\begin{equation}
	\frac{n-1}{2T}\rho_\tau(\bm{y}-\bm{\Pi} \bm{w}_\tau)+\frac{\gamma}{2}\|\bm{A}\bm{w}_\tau-\bm{v}_\tau+\bm{u}_\tau/\gamma\|_2^2,
\end{equation}
and does not have an explicit solution. Consider its augmented Lagrangian form
\begin{equation}
	\mathcal{L}_\kappa(\bm{w}_\tau,\bm{r}_\tau;\bm{h}_\tau)=\frac{n-1}{2T}\rho_\tau(\bm{r}_\tau)+\frac{\gamma}{2}\|\bm{A}\bm{w}_\tau-\bm{v}_\tau+\bm{u}_\tau/\gamma\|_2^2+\frac{\kappa}{2}\|\bm{r}_\tau+\bm{\Pi}\bm{w}_\tau-\bm{y}+\bm{h}_\tau/\kappa\|_2^2,
\end{equation}
where $\bm{r}_\tau$ is the dummy variable for $\bm{y}-\bm{\Pi}\bm{w}_\tau$, $\bm{h}_\tau$ is the Lagrangian multiplier, and $\kappa$ is the penalty parameter. We then apply another ADMM algorithm to solve it, and this results in the inner loop in Algorithm \ref{alg:admm}. Note that the $\bm{r}_\tau$-update can be solved by the asymmetric soft-thresholding operator, while the $\bm{w}_\tau$-update is a least squares problem and has a closed-form solution. Even though the dimensions of matrices $\bm{A}\in\mathbb{R}^{\binom{n}{2}H\times nH}$ and $\bm{\Pi}\in\mathbb{R}^{nT\times nH}$ are very large, the design matrix for the least squares problem in the $\bm{w}_\tau$-update remains the same for all iterations. Hence, we can apply the QR decomposition to the design matrix in the least squares problem only once, and it is not necessary to calculate the inverse of this large matrix repeatedly. 

\begin{algorithm}[H]\label{alg:admm}
	\SetAlgoLined
	\KwData{$\bm{y}$, $\bm{\Pi}$}
	Initialize $\bm{w}_\tau$, $\bm{v}_\tau$, $\bm{u}_\tau$; $\bm{h}_\tau \leftarrow$ $\bm{w}^0_\tau$, $\bm{v}^0_\tau$, $\bm{u}^0_\tau$; $\bm{h}^0_\tau$
	
	\While{k $<$ number of outer loop replications}{
		$\bm{w}_\tau^{k+1}$ update:\\
		\While{j $<$ number of inner loop replications}{
			$\bm{r}_\tau^{j+1} := \underset{\bm{r}}{\argmin}~ (n-1)/(2T)\cdot\rho_{\tau}(\bm{r}_\tau) + \kappa/2\cdot\|\bm{r}_\tau + \bm{\Pi} \bm{w}_\tau^j - \bm{y} + \bm{h}_\tau^k/\kappa\|_2^2$  \\
			$\bm{w}_\tau^{j+1} := \underset{\bm{w}_\tau}{\argmin}~ \gamma/2\cdot \|\bm{A}\bm{w}_\tau - \bm{v}_\tau^k+\bm{u}_\tau^k/\gamma\|_2^2 + \kappa/2\cdot\|\bm{\Pi}\bm{w}_\tau + \bm{r}_\tau^{j+1} - \bm{y} + \bm{h}_\tau^k/\kappa\|_2^2$  \\
			$\bm{h}_\tau^{j+1} := \bm{h}_\tau^j + \kappa(\bm{r}_\tau^{j+1} + \bm{\Pi} \bm{w}_\tau^{j+1} - \bm{y})$
		} 
		$\bm{v}_\tau^{(i)k+1} := \underset{\bm{v}_\tau^{(i)}}{\argmin} P_{\lambda}(\|\bm{v}_\tau^{(i)}\|_2) + \gamma/2\cdot\|(\bm{A}\bm{w}_\tau^{k+1})^{(i)} + \bm{u}_\tau^{(i)k}/\gamma - \bm{v}_\tau^{(i)}\|_2^2$  \\
		$\bm{u}_\tau^{k+1} := \bm{u}_\tau^k + \gamma(\bm{A}\bm{w}_\tau^{k+1} - \bm{v}_\tau^{k+1})$
	}
	\KwResult{optimized $\bm{w}_\tau$}
	\caption{ADMM algorithm for the fixed $\lambda$}
\end{algorithm}~

As the penalty function $p_\lambda(\cdot)$ is concave on $[0,\infty)$, we apply the majorization method in $\bm{v}_\tau^{(i)}$-update as in \citet{peng2015iterative}. Here $\bm{v}_\tau^{(i)}$ step has a closed-form solution:
\begin{equation}
	\bm{v}_\tau^{(i)k+1} = R_{p_\lambda'(\|\bm{v}_\tau^{(i)}\|_2+)/\gamma}\left((\bm{A}\bm{w}_\tau^{k+1})^{(i)}+\bm{u}_\tau^{k(i)}/\gamma\right),
\end{equation}
where $R_t(\bm{x})=\bm{x}(1-t/\|\bm{x}\|_2)_+$ and $p_\lambda'(a+)$ denotes the limit of the derivative $p_\lambda'(x)$ when $x\to a$ from the above. Note that, with $p_\lambda(\cdot)$ being the SCAD penalty,
\begin{equation}
	p_\lambda'(\|\bm{v}^{(i)}_\tau\|_2+)=
	\begin{cases}
		\lambda, & \text{when }0\leq\|\bm{v}_\tau^{(i)}\|_2<\lambda;\\
		(a\lambda-\|\bm{v}_\tau^{(i)}\|_2)/(a-1), & \text{when }\lambda\leq\|\bm{v}_\tau^{(i)}\|_2< a\lambda;\\
		0, & \text{when }\|\bm{v}_\tau^{(i)}\|_2\geq a\lambda.\\
	\end{cases}
\end{equation}
In Algorithm \ref{alg:admm}, all subproblems have closed-form solutions and thus can be solved efficiently.

For non-convex objective functions, it is well known that the ADMM algorithm may not converge to the global optimal solution and the results are highly sensitive to the initial values. We suggest a fine grid search with a warm-start initialization procedure to solve the problem. Specifically, we first consider a sequence of increasing tuning parameters $0=\lambda_0<\lambda_1<\dots<\lambda_M$ on a fine grid, and set the initial values of $\bm{w}_\tau,\bm{v}_\tau,\bm{u}_\tau$ and $\bm{h}_\tau$ to $\bm{0}$ for $\lambda_0=0$. For the following tuning parameters $\lambda_m$, $1\leq m\leq M$, we initialize the iterative algorithm by the solution previously obtained with respect to $\lambda_{m-1}$.

\begin{remark}
	It is of interest to conduct convergence analysis for the proposed ADMM algorithm, which involves both a non-smooth quantile loss function and a non-convex penalty function. The convergence analysis has been investigated for ADMM algorithms to non-convex regularized least squares problems \citep{ma2017concave,zhu2018cluster}, however, it cannot be applied here since the proving techniques heavily depend on the smoothness of loss functions. In the meanwhile, to solve non-convex penalized quantile regression, \citet{peng2015iterative} proposed an iterative coordinate descent algorithm, and its convergence analysis was also studied; see also \citet{yu2017admm}. It actually is a single-layer ADMM, and it is still an open problem in the literature to theoretically justify a nested two-layer ADMM as in Algorithm \ref{alg:admm}. We leave it for future research.
\end{remark}

\section{Simulation Experiments} \label{sec:sim}

\subsection{Subgroup analysis at a single quantile level}

This subsection conducts three simulation experiments to evaluate the finite-sample performance of the proposed SCAD pairwise fusion penalized estimator, $\widehat{m}_{i,\tau}(x)=\bm{\Pi}(x)^\top\widehat{\bm{\theta}}_{i,\tau}(\lambda)$, at a single quantile level $\tau$ for the cases with independent and identically distributed ($i.i.d.$) covariates, weakly dependent covariates and heavy-tailed data, respectively.
The oracle estimator with the known group structure, $\widetilde{m}_{i,\tau}(x)=\bm{\Pi}(x)^\top\widehat{\bm{\theta}}_{i,\tau}^\textup{o}$, is set to be the benchmark for comparison. 

For each experiment, individuals belong to three subgroups of equal size, and the data generating process is
\begin{equation}
	y_{it} = \mu_{i,\tau} + c_i\cdot\sin(2\pi x_{it}) + e_{it}(\tau),
\end{equation}
where we set $c_i=0.2$, 1 or 2 for each subgroup, the error term $e_{it}(\tau)$ is independent of $x_{it}$, and $\{\mu_{i,\tau},1\leq i\leq n\}$ are generated from independent standard normal distributions and kept unchanged for all replications. 
Note that $\int_0^1\sin(2\pi x)dx=0$, and $y_{it}$ has the conditional quantile function of $Q_\tau(y_{it}|x_{it})=\mu_{i,\tau}+c\cdot\sin(2\pi x_{it})+F^{-1}(\tau)$, where $F(\tau)$ is the distribution function of $e_{it}(\tau)$.
Thus, the smooth function of interest at model \eqref{eq:ConditionalQuantile} is $m_{i,\tau}(x)=c_i\sin(2\pi x)$.
The number of individuals is set to $n=60$ and 120, and that of time points is $T=100$ and 1000. There are 1000 replications for each combination of $n$ and $T$. 

The ADMM algorithm in Section \ref{sec:algo} is employed to search for the estimators, and the tuning parameter $\lambda$ is selected by the SIC in Section \ref{sec:SIC}. The estimation performance of both estimators is evaluated by mean squared errors (MSEs), which are defined as the average of $(nT)^{-1}\sum_{i=1}^n\sum_{t=1}^T[{g}_{i,\tau}(x_{it})-m_{i,\tau}(x_{it})]^2$ over 1000 replications with ${g}_{i,\tau}(\cdot)=\widehat{m}_{i,\tau}(\cdot)$ or $\widetilde{m}_{i,\tau}(\cdot)$, and the subgroup selection performance is measured by the percentage of correct recovery of the number of groups. 
Moreover, the asymptotic normality in \eqref{eq:asymp_norm}, together with Theorem \ref{thm:2}, makes it possible to use $\widehat{m}_{i,\tau}(x)$ to construct the pointwise confidence intervals for ${m}_{i,\tau}(x)$. We set the confidence level to 95\%, and the re-weighted Nadaraya-Watson method is used to estimate conditional densities in the asymptotic variance.

Experiment 1 is to check the estimation performance with $i.i.d.$ covariates. Specifically, the covariate and error terms $\{x_{it},e_{it}(\tau)\}$ are set to be $i.i.d.$ across both $i$ and $t$, where $x_{it}$ follows the standard uniform distribution, and $e_{it}(\tau)$ follows a normal distribution with the standard deviation of 0.1 and mean satisfying that $\mathbb{P}(e_{it}(\tau)\leq 0)=\tau$. 
Following \citet{yao2000nonparametric}, for subgroup $k$, we define its signal-to-noise ratio as $\text{SNR}_{(k)}=\mathbb{E}[m_{(k),\tau}^2(x_{it})]/\mathbb{E}[\sigma_{k}^2(x_{it})]$, where $\sigma_k(x_{it})$ is the conditional standard deviation of $e_{it}(\tau)$ given $x_{it}$. In this experiment, the signal-to-noise ratios for the three subgroups are 10, 50, and 100, respectively.
We consider two quantile levels of $\tau=0.5$ and $0.7$. 
Table \ref{table:ex1} gives the percentages of correct group number determination and MSEs of $\widehat{m}_{i,\tau}(\cdot)$ and $\widetilde{m}_{i,\tau}(\cdot)$, and Figure \ref{fig:cov_prob_iid_0pt5_0pt7} presents the empirical coverage probabilities for pointwise confidence intervals of ${m}_{i,\tau}(x)$.
It can be seen that the MSEs of both estimators are close to each other, and they decrease as $n$ and/or $T$ increases.
Moreover, the MSEs at $\tau=0.7$ are slightly larger than those at $\tau=0.5$, and the SIC can correctly select the subgroups for almost all cases. 
Finally, the pointwise confidence intervals can provide a reliable coverage even when $(n,T)$ is as small as $(60,100)$.

Experiment 2 is designed for the case with weak dependence in the sequence of covariates. Specifically, a sequence of random variables are first generated by an autoregressive model, $\widetilde{x}_{it}=0.5\widetilde{x}_{i,t-1}+u_{it}$, where $\{u_{it}\}$ are $i.i.d.$ standard normal random variables, and they are then transformed into covariates with the range of $[0,1]$ by $x_{it}=F_{N(0,4/3)}(\widetilde{x}_{it})$, where $F_{N(0,4/3)}(\cdot)$ is the normal distribution function with mean zero and variance $4/3$. 
The error terms $\{e_{it}(\tau)\}$ are generated as in the first experiment, and the quantile level is fixed at $\tau=0.5$. The signal-to-noise ratios are the same as those in the first experiment. Table \ref{table:ex2} presents the percentage of correct subgroup selection and MSEs of two estimators, and the left panel of Figure \ref{fig:cov_prob_depco_htail} plots the empirical coverage probabilities of pointwise confidence intervals for ${m}_{i,\tau}(x)$. The results are similar to those in the first experiment, and it suggests that, when $x_{it}$'s  are dependent and $\alpha$-mixing across $t$, the pairwise fusion penalized estimator has performance as good as the case with $\textit{i.i.d.}$ covariates. In other words, the proposed methodology is robust to the weak dependence in covariates.

Experiment 3 is for the case with heavy tails, and the errors are generated by $e_{it}(\tau)=0.1\epsilon_{it}(\tau)$, where $\{\epsilon_{it}(\tau)\}$ are independent across $i$ and $t$ and follow the Student's $t$ distribution with five degrees of freedom ($t_5$ distribution). 
The signal-to-noise ratios of three subgroups are 6, 30 and 60, respectively, and 
all other settings are the same as in the first experiment. The MSEs of both estimators and percentage of correct group number determination are listed in Table \ref{table:ex2}, and empirical coverage probabilities are presented in the right panel of Figure \ref{fig:cov_prob_depco_htail}. Note that the $t_5$ distribution has much heavier tails than those of the normal distribution. While, in general, a similar performance to that in the first experiment can be observed, although the response is subject to the heavy-tailed noise contamination. Moreover, the percentage of correct group number determination can still reach 100\%, and coverage rates of pointwise confidence intervals are close to the target level 95\%. We may argue that the proposed methodology is robust to the heavy-tailed contamination.

\subsection{Subgroup analysis at multiple quantile levels}

This subsection first introduces a data generating process for Experiment 4 with varying subgroup structures at different quantile levels, and the proposed methodology is then evaluated in terms of finite-sample performance. 

Consider a panel data model with $n=60$ individuals,
\begin{equation}
	y_{it}=\mu_i+\sin(2\pi x_{it})+\sigma_i(x_{it},I\{\varepsilon_{it}<0\})\varepsilon_{it},
\end{equation}
where $\varepsilon_{it}$ is independent of $x_{it}$, $\{x_{it},\varepsilon_{it}\}$ are $i.i.d.$ across both $i$ and $t$, $x_{it}$ and $\varepsilon_{it}$ follow the standard uniform and standard normal distributions, respectively, and $\{\mu_{i},1\leq i\leq n\}$ are generated from independent standard normal distributions and keep unchanged for all replications.
The scale function $\sigma_i(\cdot,\cdot)$ has the form of
\begin{equation}
	\sigma_i(x,1)=\begin{cases}
		\sigma_{(1),\text{L}}(x)=0.4+0.8x, & i\in G_{1,\text{L}}\\
		\sigma_{(2),\text{L}}(x)=1.2-0.8x, & i\in G_{2,\text{L}}
	\end{cases}
\end{equation}
and
\begin{equation}
	\sigma_i(x,0)=\begin{cases}
		\sigma_{(1),\text{U}}(x)=0.4+0.8x, & i\in G_{1,\text{U}}\\
		\sigma_{(2),\text{U}}(x)=1.2-0.8x, & i\in G_{2,\text{U}}\\
		\sigma_{(3),\text{U}}(x)=0.4, & i\in G_{3,\text{U}}
	\end{cases},
\end{equation}
where $\{G_{i,\text{L}}:i=1,2\}$ and $\{G_{i,\text{U}}:i=1,2,3\}$ are different partitions of all individuals. Specifically, we set $G_{1,\textup{L}}=\{1,\dots,30\}$, $G_{2,\textup{L}}=\{31,\dots,60\}$, $G_{1,\textup{U}}=\{1,\dots,20\}$, $G_{2,\textup{U}}=\{21,\dots,40\}$, and $G_{3,\textup{U}}=\{41,\dots,60\}$, and the signal-to-noise ratios are 0.721 for the first 40 individuals and 1.172 for the last 20 ones, respectively.

One can easily check that the conditional median functions are homogeneous, up to an individual effect; that is $Q_{0.5}(y_{it}|x_{it})=\mu_{i}+\sin(2\pi x_{it})$, for all $1\leq i\leq n$. For any lower quantile $\tau<0.5$, the conditional quantile has a subgroup structure,
\begin{equation}
	Q_\tau(y_{it}|x_{it})=\begin{cases}
		\mu_i+\sin(2\pi x_{it})+(0.4+0.8x_{it})\Phi^{-1}(\tau), & i\in G_{1,\text{L}}\\
		\mu_i+\sin(2\pi x_{it})+(1.2-0.8x_{it})\Phi^{-1}(\tau), & i\in G_{2,\text{L}}\\
	\end{cases},
\end{equation}
and, for any upper quantile $\tau>0.5$, the conditional quantile has another subgroup structure,
\begin{equation}
	Q_\tau(y_{it}|x_{it}) = \begin{cases}
		\mu_i + \sin(2\pi x_{it}) + (0.4+0.8x_{it})\Phi^{-1}(\tau), & i \in G_{1,\text{U}} \\
		\mu_i + \sin(2\pi x_{it}) + (1.2-0.8x_{it})\Phi^{-1}(\tau), & i \in G_{2,\text{U}}\\
		\mu_i + \sin(2\pi x_{it}) + 0.4\Phi^{-1}(\tau), & i \in G_{3,\text{U}}
	\end{cases},
\end{equation}
where $\Phi(\cdot)$ is the standard normal distribution function.
Note that $\int_0^1\sin(2\pi x)dx=0$ and $\int_0^1xdx=0.5$, and we then can define the function of $m_{i,\tau}(\cdot)$ at model \eqref{eq:ConditionalQuantile}.
Compared with those in the previous subsection, the above data generating process has three features: (i) the model is heteroskedastic, (ii) the subgroup structure varies over $\tau$, and (iii) the difference between subgroups is relatively small. 
More details about the model can be found in Appendix C of the supplementary file.

The proposed SCAD pairwise fusion penalized estimation, as well as the oracle estimation, is conducted at three quantile levels of $\tau=0.1$, 0.5 and 0.9, at which the individuals belong to two, one and three subgroups, respectively. We employ the SIC to select the tuning parameter $\lambda$, and hence the group structure.
The number of time points is set to $T=100$, and there are 1000 replications.
Table \ref{table:ex4} gives the percentages of correct group number determination and MSEs of $\widehat{m}_{i,\tau}(\cdot)$ and $\widetilde{m}_{i,\tau}(\cdot)$, and Figure \ref{fig:cov_prob_locsc} presents the empirical coverage probabilities of pointwise confidence intervals for ${m}_{i,\tau}(x)$.

From Table \ref{table:ex4}, for the case of $\tau=0.5$, the percentage of correct structure recovery is 99.6\%, and MSEs of the oracle and penalized estimators are almost the same.
It is due to the fact that all individuals have a homogeneous structure at the median.
At the lower quantile level of $\tau=0.1$, the correct subgroup selection rate is 97\%. However, when estimating the nonparametric functions, all data points are essentially split into two groups, and it then leads to a worse performance of the penalized estimator.
On the other hand, at the upper level of $\tau=0.9$, there are three subgroups, and the estimation performance  is even worse. The percentage of correct subgroup selection also drops to 88.4\%. 
Moreover, compared with those in the previous subsection, the coverage rates of the confidence intervals in Figure \ref{fig:cov_prob_locsc} are relatively less accurate. It is due to the fact that the data generating process is heteroskedastic, and this makes the estimation of conditional densities in the asymptotic variance more challenging.  
In general, due to the relatively small difference between subgroups compared with the data generating processes in the previous subsection, both the subgroup selection and estimation have a worse performance, and they can be significantly improved with a larger value of $n$ or $T$.

In sum, we may conclude that the proposed method can be used to conduct homogeneity pursuit at different quantile levels, and it is robust to the weak dependence, heavy tails, and heterogeneous errors.

\section{Real Data Analysis} \label{sec:real}

This section analyzes a climate dataset collected from different regions of the United Kingdom (UK). It can be downloaded from the website \url{https://www.metoffice.gov.uk/research/climate/maps-and-data/historic-station-data}. The climate data are collected monthly, and there are 37 weather stations in total.
We attempt to study how the temperature can be affected by the sunshine duration. The response $y_{it}$ is set to be the mean of daily maximum temperature, and the covariate $x_{it}$ is the total sunshine duration counted by hours.
We consider the data from January $1993$ to December $2009$,  and there are total $204$ time points.
After removing these stations with missing values, there are total $16$ stations left, i.e. the number of individuals is 16. 

We first consider the quantile level of $\tau=0.5$, and the proposed nonparametric quantile regression is applied to search for the subgroups. The tuning parameter $\lambda$ is selected by the SIC in Section \ref{sec:SIC}. All data are seasonally adjusted, and $x_{it}$ are standardized into the range $[0,1]$ before model fitting. For different values of the tuning parameter $\lambda$, Figure \ref{fig:k_sic_0pt5} plots the corresponding numbers of selected groups $\widehat{K}$, and the calculated values of SIC are also given. It is reasonable to choose $\widehat{K}=2$ groups, and the corresponding tuning parameter is $1.86$. Figure \ref{fig:subgroup_func} presents the estimated nonlinear functions, $\widehat{m}_{(k),0.5}$, with $k=1$ and $2$, and Figure \ref{fig:subgroup_map} plots the selected weather stations for these two groups.

It can be seen that the big Group 1 consists of 14 stations out of 16: Armagh, Bradford, Eastbourne, Paisley, Shawbury, Sheffield, Waddington, Cambridge, Eskdalemuir, Heathrow, Hurn, Leuchars, Oxford, and Ross-on-Wye. Therefore, up to the individual effects on the daily maximum temperature, most of the stations share exactly the same relationship between the total sunshine duration and daily maximum temperature. Two exceptions are the stations in Camborne and Lerwick, and they locate on extreme corners. 
With the individual effect presented in our model, the grouping structures are identified through the different shapes of the fitted nonparametric functions in Figure \ref{fig:subgroup_map}.
After checking the dataset in details, we found that, in Camborne and Lerwick, the sunshine duration reaches the maximum level from April to June, but the hottest months are July, August and September. In the meanwhile, for the other stations in Group 1, the peaks of both sunshine duration and daily maximum temperature are in July and August. This could explain the heterogeneity of $\widehat{m}_{(1),0.5}(x)$ and $\widehat{m}_{(2),0.5}(x)$ on the interval $x\in(0.5,1)$.

The proposed methodology is also applied to the quantile level of $\tau=0.9$, which corresponds to the scenario with extremely high temperatures. There are also two groups detected and the fitted nonparametric functions $\widehat{m}_{(k),0.9}(x)$ are presented in Figure \ref{fig:subgroup_func}. However, the subgroup memberships at $\tau=0.9$ are different from those at $\tau=0.5$. Specifically, Eastbourne station joins the small Group 2 with Camborne and Lerwick, and all the other group memberships are the same as those at $\tau=0.5$. 
From the map in Figure \ref{fig:subgroup_map}, Eastbourne station locates at the southeast edge of England and is along the English Channel. The unique location makes the corresponding maximum temperatures relatively lower, and hence Eastbourne joins Camborne and Lerwick at the high quantile level of $\tau=0.9$.
Finally, we consider the quantile level of $\tau=0.1$, and all individuals belong to one group in this case. Note that $y_{it}$ refers to the mean of daily maximum temperatures, and then its lower quantiles may be related to relatively mild temperatures, which can be affected by the sunshine duration for all stations in a same way.

\section{Conclusions and Discussions}

For panel data models with subgroup effects, this article considers a nonparameteric method to explore the relationship between response and predictors, which can avoid the wrong grouping results due to possible model misspecification.
In addition, quantile regression is also employed to detect different pursuit results at various quantile levels. More importantly, these two features are both supported by the real analysis on a climate dataset.
A concave fused penalty is used to select the groups and estimate models simultaneously, and the corresponding oracle properties are hence expected.
Moreover, the developed ADMM algorithm can be used to efficiently solve the pairwise fusion penalized minimization problem.

This article can be extended in three directions. First, the individual effects are considered in the conditional quantile functions in this article, and the time effects could also be included in the model. However, to consistently the individual effects and time effects, the numbers of individual $n$ and time points $T$ are required to diverge to infinity simultaneously.
Second, when there are many covariates, the nonparametric method will lead to a large number of parameters and, as in \citet{lian2019homogeneity}, we may consider some semiparametric approaches to substantially reduce the model complexity. 
Finally, this article considers the subgroup structure with a fixed number of groups, and the exact homogeneity is also assumed within each group. It may be interesting in theory to relax the exact homogeneity to approximate homogeneity with a diverging number of groups to flexibly characterize the heterogeneous nature in real applications.

%\section*{Supplementary Materials}
%The supplementary materials contain all technical proofs and additional simulation results.
%
%\section*{Acknowledgements}
%
%We are grateful to the co-editor, the associate editor and two anonymous referees for
%their valuable comments which led to substantial improvement of this article.
%Lian and Li are the co-corresponding authors.
%
%\section*{Funding}
%
%Lian's research is partially supported by the Hong Kong Research Grant Council (GRF grants 11300721 and 11311822). Li's research is partially supported by the Hong Kong Research Grant Council (GRF grants 17306519, 17305319 and 17306121) and the National Social Science Fund of China (grant No. 72033002).
%
%The authors report there are no competing interests to declare.

\bibliography{QuantileHomogeneity}

\newpage
\newcolumntype{L}{>{\centering\arraybackslash}m{3cm}}

\begin{table}[!htp]
	\caption{Mean squared errors (MSEs) of the oracle and SCAD-penalized estimators, and  percentages of correct subgroup recovery for the case with $i.i.d.$ covariates (Experiment 1). The estimators are evaluated at two quantile levels of $\tau=0.5$ and 0.7.}
	\centering
	\begin{tabular}{ll|ccccccc}
		\specialrule{0.07em}{0.5em}{0.1em} \hline
		\multirow{3}{*}{$n$}& \multirow{3}{*}{$T$} & \multicolumn{3}{c}{$\tau = 0.5$} && \multicolumn{3}{c}{$\tau = 0.7$} \\
		\cline{3-5} \cline{7-9}
		&&\multirow{2}{3cm}{\centering \% of correct group number} & \multicolumn{2}{c}{MSE $(\times 10^{-4})$}  && \multirow{2}{3cm}{\centering \% of correct group number} & \multicolumn{2}{c}{MSE $(\times 10^{-4})$} \\
		\cline{4-5}\cline{8-9}
		&&& Oracle & SCAD && & Oracle & SCAD \\
		\hline
		% \multirow{3}{*}{60}& 20 & 3.18 & 5.119 & 5.403 && 3.08 & 5.609 & 5.991\\
		% & 100 & 3.00 & 0.981 & 0.964 && 3.00 & 1.165 & 1.156 \\
		% & 1000 & 3.00 & 0.100 & 0.101 && 3.00 & 0.115 & 0.128 \\
		% \hline
		% \multirow{3}{*}{120}& 20 & 3.01 & 1.706 & 1.710 && 3.00 & 1.945 & 1.778\\
		% & 100 & 3.00 & 0.512 & 0.505 && 3.00 & 0.602 & 0.572 \\
		% & 1000 & 3.00 & 0.051 & 0.051 && 3.00 & 0.057 & 0.061 \\
		\multirow{2}{*}{60}& 100  & 99.0\%   & 1.005 & 1.242 && 99.0\%  & 1.086 & 1.261 \\
		& 1000 & 99.8\% & 0.097 & 0.108 && 100\% & 0.109 & 0.122 \\
		\hline
		\multirow{2}{*}{120}& 100  & 100\% & 0.490 & 0.464 && 100\% & 0.545 & 0.508 \\
		& 1000 & 100\% & 0.049 & 0.050 && 100\% & 0.055 & 0.059 \\
		\hline
	\end{tabular}
	\label{table:ex1}
\end{table}

\begin{table}[!htp]
	\caption{Mean squared errors (MSEs) of the oracle and SCAD-penalized estimators, and  percentages of correct subgroup recovery for the cases with weakly dependent covariates (Experiment 2) and heavy-tailed data (Experiment 3).}
	\centering
	\begin{tabular}{ll|ccccccc}
		\specialrule{0.07em}{0.5em}{0.1em} \hline
		\multirow{3}{*}{$n$}& \multirow{3}{*}{$T$} & \multicolumn{3}{c}{Weakly dependent covariates} && \multicolumn{3}{c}{Heavy-tailed data} \\%&& \multicolumn{3}{c}{Experiment IV}\\
		\cline{3-5} \cline{7-9} %\cline{11-13}
		&&\multirow{2}{3cm}{\centering \% of correct group number} & \multicolumn{2}{c}{MSE $(\times 10^{-4})$}  && \multirow{2}{3cm}{\centering \% of correct group number} & \multicolumn{2}{c}{MSE $(\times 10^{-4})$} \\%&& \multirow{2}{*}{$\widehat{K}_{pc}$} & \multicolumn{2}{c}{MSE $(\times 10^{-4})$}\\
		\cline{4-5}\cline{8-9}%\cline{12-13}
		&&& Oracle & SCAD && & Oracle & SCAD \\%&& & Oracle & SCAD \\
		\hline
		% \multirow{3}{*}{60}& 20 & 3.28 & 5.067 & 10.888 && 3.98 & 5.830 & 8.758 && 3.51 & 7.393 & 8.027 \\
		% & 100 & 3.00 & 1.036 & 1.021 && 3.00 & 1.171 & 1.159 && 3.00 & 1.541 & 1.494\\
		% & 1000 & 3.00 & 0.100 & 0.101 && 3.00 & 0.115 & 0.116 && 3.00 & 0.156 & 0.156\\
		% \hline
		% \multirow{3}{*}{120}& 20 & 3.02 & 1.695 & 1.999 && 3.11 & 1.926 & 2.026 && 3.11 & 2.288 & 2.241\\
		% & 100 & 3.00 & 0.502 & 0.490 && 3.00 & 0.555 & 0.541 && 3.00 & 6.938 & 6.815\\
		% & 1000 & 3.00 & 0.053 & 0.053 && 3.00 & 0.056 & 0.054 && 3.00 & 0.707 & 0.725\\
		\multirow{2}{*}{60}& 100  & 99.2\% & 1.010 & 1.214 && 99.6\% & 1.108 & 1.356 \\
		& 1000 & 100\%  & 0.097 & 0.106 && 100\%  & 0.109 & 0.120 \\
		\hline
		\multirow{2}{*}{120}& 100  & 100\% & 0.497 & 0.468 && 100\% & 0.551 & 0.523 \\
		& 1000 & 100\% & 0.049 & 0.051 && 100\% & 0.054 & 0.055 \\
		\hline
	\end{tabular}
	\label{table:ex2}
\end{table}

\begin{table}[!htp]
    \caption{Mean squared errors (MSEs) of the oracle and SCAD-penalized estimators, and  percentages of correct subgroup recovery for Experiment 4 with varying subgroup structures at different quantile levels. We consider three quantile levels of $\tau=0.1$, 0.5 and 0.9. }
    \centering
    \begin{tabular}{ll|ccccc}
        \specialrule{0.07em}{0.5em}{0.1em} \hline
        % \multirow{2}{*}{$(n,T)=(120, 1000)$} & \multirow{2}{*}{$\tau = 0.1$} && \multirow{2}{*}{$\tau = 0.5$} && \multirow{2}{*}{$\tau = 0.9$} \\
        \multicolumn{2}{c|}{$(n,T)=(60, 100)$} & $\tau = 0.1$ && $\tau = 0.5$ && $\tau = 0.9$ \\
        %\cline{1-2} %\cline{3-3} \cline{5-5} \cline{7-7}
        %\multirow{2}{*}{$(n,T)=(120, 1000)$}  &&&&& \\
        \hline
        \multicolumn{2}{c|}{\centering \% of correct group number} & $97.0\%$ && $99.6\%$ && $88.4\%$ \\
        \hline
        %\multirow{2}{*}{MSE $(\times 10^{-4})$} & Oracle & 0.256  && 0.039  && 0.289  \\
        \multirow{2}{*}{MSE $(\times 10^{-4})$} & Oracle & 4.512  && 0.849  && 6.217  \\
        % \cline{2-3}\cline{5-5} \cline{7-7}
        %& SCAD & 0.390 && 0.038 && 0.582 \\
        & SCAD & 5.188 && 0.841 && 8.602 \\
        \hline
    \end{tabular}
    \label{table:ex4}
\end{table}

\begin{figure}[h]
    \centering
    \begin{subfigure}[b]{0.525\textwidth}
        \includegraphics[width=\textwidth]{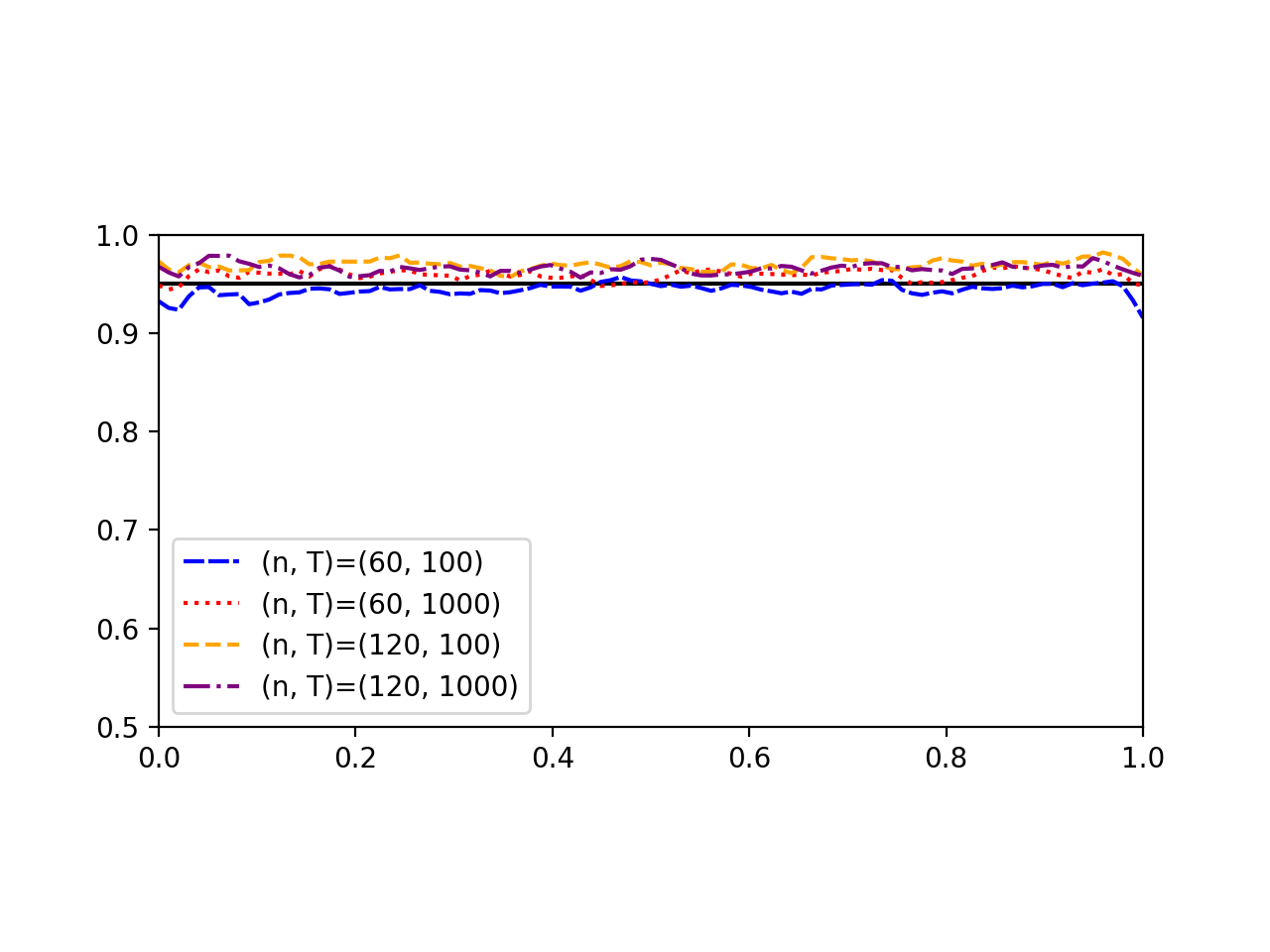}
        \vspace{-0.35in}
%        \caption{$\tau=0.5$}
    \end{subfigure}\hspace{-0.95cm}
    \begin{subfigure}[b]{0.525\textwidth}
        \includegraphics[width=\textwidth]{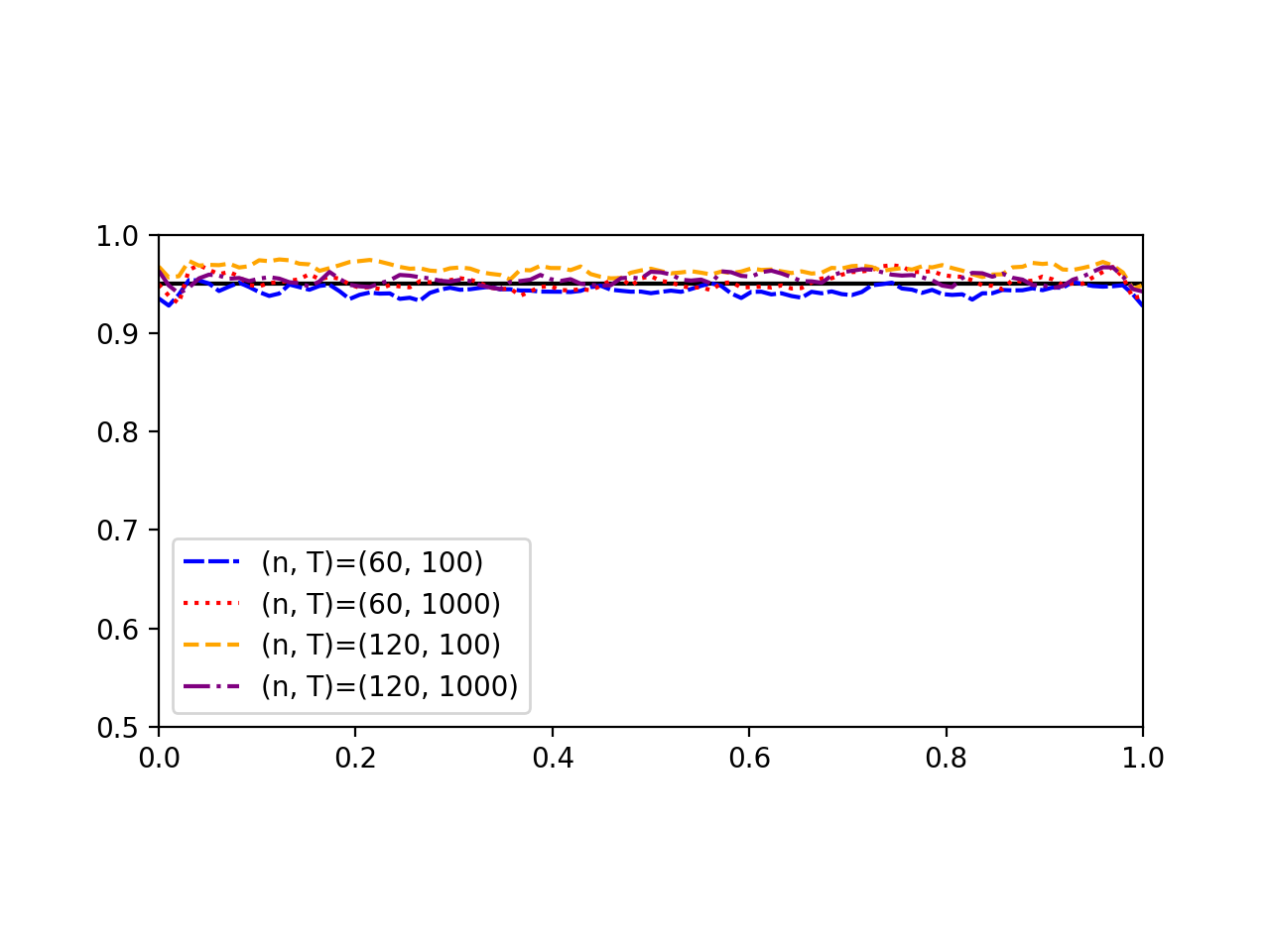}
        \vspace{-0.35in}
%        \caption{$\tau=0.7$}
    \end{subfigure}
    \caption{Empirical coverage probabilities for pointwise confidence intervals of ${m}_{i,\tau}(\cdot)$ in Experiment 1 at two quantile levels of $\tau=0.5$ (left panel) and 0.7 (right panel).}
    \label{fig:cov_prob_iid_0pt5_0pt7}
\end{figure}

\begin{figure}[h]
	\centering
	\begin{subfigure}[b]{0.525\textwidth}
		\includegraphics[width=\textwidth]{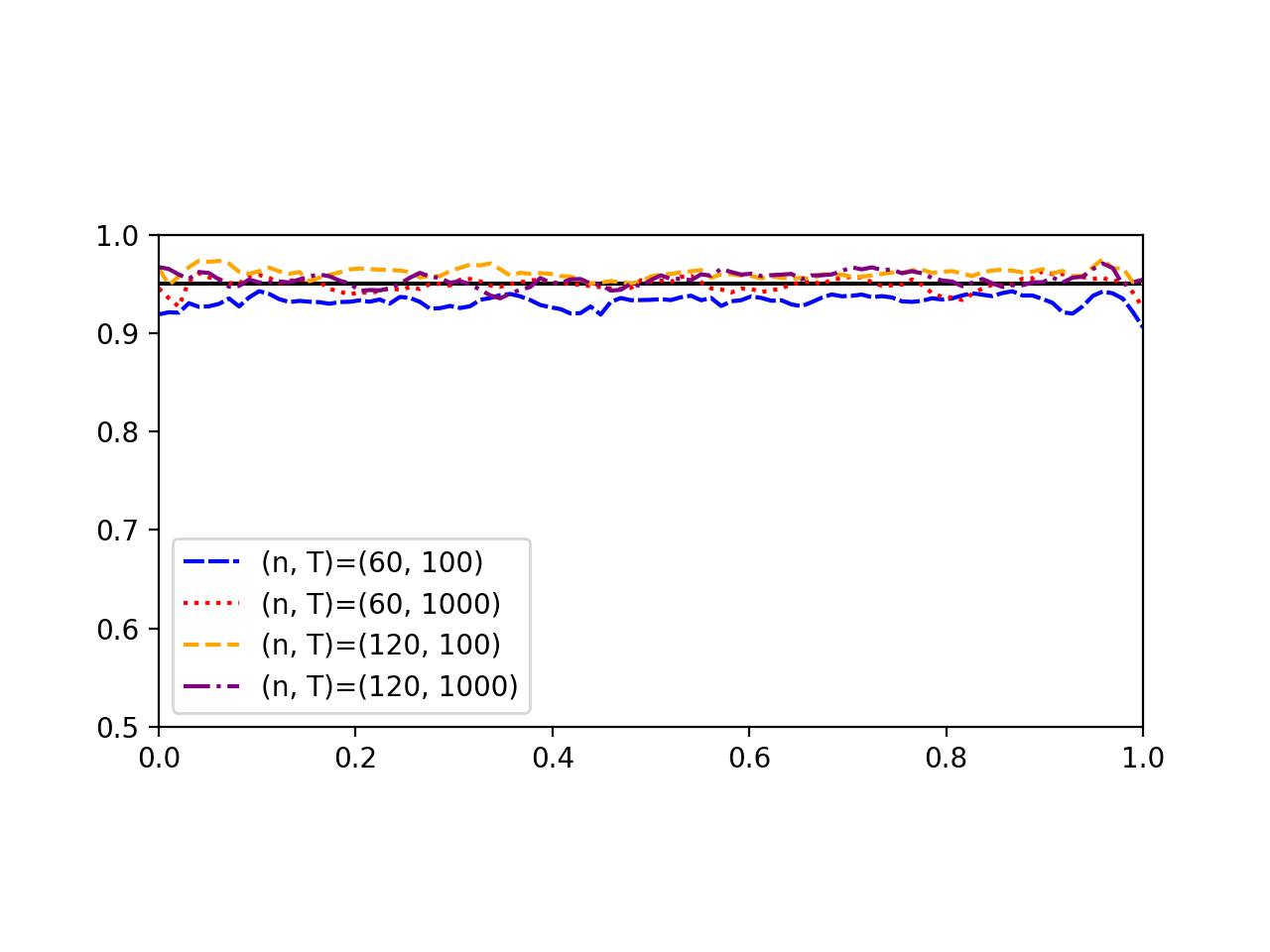}
		\vspace{-0.35in}
%		\caption{Experiment II (dependent covariate design)}
	\end{subfigure}\hspace{-0.95cm}
	\begin{subfigure}[b]{0.525\textwidth}
  	\includegraphics[width=\textwidth]{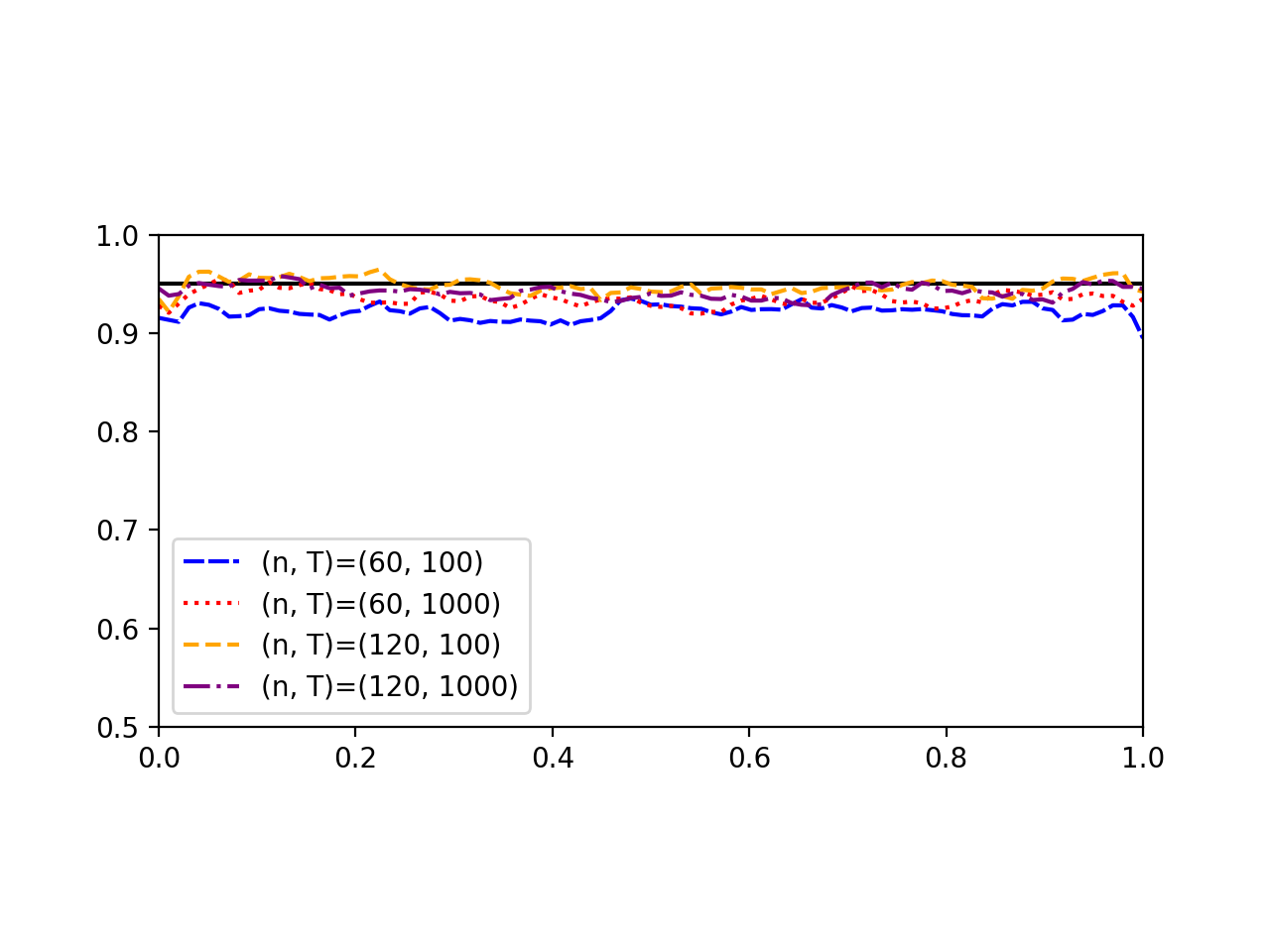}
		\vspace{-0.35in}
%		\caption{Experiment III (heavy-tailed design)}
	\end{subfigure}
	\caption{Empirical coverage probabilities for pointwise confidence intervals of ${m}_{i,0.5}(\cdot)$ in Experiment 2 (left panel) and Experiment 3 (right panel).}
	\label{fig:cov_prob_depco_htail}
\end{figure}

\begin{figure}[h]
	\begin{center}
	\includegraphics[width=0.6\textwidth]{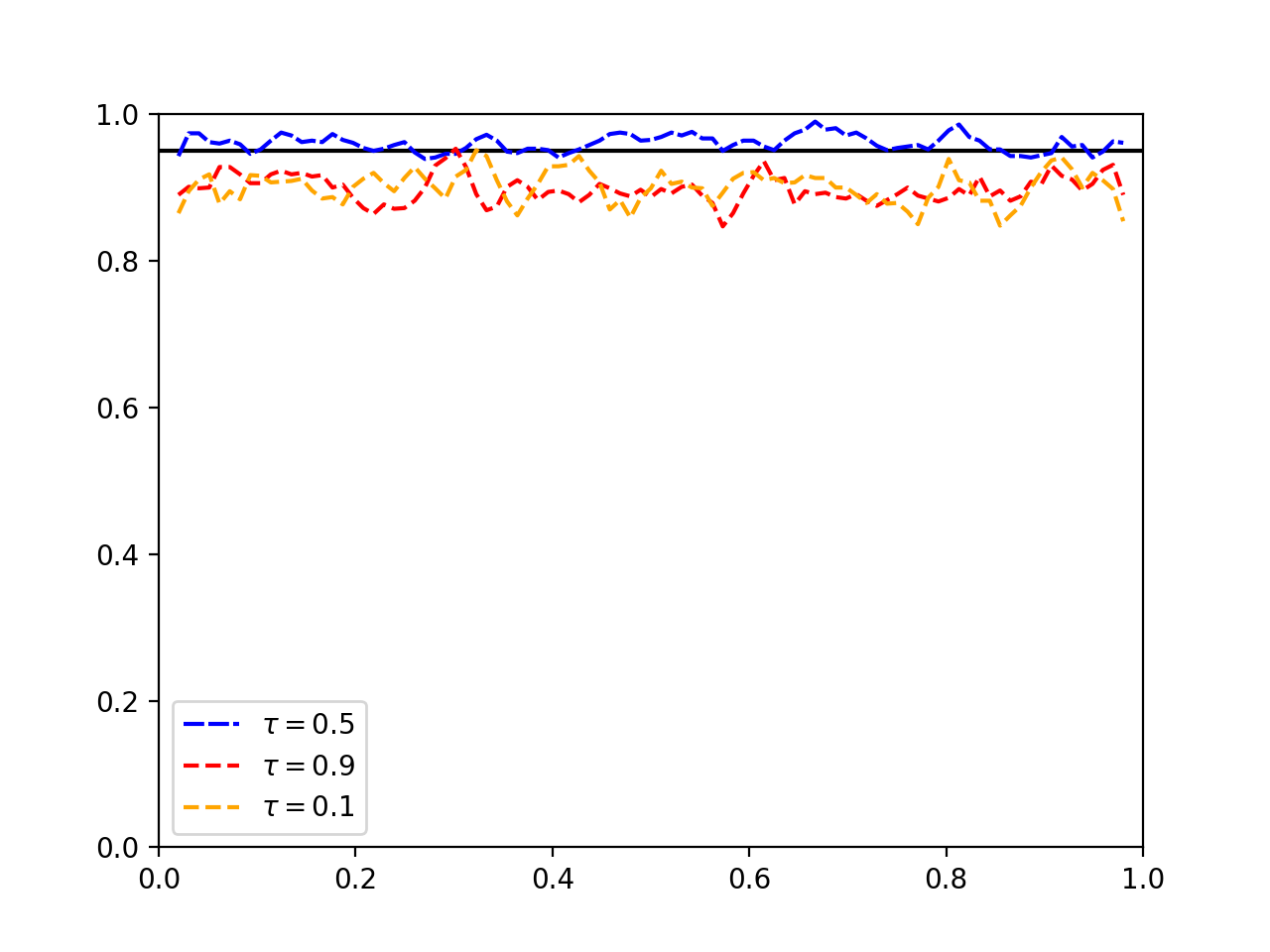}
	\caption{Empirical coverage probabilities for pointwise confidence intervals of ${m}_{i,\tau}(\cdot)$ in Experiment 4 with varying subgroup structures at different quantile levels. We consider $(n,T)=(60,100)$ and three quantile levels of $\tau=0.1$, 0.5 and 0.9.}
	\label{fig:cov_prob_locsc}
	\end{center}
\end{figure}

\clearpage

\begin{figure}[!htp]
    \centering
    \begin{subfigure}[b]{0.525\textwidth}
        \includegraphics[width=\textwidth]{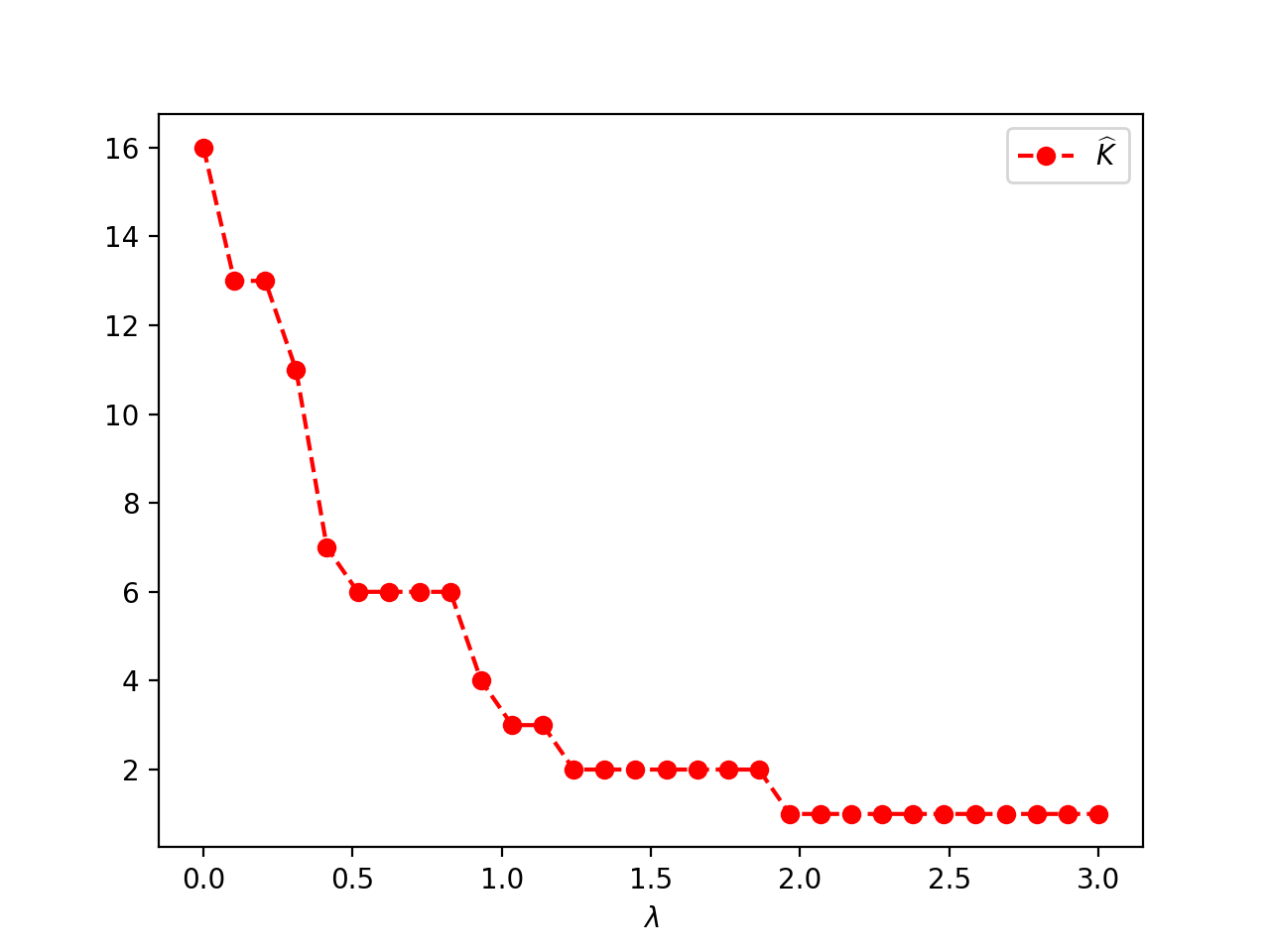}
    \end{subfigure}\hspace{-0.95cm}
    \begin{subfigure}[b]{0.525\textwidth}
        \includegraphics[width=\textwidth]{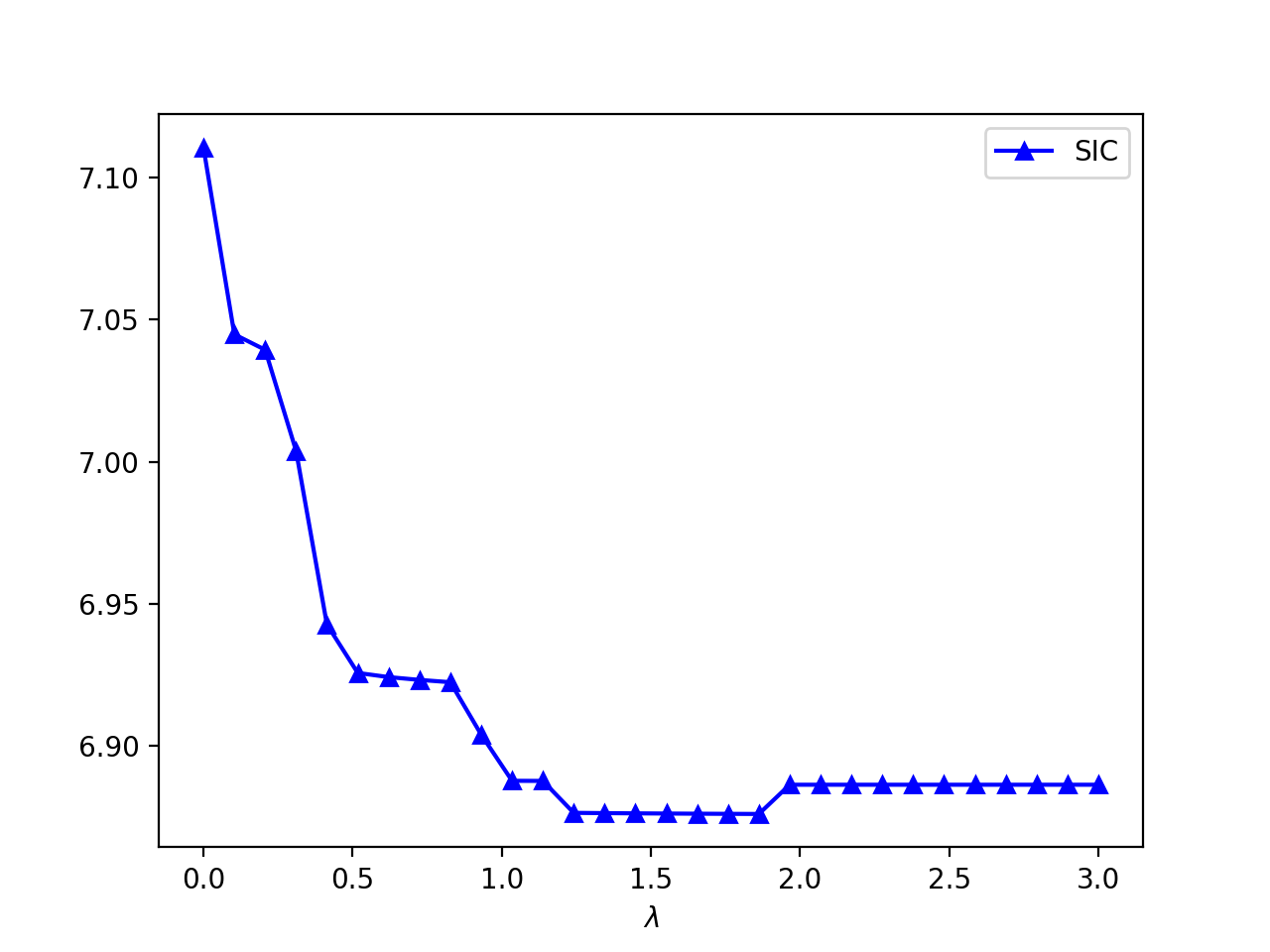}
    \end{subfigure}
    \caption{Numbers of selected groups $\widehat{K}$ in real data analysis under different tuning parameters $\lambda$ (left panel) and their corresponding values of SIC (right panel), for $\tau=0.5$.}
    \label{fig:k_sic_0pt5}
\end{figure}

\begin{figure}[!htp]
    \centering
    \begin{subfigure}[b]{0.525\textwidth}
        \includegraphics[width=\textwidth]{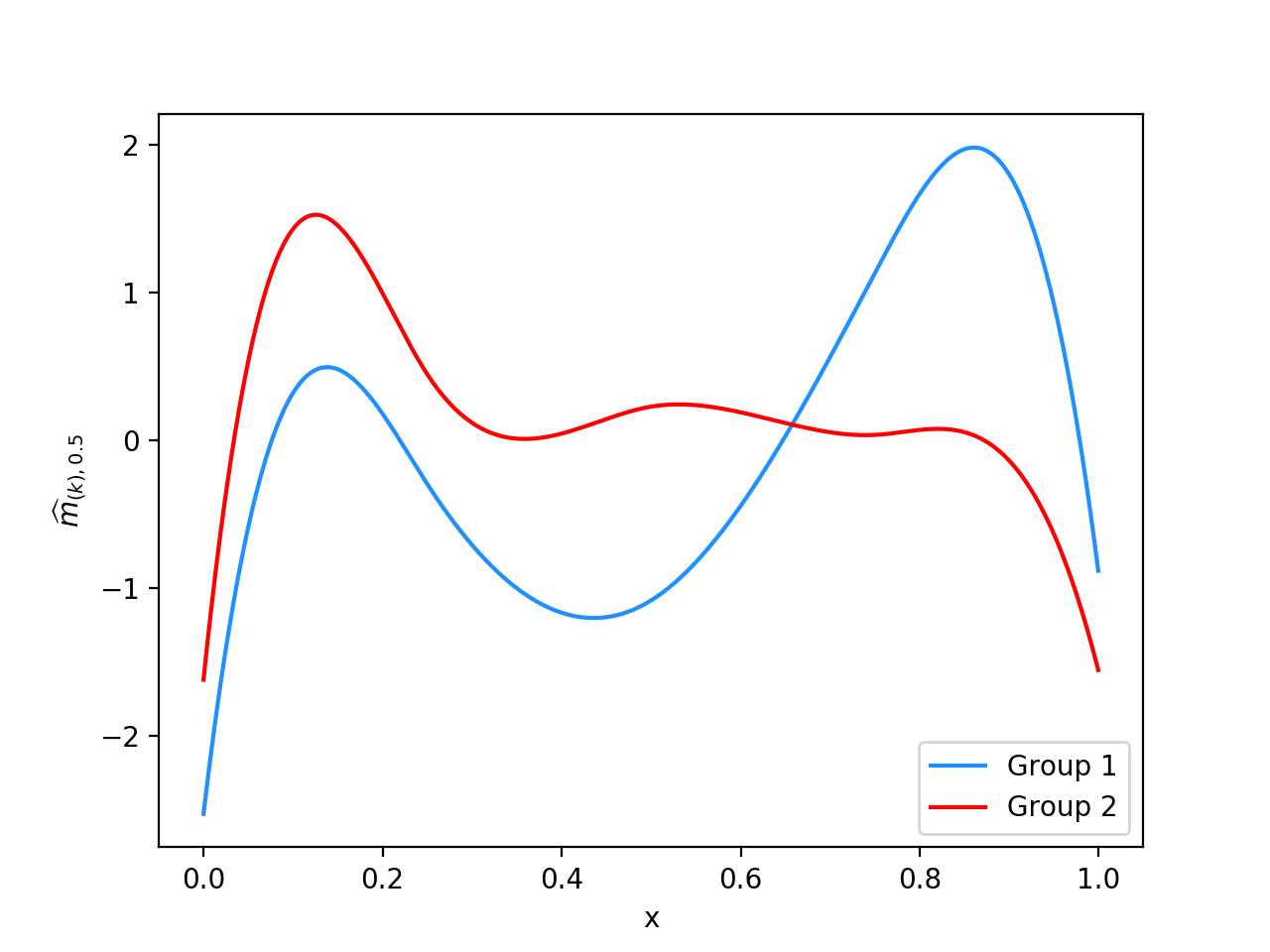}
    \end{subfigure}\hspace{-0.95cm}
    \begin{subfigure}[b]{0.525\textwidth}
        \includegraphics[width=\textwidth]{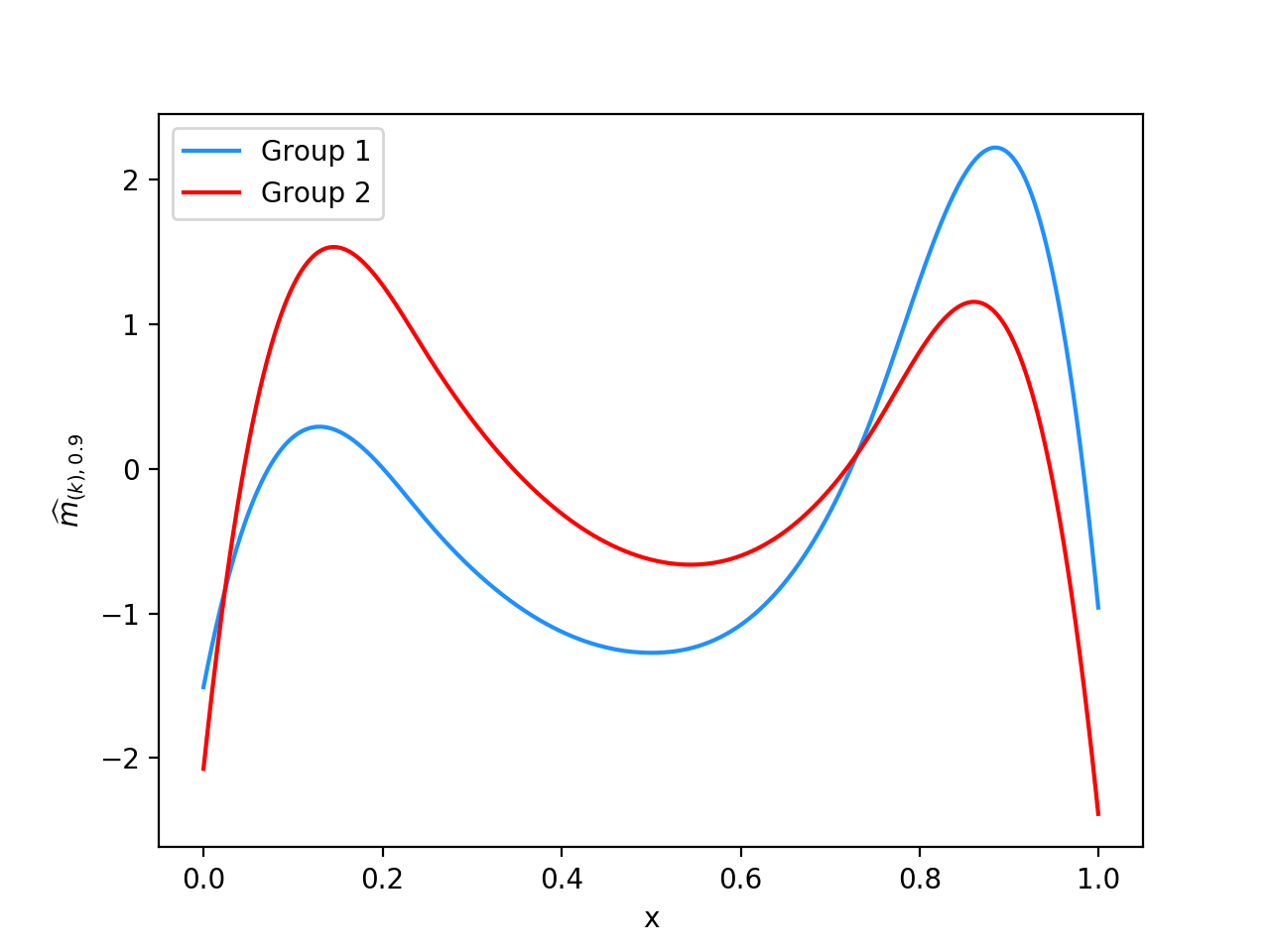}
    \end{subfigure}
    \caption{Fitted nonparametric functions, $\widehat{m}_{(k),\tau}$'s, in real data analysis at quantile levels $\tau=0.5$ (left panel) and $\tau=0.9$ (right panel).}
    \label{fig:subgroup_func}
\end{figure}

\begin{figure}[ht]
	\centering
	\includegraphics[width=11.6cm, height=13.42cm]{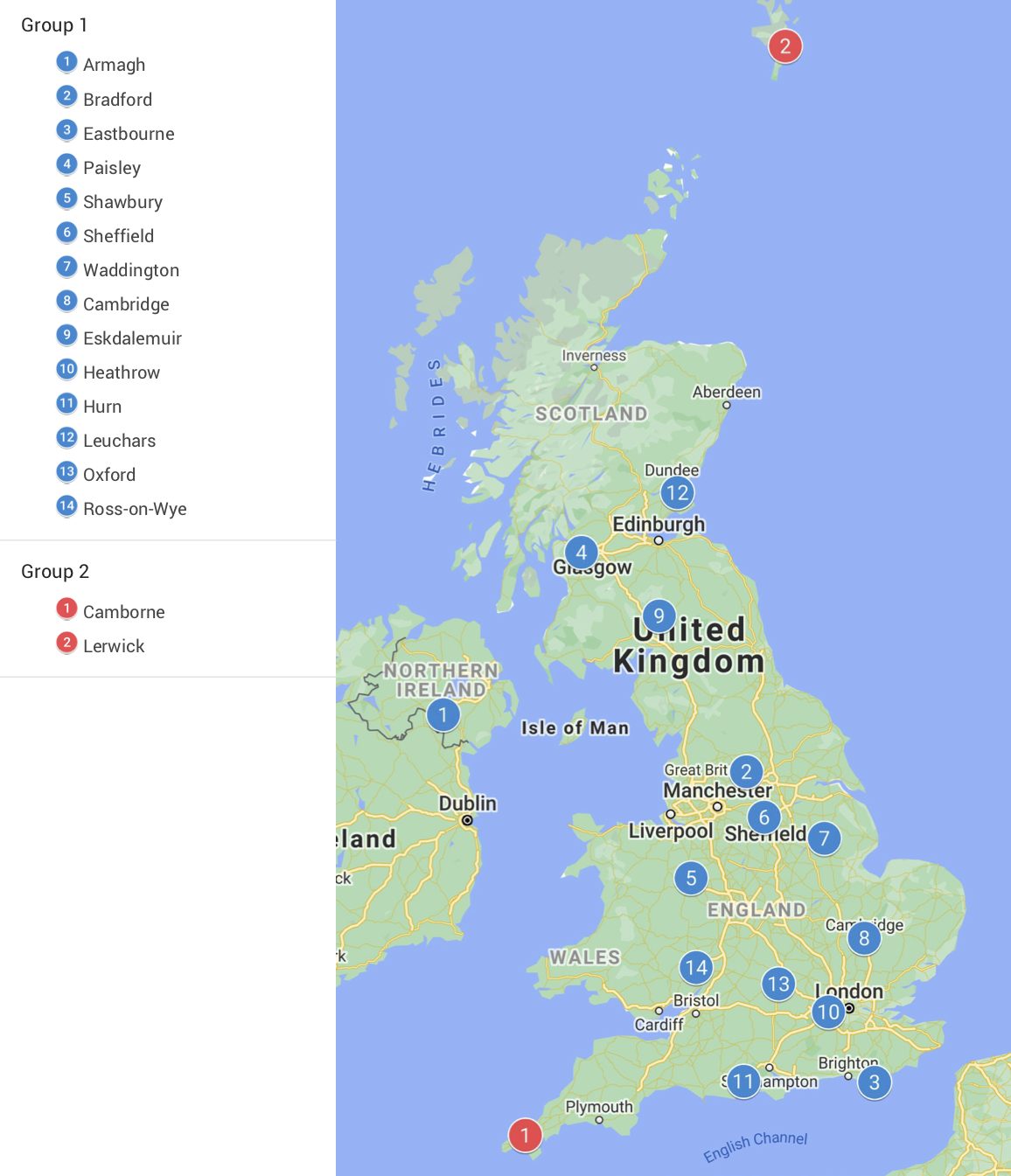}
	\caption{Subgroups of weather stations at the quantile level of $\tau=0.5$.}
	\label{fig:subgroup_map}
\end{figure}

\newpage

\section{Technical proofs}
	This section gives the technical proofs. Specifically, Section \ref{append:A} presents the asymptotic properties of the oracle estimator by providing the proofs of Theorem \ref{thm:1} and some auxiliary lemmas.
	The oracle property of the SCAD-penalized estimator and the consistency of the SIC are proved in Section \ref{append:B}.

\subsection{Proofs of Theorem \ref{thm:1} and auxiliary lemmas}\label{append:A}

In this subsection, we present the proof of Theorem \ref{thm:1} and relegate some auxiliary lemmas to the end of this appendix.

We start with some notations.
Throughout subsection \ref{append:A}, since we focus on the oracle estimator with a fixed quantile level $\tau$, we omit $\tau$ in all notations and simplify $\widehat{\mu}_i^\text{o}$ and $\widehat{\bm{\theta}}^\text{o}$ to $\widehat{\mu}$ and $\widehat{\bm{\theta}}$. let $\widehat{\bm{\mu}}=(\widehat{\mu}_1,\dots,\widehat{\mu}_n)^\top$.

For any $i=1,2,\dots,n$, denote $\widetilde{\bm{\Pi}}(x_{it})=(1,\bm{\Pi}(x_{it}^\top)^\top)$, $\bm{\vartheta}_i=(\mu_i,\bm{\theta}^\top)^\top$ and $\bm{\vartheta}_{0i}=(\mu_{0i},\bm{\theta}_{0}^\top)^\top$.
Let $M_{i}(\bm{\vartheta}_i):=T^{-1}\sum_{t=1}^T\rho_\tau(y_{it}-\widetilde{\bm{\Pi}}(x_{it})^\top\bm{\vartheta}_i)$, $\Delta_i^{(1)}(\bm{\vartheta}_i):=M_i(\bm{\vartheta}_i)-M_i(\bm{\vartheta}_{0i})$ and $\Delta_i^{(2)}(\bm{\vartheta}_i):=T^{-1}\sum_{t=1}^T[(\mu_i-\mu_{0i})+\bm{\Pi}(x_{it})^\top(\bm{\theta}-\bm{\theta}_0)](\tau-I\{e_{it}\leq0\})$. Define $f_i(0)=\mathbb{E}[f_i(0|x_{it})]$, $\bm{\gamma}_i:=f_i(0)^{-1}\mathbb{E}[f_i(0|x_{it})\bm{\Pi}(x_{it})]$, and $\bm{\Gamma}:=n^{-1}\sum_{i=1}^n\mathbb{E}[f_i(0|x_{it})\bm{\Pi}(x_{it})(\bm{\Pi}(x_{it})-\bm{\gamma}_i)^\top]$. Define the score vectors of the quantile regression problem $(nT)^{-1}\sum_{i=1}^n\sum_{t=1}^T\rho_\tau(y_{it}-\mu_i-\bm{\Pi}(x_{it})^\top\bm{\theta})$ as
\begin{equation}
	\begin{split}
		\mathbb{H}_i^{(1)}(\mu_i,\bm{\theta}):=&\frac{1}{T}\sum_{t=1}^T\{\tau-I(y_{it}\leq\mu_i+\bm{\Pi}(x_{it})^\top\bm{\theta})\}\\
		H_i^{(1)}(\mu_i,\bm{\theta}):=&\mathbb{E}[\mathbb{H}_i^{(1)}(\mu_i,\bm{\theta})]\\
		=&\mathbb{E}[\tau-F_i(\mu_i-\mu_{i0}+\bm{\Pi}(x_{it})(\bm{\theta}-\bm{\theta}_0)-R_{it}|x_{it})]\\
		\mathbb{H}^{(2)}(\bm{\mu},\bm{\theta}):=&\frac{1}{nT}\sum_{i=1}^n\sum_{t=1}^T\{\tau-I(y_{it}\leq\mu_i+\bm{\Pi}(x_{it})^\top\bm{\theta})\}\bm{\Pi}(x_{it})\\
		H^{(2)}(\bm{\mu},\bm{\theta}):=&\mathbb{E}[\mathbb{H}^{(2)}(\bm{\mu},\bm{\theta})]\\
		=&\frac{1}{n}\sum_{i=1}^n\mathbb{E}[\{\tau-F_i(\mu_i-\mu_{i0}+\bm{\Pi}(x_{it})(\bm{\theta}-\bm{\theta}_0)-R_{it}|x_{it})\}\bm{\Pi}(x_{it})]
	\end{split}
\end{equation}

\begin{proof}[\textbf{Proof of Theorem 1}]
	
	Throughout the proof, to ease the notations, we focus on the case with $K=1$. Since the subgroup structure is completely known when we define the oracle estimator, the results can be directly extended to the general $K>1$ groups.
	The proof consists of three steps. In the first step, we show the consistency of the oracle estimator $(\widehat{\bm{\mu}},\widehat{\bm{\theta}})$. In the second step, we show the convergence rates of $\max_{1\leq i\leq n}|\widehat{\mu}_i-\mu_{0i}|$ and $\|\widehat{\bm{\theta}}-\bm{\theta}_0\|_2$. Finally, we prove the asymptotic normality of $\widehat{\bm{\theta}}$ and the conditional variance of the estimated smooth function.\\
	
	\noindent\textit{Step 1. Consistency of $\widehat{\bm{\mu}}$ and $\widehat{\bm{\theta}}$}
	
	We first prove the consistency of $\widehat{\bm{\theta}}$. For $i=1,\dots,n$, note that
	\begin{equation}
		\begin{split}
			& M_i(\bm{\vartheta}_i)-M_i(\bm{\vartheta}_{0i})=\Delta_i^{(1)}(\bm{\vartheta}_{i})\\
			= & \underbrace{\Delta_i^{(1)}(\bm{\vartheta}_{i})+\Delta_i^{(2)}(\bm{\vartheta}_i)-\mathbb{E}[\Delta_i^{(1)}(\bm{\vartheta}_i)|\{x_{it}\}]}_{T_{1i}}
			- \underbrace{\Delta_i^{(2)}(\bm{\vartheta}_i)}_{T_{2i}} + \underbrace{\mathbb{E}[\Delta_i^{(1)}(\bm{\vartheta}_i)|\{x_{it}\}]}_{T_{3i}}.
		\end{split}
	\end{equation}
	
	Let $\xi_1(n,T)=\sqrt{H\log(n)^2/T}+H^{-d}$. Suppose that $\|\widehat{\bm{\theta}}-\bm{\theta}_0\|_2\geq L\xi_1(n,T)$ for some constant $L>0$. Then $\widehat{\bm{\vartheta}}_i=(\widehat{\mu}_i,\widehat{\bm{\theta}}^\top)^\top$ satisfies $\|\widehat{\bm{\vartheta}}_i-\bm{\vartheta}_{0i}\|_2\geq L\xi_1(n,T)$, for all $1\leq i\leq n$.
	By Lemmas \ref{lemma:T1_concentration}, \ref{lemma:4} and \ref{lemma:5}, we have $\max_{1\leq i\leq n}T_{1i}=o_p(\xi_1^2(n,T))$, $T_{2i}=L\cdot O_p(\xi^2_1(n,T))$, and $T_{3i}\geq CL^2\xi_1^2(n,T)$, respectively. Hence, for some sufficiently large $L$, $M_i(\bm{\vartheta}_i)>M_i(\bm{\vartheta}_{0i})$ for all $1\leq i\leq n$. 
	Hence, with probability approaching one, $\sum_{i=1}^nM_i(\bm{\vartheta}_i)>\sum_{i=1}^nM_i(\bm{\vartheta}_{0i})$, which however contradicts the definition of $\widehat{\mu}_i$ and $\widehat{\bm{\theta}}$. Therefore, we conclude that $\widehat{\bm{\theta}}=\bm{\theta}_0+O_p(\xi_1(n,T))=\bm{\theta}_0+o_p(1)$.
	
	Next, we prove the consistency of $\widehat{\mu}_i$, for $i=1,\dots,n$. Note that each $\widehat{\mu}_i$ is the minimizer of $M_i((\mu,\widehat{\bm{\theta}}^\top)^\top)$. Note that
	\begin{equation}
		\begin{split}
			& M_i((\mu_i,\widehat{\bm{\theta}}^\top)^\top)-M_i((\mu_{0i},\widehat{\bm{\theta}}^\top)^\top)= \Delta_i^{(1)}((\mu_i,\widehat{\bm{\theta}}^\top)^\top) - \Delta_i^{(1)}((\mu_{0i},\widehat{\bm{\theta}}^\top)^\top)\\
			= & \left[\Delta_i^{(1)}((\mu_i,\widehat{\bm{\theta}}^\top)^\top)+\Delta_i^{(2)}((\mu_i,\widehat{\bm{\theta}}^\top)^\top)-\mathbb{E}[\Delta_i^{(1)}((\mu_i,\bm{\theta}^\top)^\top)|\{x_{it}\}]|_{\bm{\theta}=\widehat{\bm{\theta}}}\right]\\
			- & \left[\Delta_i^{(1)}((\mu_{0i},\widehat{\bm{\theta}}^\top)^\top)+\Delta_i^{(2)}((\mu_{0i},\widehat{\bm{\theta}}^\top)^\top)-\mathbb{E}[\Delta_i^{(1)}((\mu_{0i},\bm{\theta}^\top)^\top)|\{x_{it}\}]|_{\bm{\theta}=\widehat{\bm{\theta}}}\right]\\
			+ & \left[\mathbb{E}\Delta_i^{(1)}[((\mu_i,\bm{\theta}^\top)^\top)|\{x_{it}\}]|_{\bm{\theta}=\widehat{\bm{\theta}}} - \mathbb{E}[\Delta_i^{(1)}((\mu_i,\bm{\theta}_0^\top)^\top)|\{x_{it}\}]\right]\\
			- & \left[\mathbb{E}[\Delta_i^{(1)}((\mu_{0i},\bm{\theta}^\top)^\top)|\{x_{it}\}]|_{\bm{\theta}=\widehat{\bm{\theta}}}-\mathbb{E}[\Delta_i^{(1)}((\mu_{0i},\bm{\theta}_0^\top)^\top)|\{x_{it}\}]\right]\\
			+ & \underbrace{\mathbb{E}[\Delta_i^{(1)}((\mu_i,\bm{\theta}_0^\top)^\top)|\{x_{it}\}] - \Delta_i^{(2)}((\mu_i,\widehat{\bm{\theta}}^\top)^\top) - \Delta_i^{(2)}((\mu_{0i},\widehat{\bm{\theta}}^\top)^\top)}_{T_{4i}}.
		\end{split}
	\end{equation}
	As $H^2\log(n)^2/T\to0$, we consider a positive sequence $\xi_2(n,T)$ such that $\sqrt{H}\xi_1(n,T)=o(\xi_2^2(n,T))$, suppose that $|\widehat{\mu}_i-\mu_{0i}|=L\xi_2(n,T)$. Then, by Lemmas \ref{lemma:4} and \ref{lemma:5}, we have that $\mathbb{E}[\Delta_i^{(1)}((\mu_i,\bm{\theta}_0^\top)^\top)|\{x_{it}\}]\geq CL^2\xi^2_2(n,T)$, $\Delta_i^{(2)}((\mu_i,\widehat{\bm{\theta}}^\top)^\top)= O_p(\xi_2^2(n,T))$, and $\Delta_i^{(2)}((\mu_{0i},\widehat{\bm{\theta}}^\top)^\top)=O_p(\xi_2^2(n,T))$. For some sufficiently large constant $L$, we have that $T_{8i}=CL\xi^2_2(n,T)$.
	Hence
	\begin{equation}
		\begin{split}
			&\mathbb{P}\left(\max_{1\leq i\leq n}|\widehat{\mu}_i-\mu_{0i}|> C\xi_2(n,T)\right)\\
			\leq & \mathbb{P}\left(M_i((\mu_i,\widehat{\bm{\theta}}^\top)^\top)< M_i((\mu_{0i},\widehat{\bm{\theta}}^\top)^\top),~\exists{1\leq i\leq n},~\exists|\mu_i-\mu_{0i}|>C\xi_2(n,T)\right)\\
			\leq & \mathbb{P}\left(\max_{1\leq i\leq n}\sup_{|\mu-\mu_{0i}|\leq L\xi_2(n,T)}\left|\Delta_i^{(1)}((\mu,\widehat{\bm{\theta}}^\top)^\top)+\Delta_i^{(2)}((\mu,\widehat{\bm{\theta}}^\top)^\top)-\mathbb{E}[\Delta_i^{(1)}((\mu,\bm{\theta}^\top)^\top)|\{x_{it}\}]|_{\bm{\theta}=\widehat{\bm{\theta}}}\right|>\frac{T_{8i}}{4}\right)\\
			+ & \mathbb{P}\left(\max_{1\leq i\leq n}\sup_{|\mu-\mu_{0i}|\leq L\xi_2(n,T)}\left|\mathbb{E}[\Delta_i^{(1)}((\mu_{0i},\bm{\theta}^\top)^\top)|\{x_{it}\}]|_{\bm{\theta}=\widehat{\bm{\theta}}}-\mathbb{E}[\Delta_i^{(1)}((\mu_{0i},\bm{\theta}_0^\top)^\top)|\{x_{it}\}]\right|>\frac{T_{8i}}{4}\right)\\
			:=&\mathbb{P}(A_1)+\mathbb{P}(A_2).
		\end{split}
	\end{equation}
	By Lemma \ref{lemma:T1_concentration}, we have $\mathbb{P}(A_1)\to0$, as $T\to\infty$. In addition, since
	\begin{equation}
		|\Delta_i^{(1)}((\mu_{0i},\bm{\theta}^\top)^\top)-\Delta_i^{(1)}((\mu_{0i},\bm{\theta}_0^\top)^\top)|\leq 2\|\widetilde{\bm{\Pi}}(x_{it})\|_2\cdot\|\bm{\theta}-\bm{\theta}_0\|_2,
	\end{equation}
	it is obtained that $\mathbb{P}(A_2)\to0$ provided that $\|\bm{\theta}-\bm{\theta}_0\|_2=O_p(\xi_1(n,T))$, $\|\widetilde{\bm{\Pi}}(x_{it})\|_2\leq\sqrt{H}$, and $H^2\log(n)^2/T\to0$. Therefore, we prove the consistency of $\widehat{\mu}_1,\dots,\widehat{\mu}_n,\widehat{\bm{\theta}}$ under the conditions in Theorem 1.\\
	
	\noindent\textit{Step 2. Rate of $\max_{1\leq i\leq n}|\widehat{\mu}_i-\mu_{0i}|$ and $\|\widehat{\bm{\theta}}-\bm{\theta}_0\|_2$}
	
	As $\widehat{\mu}_1,\dots,\widehat{\mu}_n,\widehat{\bm{\theta}}$ are consistent, by Lemma \ref{lemma:asymp_represent}, we have the following asymptotic representations
	\begin{equation}
		\begin{split}
			&\widehat{\mu}_i-\mu_{0i} + o_p(|\widehat{\mu}_i-\mu_{0i}|)\\
			= & -\bm{\gamma}_i^\top(\widehat{\bm{\theta}}-\bm{\theta}_0)+f_i(0)^{-1}\left\{\mathbb{H}_i^{(1)}(\widehat{\mu}_i,\widehat{\bm{\theta}}) - \mathbb{H}_i^{(1)}(\mu_{0i},\bm{\theta}_0) - H_i^{(1)}(\widehat{\mu}_i,\widehat{\bm{\theta}})\right\}\\
			+&f_i(0)^{-1}\mathbb{H}_i^{(1)}(\mu_{0i},\bm{\theta}_0)+O_p(T^{-1}\vee H^{-d}\vee\|\widehat{\bm{\theta}}-\bm{\theta}_0\|_2^2),
		\end{split}
	\end{equation}
	for all $i=1,\dots,n$, and
	\begin{equation}
		\begin{split}
			& \widehat{\bm{\theta}}-\bm{\theta}_0 + o_p(\|\widehat{\bm{\theta}}-\bm{\theta}_0\|_2)\\
			= & \bm{\Gamma}^{-1}\underbrace{\left[-\frac{1}{n}\sum_{i=1}^n\mathbb{H}_i^{(1)}(\mu_{0i},\bm{\theta}_0)\bm{\gamma}_i + \mathbb{H}^{(2)}(\bm{\mu}_0,\bm{\theta}_0)\right]}_{T_{5i}}\\
			- & \bm{\Gamma}^{-1} \underbrace{\left[\frac{1}{n}\sum_{i=1}^n\left\{\mathbb{H}_i^{(1)}(\widehat{\mu}_i,\widehat{\bm{\theta}}) - \mathbb{H}_i^{(1)}(\mu_{0i},\bm{\theta}_0) - H_i^{(1)}(\widehat{\mu}_i,\widehat{\bm{\theta}})\right\}\bm{\gamma}_i\right]}_{T_{6i}}\\
			+ & \bm{\Gamma}^{-1}\underbrace{\left[\mathbb{H}^{(2)}(\widehat{\bm{\mu}},\widehat{\bm{\theta}}) - \mathbb{H}^{(2)}(\bm{\mu}_0,\bm{\theta}_0) - H^{(2)}(\widehat{\bm{\mu}},\widehat{\bm{\theta}})\right]}_{T_{7i}}\\
			+ & O_p\left(T^{-1}H^{1/2}\vee H^{-d}\vee \max_{1\leq i\leq n}|\widehat{\mu}_i-\mu_{0i}|^2\right).
		\end{split}
	\end{equation}
	As $\bm{\gamma}_i\leq\sqrt{H}$, $\|T_{5i}\|_2=O_p(\sqrt{H/(nT)})$.
	Because of the consistency of $(\widehat{\bm{\mu}},\widehat{\bm{\theta}})$, by taking $\delta=H^{-1/2}n^{-1/2}T^{-1/3}$ in Lemma \ref{lemma:ep_rates}, $\|T_{6i}\|_2$ and $\|T_{7i}\|_2$ are both $o_p(\sqrt{H/(nT)})$, which implies that
	\begin{equation}
		\|\widehat{\bm{\theta}}-\bm{\theta}_0\|_2=O_p\left(\max_{1\leq i\leq n}|\widehat{\mu}_i-\mu_{0i}|^2\right)+O_p(\sqrt{H/(nT)}\vee T^{-1}H^{1/2}\vee H^{-d})
	\end{equation}
	and
	\begin{equation}
		\begin{split}
			& \max_{1\leq i\leq d}|\widehat{\mu}_{i}-\mu_{0i}|\\
			\leq & C\left\{\max_{1\leq i\leq n}|\mathbb{H}_i^{(1)}(\mu_{0i},\bm{\theta}_0)|+\max_{1\leq i\leq n}|\mathbb{H}_i^{(1)}(\widehat{\mu}_i,\widehat{\bm{\theta}}) - \mathbb{H}_i^{(1)}(\mu_{0i},\bm{\theta}_0) - H_i^{(1)}(\widehat{\mu}_i,\widehat{\bm{\theta}})|\right\}\\
			+&O_p(\sqrt{H/(nT)}\vee T^{-1}H^{1/2}\vee H^{-d}).
		\end{split}
	\end{equation}
	
	By taking the union upper bound and Lemma \ref{lemma:1},
	\begin{equation}
		\begin{split}	
			&\mathbb{P}\left[\max_{1\leq i\leq n}|\mathbb{H}_i^{(1)}(\mu_{0i},\bm{\theta}_0)|\geq C\sqrt{\frac{\log(n)}{T}}\right]\\
			\leq & \sum_{i=1}^n\mathbb{P}\left[\max_{1\leq i\leq n}|\mathbb{H}_i^{(1)}(\mu_{0i},\bm{\theta}_0)|\geq C\sqrt{\frac{\log(n)}{T}}\right]\leq2\exp(-C\log(n)),
		\end{split}
	\end{equation}
	which implies that $\max_{1\leq i\leq n}|\mathbb{H}_i^{(1)}(\mu_{0i},\bm{\theta}_0)|=O_p(\sqrt{\log(n)/T})$. Additionally, because of consistency of $\widehat{\bm{\mu}}$ and $\widehat{\bm{\theta}}$, by Lemma \ref{lemma:ep_rates}, for any $\epsilon>0$,
	\begin{equation}
		\max_{1\leq i\leq n}\mathbb{P}\left[|\mathbb{H}_i^{(1)}(\widehat{\mu}_i,\widehat{\bm{\theta}}) - \mathbb{H}_i^{(1)}(\mu_{0i},\bm{\theta}_0) - H_i^{(1)}(\widehat{\mu}_i,\widehat{\bm{\theta}})|>\epsilon\sqrt{\log(n)/T}\right]=o(n^{-1}).
	\end{equation}
	Therefore, we have $\max_{1\leq i\leq n}|\widehat{\mu}_i-\mu_{0i}|=O_p(\sqrt{\log(n)/T}+H^{-d})$ and $\|\widehat{\bm{\theta}}-\bm{\theta}_0\|_2=O_p(\sqrt{H/(nT)}+(H\log(n)/T)^{3/4}+H^{-d})$. If $Hn^2\log(n)^3/T\to0$, then $(H\log(n)/T)^{3/4}=o(\sqrt{H/(nT)})$ and $\|\widehat{\bm{\theta}}-\bm{\theta}_0\|_2=O_p(\sqrt{H/(nT)}+H^{-d})$.\\
	
	\noindent\textit{Step 3. Asymptotic normality of $\widehat{\bm{\theta}}$ and estimated function}
	
	Note that $\|\widehat{\bm{\theta}}-\bm{\theta}_0\|_2=O_p(\sqrt{H/(nT)}+H^{-d})$ and
	\begin{equation}\label{eq:expansion}
		\begin{split}
			&\sum_{t=1}^T\sum_{i=1}^n\mathbb{E}[\rho_\tau(y_{it}-\widehat{\mu}_{i}-\bm{\Pi}(x_{it})^\top\bm{\widehat{\theta}})|x_{it}]-\sum_{t=1}^T\sum_{i=1}^n\mathbb{E}[\rho_\tau(y_{it}-\widehat{\mu}_i-\bm{\Pi}(x_{it})^\top{\bm{\theta}}_{0})|x_{it}]\\
			=&\sum_{t=1}^T\sum_{i=1}^n\int_{\bm{\Pi}(x_{it})^\top\bm{\theta}_{0}-m_{it}+\widehat{\mu}_i-\mu_{0i}}^{\bm{\Pi}(x_{it})^\top\bm{\widehat{\theta}}-m_{it}+\widehat{\mu}_i-\mu_{0i}}F_i(z|x_{it})-F_i(0|x_{it})dz\\
			=&\frac{1}{2}\sum_{t=1}^T\sum_{i=1}^nf_k(0|x_{it})[(\bm{\Pi}(x_{it})^\top(\bm{\widehat{\theta}}-\bm{\theta}_{0}))^2+2\bm{\Pi}(x_{it})^\top(\bm{\widehat{\theta}}-\bm{\theta}_{0})\widetilde{R}_{it}]\\
			+&O_p\left(nT[\sqrt{H}(\sqrt{H/(nT)}+H^{-d})]^3\right),
		\end{split}
	\end{equation}
	where $\widetilde{R}_{it}=\bm{\Pi}(x_{it})^\top\bm{\theta}_{0}-m_{i}(x_{it})+\widetilde{\mu}_i-\mu_{0i}=O(H^{-d}+\sqrt{\log(n)/T})$. Define
	\begin{equation}
		\begin{split}
			\widetilde{\bm{\theta}}:=\underset{\bm{\theta}}{\argmin}&\sum_{t=1}^T\sum_{i=1}^n\Big\{-\bm{\Pi}(x_{it})^\top(\bm{\theta}-\bm{\theta}_{0})(\tau-I\{e_{it}\leq0\})\\
			+&\frac{1}{2}f_i(0|x_{it})[(\bm{\Pi}(x_{it})^\top(\bm{\theta}-\bm{\theta}))^2+2\bm{\Pi}(x_{it})^\top(\bm{\theta}-\bm{\theta}_0)\widetilde{R}_{it}]\Big\}.
		\end{split}
	\end{equation}
	We have obviously
	\begin{equation}
		\widetilde{\bm{\theta}}=\bm{\theta}_{0}+(\bm{Z}^\top\bm{f}\bm{Z})^{-1}(-\bm{Z}^\top \bm{f}\bm{R}+\bm{Z}^\top\bm{\epsilon}),
	\end{equation}
	where $\bm{Z}=[\bm{\Pi}(x_{11}),\dots,\bm{\Pi}(x_{1T}),\bm{\Pi}(x_{21}),\dots,\bm{\Pi}(x_{nT})]^\top$, $\bm{f}=\text{diag}(f_1(0|x_{11}),\dots,f_n(0|x_{nT}))$, $\widetilde{\bm{R}}=(\widetilde{R}_{11},\dots,\widetilde{R}_{nT})^\top$, and $\bm{\epsilon}=((\tau-I\{e_{11}\leq0\}),\dots,(\tau-I\{e_{nT}\leq0\}))^\top$.
	
	First consider $\bm{\Pi}(x)^\top(\bm{Z}^\top \bm{f}\bm{Z})^{-1}\bm{Z}^\top\bm{\epsilon}$. Its conditional asymptotic variance is given by $\tau(1-\tau)\bm{\Pi}(x)^\top(\bm{Z}^\top\bm{f}\bm{Z})^{-1}(\bm{Z}^\top\bm{Z})(\bm{Z}^\top\bm{f}\bm{Z})^{-1}\bm{\Pi}(x)\asymp H/(nT)$. Using Lindeberg-Feller condition, similar to the proof of Theorem 3.1 of \citet{zhou1998local}, and by a central limit theorem for $\alpha$-mixing sequences, we have
	\begin{equation}
		\left[\tau(1-\tau)\bm{\Pi}(x)^\top(\bm{Z}^\top \bm{f}\bm{Z})^{-1}(\bm{Z}^\top\bm{Z})(\bm{Z}^\top \bm{f}\bm{Z})^{-1}\bm{\Pi}(x)\right]^{-\frac{1}{2}}\bm{\Pi}(x)^\top(\bm{Z}^\top\bm{f}\bm{Z})^{-1}\bm{Z}^\top\bm{\epsilon}\overset{d}{\to}N(0,1).
	\end{equation}
	By Lemma \ref{lemma:3} and $|R_{it}|=O(H^{-d}+\sqrt{\log(n)/T})$,
	\begin{equation}
		\bm{\Pi}(x)^\top(\bm{Z}^\top\bm{f}\bm{Z})^{-1}\bm{Z}_k^\top \bm{R}_k=O_p(\sqrt{H/(nT)}(H^{-d}+\sqrt{\log(n)/T}))=o_p(\sqrt{H/(n_T)}).
	\end{equation}
	Thus,
	\begin{equation}
		\frac{\bm{\Pi}(x)^\top(\widetilde{\bm{\theta}}-\bm{\theta}_{0})}{(\tau(1-\tau)\bm{\Pi}(x)^\top(\bm{Z}^\top\bm{f}\bm{Z})^{-1}(\bm{Z}^\top\bm{Z})(\bm{Z}^\top\bm{f}\bm{Z})^{-1}\bm{\Pi}(x))^{1/2}}\to N(0,1).
	\end{equation}
	
	Denote
	\begin{equation}
		\begin{split}
			& Q(\bm{\theta}) = -\sum_{i=1}^n\sum_{t=1}^T[\bm{\Pi}(x_{it})^\top(\bm{\theta}-\bm{\theta}_0)](\tau-I(e_{it}\leq 0))\\
			& + \sum_{i=1}^n\sum_{t=1}^T\mathbb{E}[\rho_\tau(y_{it}-\widehat{\mu}_i-\bm{\Pi}(x_{it})^\top\bm{\theta})|x_{it}]- \sum_{i=1}^n\sum_{t=1}^T\mathbb{E}[\rho_\tau(y_{it}-\widehat{\mu}_{i}-\bm{\Pi}(x_{it})^\top\bm{\theta}_0)|x_{it}].
		\end{split}
	\end{equation}
	If $\|\bm{\theta}-\widetilde{\bm{\theta}}\|_2=\delta\xi(n,T)$ where $\delta$ is any positive constant, by a similar argument as Lemma \ref{lemma:T1_concentration} with all information of $n$ individuals combined, we have
	\begin{equation}
		\begin{split}	
			\sup_{\|\bm{\theta}-\widetilde{\bm{\theta}}\|_2\leq\delta\xi(n,T)}&\Bigg|\sum_{i=1}^n\sum_{t=1}^T\rho_\tau(y_{it}-\widehat{\mu}_i-\bm{\Pi}(x_{it})^\top\bm{\theta})-\sum_{i=1}^n\sum_{t=1}^T\rho_\tau(y_{it}-\widehat{\mu}_i-\bm{\Pi}(x_{it})^\top\widetilde{\bm{\theta}})\\
			&\quad\quad\quad-[Q(\bm{\theta})-Q(\widetilde{\bm{\theta}})]\Bigg| = o_p(nT\xi^2(n,T)).
		\end{split}
	\end{equation}
	By comparing $Q(\bm{\theta})$ with \eqref{eq:expansion}, $Q(\bm{\theta})$ is a quadratic function of $\bm{\theta}-\widetilde{\bm{\theta}}$ after ignoring the small term $O_p(nT[\sqrt{H}(\sqrt{H/(nT)}+H^{-d})]^3)$. As $\widetilde{\bm{\theta}}$ is the minimizer of the quadratic function. When $\|\bm{\theta}-\widetilde{\bm{\theta}}\|_2=\delta\xi(n,T)$,
	\begin{equation}
		|Q(\bm{\theta})-Q(\widetilde{\bm{\theta}})|\geq CnT\|\bm{\theta}-\widetilde{\bm{\theta}}\|_2^2-O_p(nT[\sqrt{H}(\sqrt{H/(nT)}+H^{-d})]^3)\geq CnT\|\bm{\theta}-\widetilde{\bm{\theta}}\|_2^2.
	\end{equation}
	
	Therefore, we have that with probablity approaching one
	\begin{equation}
		\inf_{\|\bm{\theta}-\widetilde{\bm{\theta}}\|_2=\delta\xi(n,T)}\sum_{i=1}^n\sum_{t=1}^T\left[\rho_\tau(y_{it}-\widehat{\mu}_i-\bm{\Pi}(x_{it})^\top\bm{\theta})-\rho_\tau(y_{it}-\widehat{\mu}_{i}-\bm{\Pi}(x_{it})^\top\widetilde{\bm{\theta}})\right]>0.
	\end{equation}
	By the convexity of $\rho_\tau(\cdot)$ function and the definition of $\widehat{\mu}_i$ and $\widehat{\bm{\theta}}$, this implies that $\|\bm{\theta}-\widetilde{\bm{\theta}}\|_2=o_p(\xi(n,T))$. Therefore, $\widehat{\bm{\theta}}$ has the same asymptotic properties as $\widetilde{\bm{\theta}}$.
	
	Finally, by the B-spline approximation error, if $H(nT)^{-1/(2d+1)}\to\infty$, $|\bm{\Pi}(x)^\top\bm{\theta}_{0}-m_{i}(x)|=o_p(\sqrt{H/nT})$, and the above results imply that
	\begin{equation}
		\frac{\bm{\Pi}(x)^\top\widehat{\bm{\theta}}-m_i(x)}{(\tau(1-\tau)\bm{\Pi}(x)^\top(\bm{Z}^\top\bm{f}\bm{Z})^{-1}(\bm{Z}^\top\bm{Z})(\bm{Z}^\top\bm{f}\bm{Z})^{-1}\bm{\Pi}(x))^{1/2}}\to N(0,1).
	\end{equation}
	
\end{proof}

We state some auxiliary lemmas used for the proof of Theorem \ref{thm:1}.
The first lemma is the Bernstein-type inequality for the geometrically $\alpha$-mixing sequence.
It is a corollary of Theorem 2.19 in \citet{fan2008nonlinear} by taking $q\asymp n/\log(n)$ in their theorem.

\begin{lemma}\label{lemma:1}
	Let $\{x_{it}\}$ be a strictly stationary $\alpha$-mixing process with mean zero and mixing coefficient $\alpha(l)\leq r^l$ for some $r\in(0,1)$. Suppose that $\mathbb{E}|x_t|^k\leq Ck!A^{k-2}D^2$, $k=3,4,\dots$, then for any $\varepsilon>0$,
	\begin{equation}
		\mathbb{P}\left(\left|\sum_{t=1}^Tx_t\right|>T\varepsilon\right)\leq C\log(T)\exp\left[-C\frac{T}{\log(T)}\frac{\varepsilon^2}{\varepsilon A+D^2}\right].
	\end{equation}
\end{lemma}~

Next, we state some intermediate results in the proof of Theorem \ref{thm:1} in the following lemmas and present their proofs.

\begin{lemma}\label{lemma:T1_concentration}
	Let $d(n,T)$ be a sequence depending on $n$ and $T$ such that $d(n,T)\to0$ as $T\to\infty$ and $\sqrt{H\log(n)^2T^{-1}}+H^{-d}=O(\xi(n,T))$. Under the conditions in Theorem \ref{thm:1},
	\begin{equation}\begin{split}
			&\max_{1\leq i\leq n}\sup_{\|\bm{\bm{\vartheta}}_{i}-\bm{\bm{\vartheta}}_{0i}\|_2=d(n,T)}\left|\Delta_i^{(1)}(\bm{\vartheta}_{i})+\Delta_i^{(2)}(\bm{\vartheta}_i)-\mathbb{E}[\Delta_i^{(1)}(\bm{\vartheta}_i)|\{x_{it}\}]\right|=o_p(d^2(n,T)).
	\end{split}\end{equation}
\end{lemma}

\begin{proof}[\textbf{Proof of Lemma \ref{lemma:T1_concentration}}]
	Let $\mathcal{N}_i=\{\bm{\vartheta}_i^{(1)},\dots,\bm{\vartheta}_i^{(N)}\}$ be a $\delta(n,T)$ covering of $\{\bm{\theta}:\|\bm{\vartheta}_i-\bm{\vartheta}_{0i}\|_2\leq d(n,T)\}$. The size of $\mathcal{N}_{(k)}$ is bounded by $N\leq(Cd(n,T)/\delta(n,T))^{H}$ and thus $\log N\leq CH\log(T)$ if we choose $\delta(n,T)\asymp T^{-a}d(n,T)$ for some $a>0$.
	
	Let $\Delta_{it}(\bm{\vartheta}_i)=\rho_\tau(y_{it}-\widetilde{\bm{\Pi}}(x_{it})^\top\bm{\vartheta}_i)-\rho_\tau(y_{it}-\widetilde{\bm{\Pi}}(x_{it})^\top\bm{\vartheta}_{0i})+\widetilde{\bm{\Pi}}(x_{it})^\top(\bm{\vartheta}_i-\bm{\vartheta}_{0i})(\tau-I\{e_{it}\leq0\})$.
	Using the Lipschitz property of $\rho_\tau(\cdot)$, and that for any $\bm{\vartheta}_{i}$, there exists some  $\bm{\vartheta}_i^{(l)}$ such that $\|\bm{\vartheta}_{i}-\bm{\vartheta}^{(l)}_i\|_2\leq\delta(n,T)$, we have
	\begin{equation}
		\begin{split}
			&\left|\Delta_i^{(1)}(\bm{\vartheta}_{i})+\Delta_i^{(2)}(\bm{\vartheta}_i)-\mathbb{E}[\Delta_i^{(1)}(\bm{\vartheta}_i)|\{x_{it}\}]-\Delta_i^{(1)}(\bm{\vartheta}_{i}^{(l)})-\Delta_i^{(2)}(\bm{\vartheta}_i^{(l)})+\mathbb{E}[\Delta_i^{(1)}(\bm{\vartheta}_i^{(l)})|\{x_{it}\}]\right|\\
			&=\left|\frac{1}{T}\sum_{t=1}^T\Delta_{it}(\bm{\vartheta}_i)-\mathbb{E}[\Delta_{it}(\bm{\vartheta}_i)|\{x_{it}\}]\right|\leq \frac{C}{T}\sum_{t=1}^T|\widetilde{\bm{\Pi}}(x_{it})^\top(\bm{\vartheta}_i-\bm{\vartheta}_i^{(l)})|=O(\sqrt{H}\delta(n,T)),
		\end{split}
	\end{equation}
	which can obviously be made to be $o_p(\xi^2(n,T))$ by setting $\delta(n,T)\asymp T^{-a}\xi(n,T)$ for some $a$ large enough.
	
	Denote $m_{it}=m(x_{it})$. Using that $\rho_\tau(x)=|x|/2+(\tau-1/2)x$, by simple algebra,
	\begin{equation}
		\begin{split}
			|\Delta_{it}&(\bm{\vartheta}_{i})|
			=\Big|\frac{1}{2}|e_{it}+m_{it}+\mu_{0i}-\widetilde{\bm{\Pi}}(x_{it})^\top\bm{\vartheta}_{i}|-\frac{1}{2}|e_{it}+m_{it}+\mu_{0i}-\widetilde{\bm{\Pi}}(x_{it})^\top\bm{\vartheta}_{0i}|\\
			+&\widetilde{\bm{\Pi}}(x_{it})^\top(\bm{\vartheta}_{i}-\bm{\vartheta}_{0i})(1/2-I\{e_{it}\leq0\})\Big|\\
			\leq&|\widetilde{\bm{\Pi}}(x_{it})^\top(\bm{\vartheta}_{i}-\bm{\vartheta}_{0i})|\cdot I\{|e_{it}|\leq|\widetilde{\bm{\Pi}}(x_{it})^\top(\bm{\vartheta}_{i}-\bm{\vartheta}_{0i})|+|m_{it}+\mu_{0i}-\widetilde{\bm{\Pi}}(x_{it})^\top\bm{\vartheta}_{0i}|\}.
		\end{split}
	\end{equation}
	Thus, $|\Delta_{it}(\bm{\vartheta}_{i})|\leq C\sqrt{H}\xi(n,T):=A$.
	
	Furthermore, we have
	\begin{equation}
		\begin{split}
			&\mathbb{E}[(\Delta_{it}(\bm{\vartheta}_{i})-\mathbb{E}[\Delta_{it}(\bm{\vartheta}_{i})|\{x_{it}\}])^2]\leq\mathbb{E}|\Delta_{it}(\bm{\vartheta}_{i})|^2\\
			\leq&\mathbb{P}\left\{|e_{it}|\leq|\widetilde{\bm{\Pi}}(x_{it})^\top(\bm{\vartheta}_{i}-\bm{\vartheta}_{0i})|+|m_{it}+\mu_{0i}-\widetilde{\bm{\Pi}}(x_{it})^\top\bm{\vartheta}_{0i}|\right\}\cdot\mathbb{E}|\widetilde{\bm{\Pi}}(x_{it})^\top(\bm{\vartheta}_{i}-\bm{\vartheta}_{0i})|^2\\
			\leq&[C\sqrt{H}d(n,T)]\cdot\mathbb{E}|\widetilde{\bm{\Pi}}(x_{it})^\top(\bm{\vartheta}_{i}-\bm{\vartheta}_{0i})|^2\\
			\leq & C\sqrt{H}d^3(n,T):=D^2,
		\end{split}
	\end{equation}
	where the first factor $C\sqrt{H}d(n,T)$ comes from $\mathbb{P}(|e_{it}|\leq |\widetilde{\bm{\Pi}}(x_{it})^\top(\bm{\vartheta}_{i}-\bm{\vartheta}_{0i})|+|m_{it}+\mu_{0i}-\widetilde{\bm{\Pi}}(x_{it})^\top\bm{\vartheta}_{0i}|)$ by Assumption (A3).
	
	Using Bernstein's inequality in Lemma \ref{lemma:1}, together with the union bound, we have that for any $a>0$,
	\begin{equation}
		\begin{split}
			&\mathbb{P}\left(\sup_{\bm{\vartheta}_{i}\in\mathcal{N}_{i}}\left|\Delta^{(1)}_i(\bm{\vartheta}_{i})+\Delta^{(2)}_i(\bm{\vartheta}_{i})-\mathbb{E}[\Delta^{(1)}_i(\bm{\vartheta}_{i})|\{x_{it}\}]\right|>a\right)\\
			\leq& C(T)^{CH}\log(T)\exp\left[-C\frac{T}{\log(T)}\frac{a^2}{aA+D^2}\right].
		\end{split}
	\end{equation}
	Letting $a=Cd^2(n,T)$, we have
	\begin{equation}
		\begin{split}
			&\mathbb{P}\left(\max_{1\leq i\leq n}\sup_{\bm{\vartheta}_{i}\in\mathcal{N}_{i}}\left|\Delta^{(1)}_i(\bm{\vartheta}_{i})+\Delta^{(2)}_i(\bm{\vartheta}_{i})-\mathbb{E}[\Delta^{(1)}_i(\bm{\vartheta}_{i})|\{x_{it}\}]\right|>Cd^2(n,T)\right)\\
			\leq& Cn(T)^{CH}\log(T)\exp\left[-C\frac{T}{\log(T)}\frac{a^2}{aA+D^2}\right]\\
			\leq&Cn(T)^{CH}\log(T)\exp\left[-C\frac{T}{\log(T)}H^{-1/2}\xi(n,T)\right]\\
			\leq & \exp\left[\log(n)+CH\log(T)+C\log\log(T)-C\frac{\sqrt{T}\log(n)}{\log(T)}-C\frac{TH^{-d-1/2}}{\log(T)}\right]\to0.
		\end{split}
	\end{equation}
	
\end{proof}

\begin{lemma}
	\label{lemma:3}
	Under the conditions of Theorem \ref{thm:1}, the eigenvalues of
	$T^{-1}\sum_{t=1}^T\bm{\Pi}(x_{it})\bm{\Pi}(x_{it})^\top$ and $T^{-1}\sum_{t=1}^T\widetilde{\bm{\Pi}}(x_{it})\widetilde{\bm{\Pi}}(x_{it})^\top$
	are bounded and bounded away from zero uniformly over $i=1,\dots,n$, with probability approaching one.
\end{lemma}

\begin{proof}[\textbf{Proof of Lemma \ref{lemma:3}}]
	
	We focus on the proof of $T^{-1}\sum_{t=1}^T\bm{\Pi}(x_{it})\bm{\Pi}(x_{it})^\top$, since the statement for $T^{-1}\sum_{t=1}^T\widetilde{\bm{\Pi}}(x_{it})\widetilde{\bm{\Pi}}(x_{it})^\top$ can be proved in an analogous fashion.
	
	Since $\bm{\Pi}(x_{it})=\sqrt{H}\widetilde{\bm{O}}\bm{B}(x_{it})$ and the eigenvalues of $\mathbb{E}[\bm{\Pi}(x_{it})\bm{\Pi}(x_{it})^\top]$ are bounded away from zero and infinity, the desired statement for $\bm{\Pi}(x_{it})$ is implied by
	\begin{equation}
		\left|\frac{H}{T}\sum_{t=1}^T\bm{B}_h(x_{it})\bm{B}_{h'}(x_{it})-H\mathbb{E}\bm{B}_h(x_{it})\bm{B}_{h'}(x_{it})\right|=o_p(1/H)
	\end{equation}
	for all $1\leq h,h'\leq H$. Denote $V_{it}^{h,h'}=H\bm{B}_h(x_{it})\bm{B}_{h'}(x_{it})=\bm{e}_h^\top\sqrt{H}\bm{B}(x_{it})\sqrt{H}\bm{B}(x_{it})^\top\bm{e}_{h'}$, where $\bm{e}_i$ is the vector whose $i$-th entry is one and other entries are all zero. We have
	\begin{equation}
		\mathbb{E}[(V_{it}^{h,h'})^2]\leq|e_h^\top\sqrt{H}\bm{B}(x_{it})|^2\mathbb{E}[e_{h'}^\top\sqrt{H}\bm{B}(x_{it})\sqrt{H}\bm{B}(x_{it})^\top e_{h'}]\leq CH,
	\end{equation}
	as all eigenvalues of $\mathbb{E}[\sqrt{H}\bm{B}(x_{it})\sqrt{H}\bm{B}(x_{it})^\top]$ are bounded. Note that $|V_{it}^{h,h'}|\leq|e_h^\top\sqrt{H}\bm{B}(x_{it})|\cdot|e_{h'}^\top\sqrt{H}\bm{B}(x_{it})|\leq H$. By the $\alpha$-mixing property of $x_{it}$, we know that $V_{it}^{h,h'}$ is also $\alpha$-mixing with mixing coefficients bounded by those of $x_{it}$. By Lemma \ref{lemma:1}, for any fixed $\epsilon>0$,
	\begin{equation}
		\mathbb{P}\left(\left|\frac{1}{T}\sum_{t=1}^TV_{it}^{h,h'}-\mathbb{E}V_{it}^{h,h'}\right|\geq\frac{\epsilon}{H}\right)\leq C\log(T)\exp\left[-C\frac{T}{\log(T)}\frac{\epsilon^2/H^2}{\epsilon+CH}\right]
	\end{equation}
	Taking a union bound for all $1\leq h,h'\leq H$, as $H^3\log(T)/T\to0$,
	\begin{equation}
		\begin{split}
			&\mathbb{P}\left(\max_{1\leq h,h'\leq H}\left|\frac{1}{T}\sum_{t=1}^TV_{it}^{h,h'}-\mathbb{E}V_{it}^{h,h'}\right|\geq\frac{\epsilon}{H}\right)\\
			\leq& C\exp\left[2\log(H)+\log\log(T)-\frac{\epsilon^2T/\log(T)}{CH^3}\right]\to0.
		\end{split}
	\end{equation}

\end{proof}

\begin{lemma} \label{lemma:4}
	For any positive sequence $d(n,T)$ depending on $n$ and $T$,
	\begin{equation}\begin{split}
			\inf_{\|\bm{\vartheta}_i-\bm{\vartheta}_{0i}\|_2=d(n,T)}&\sum_{t=1}^T\mathbb{E}\left[\rho_\tau(y_{it}-\widetilde{\bm{\Pi}}(x_{it})^\top\bm{\vartheta}_i)\Big|x_{it}\right]
			-\sum_{t=1}^T\mathbb{E}\left[\rho_\tau(y_{it}-\widetilde{\bm{\Pi}}(x_{it})^\top\bm{\vartheta}_{0i})\Big|x_{it}\right]\\
			\geq &CTd^2(n,T)
	\end{split}\end{equation}
	with probability approaching 1.
\end{lemma}

\begin{proof}[\textbf{Proof of Lemma \ref{lemma:4}}]
	
	For convenience of notation, denote $m_{it}=m(x_{it})$.
	Using the Knight's identity, namely $\rho_\tau(x-y)-\rho_\tau(x)=-y(\tau-I(x\leq 0))+\int_0^y(I(x\leq t)-I(x\leq0))dt$, and mean value expansion, we have that, for each $1\leq i\leq n$ and $\widetilde{z}\in[\bm{\Pi}(x_{it})^\top\bm{\theta}_{0(k)}-m_{it},\bm{\Pi}(x_{it})^\top\bm{\theta}_{(k)}+\mu_i-\mu_{0i}-m_{it}]$,
	\begin{equation}
		\begin{split}
			&\sum_{t=1}^T\mathbb{E}[\rho_\tau(e_{it}+m_{it}+\mu_{0i}-\widetilde{\bm{\Pi}}(x_{it})^\top\bm{\vartheta}_{i})|x_{it}]\\
			-&\sum_{t=1}^T\mathbb{E}[\rho_\tau(e_{it}+m_{it}+\mu_{0i}-\widetilde{\bm{\Pi}}(x_{it})^\top\bm{\vartheta}_{0i})|x_{it}]\\
			=&\sum_{t=1}^T\int_{\bm{\Pi}(x_{it})^\top\bm{\theta}_{0i}-m_{it}}^{\bm{\Pi}(x_{it})^\top\bm{\theta}_{i}+\mu_i-\mu_{0i}-m_{it}} [F_k(z|x_{it})-F_k(0|x_{it})]dz\\
			=&\sum_{t=1}^T\int_{\bm{\Pi}(x_{it})^\top\bm{\theta}_{0i}-m_{it}}^{\bm{\Pi}(x_{it})^\top\bm{\theta}_i+\mu_i-\mu_{0i}-m_{it}}\left[zf_k(0|x_{it})+\frac{z^2}{2}f_i'(\widetilde{z}|x_{it})\right]dz\\
			\geq& \frac{1}{2}\sum_{t=1}^Tf_i(0|x_{it})\left[(\widetilde{\bm{\Pi}}(x_{it})^\top(\bm{\vartheta}_{i}-\bm{\vartheta}_{0i}))^2+2\widetilde{\bm{\Pi}}(x_{it})^\top(\bm{\vartheta}_{i}-\bm{\vartheta}_{0i})R_{it}\right]\\
			-&\frac{\overline{f'}}{6}\sum_{t=1}^T|(R_{it}+\widetilde{\bm{\Pi}}(x_{it})^\top(\bm{\vartheta}_i-\bm{\vartheta}_{0i}))^3-R_{it}^3|,
		\end{split}
	\end{equation}
	where $R_{it}=m_{it}-\bm{\Pi}(x_{it})^\top\bm{\theta}_0$.
	
	By the property of B-splines, we have $|R_{it}|=O(H^{-d})$. By Cauchy's inequality and Lemma \ref{lemma:3},
	\begin{equation}
		\begin{split}
			& \sum_{t=1}^T\widetilde{\bm{\Pi}}(x_{it})^\top(\bm{\vartheta}_i-\bm{\vartheta}_{0i})R_{it}\\
			\leq & \left[\sum_{t=1}^T\left(\widetilde{\bm{\Pi}}(x_{it})^\top(\bm{\vartheta}_i-\bm{\vartheta}_{0i})\right)^2\right]^{1/2}\left[\sum_{t=1}^TR_{it}^2\right]^{1/2}\\
			=&\left[(\bm{\vartheta}_i-\bm{\vartheta}_{0i})^\top\left(\sum_{t=1}^T\widetilde{\bm{\Pi}}(x_{it})\widetilde{\bm{\Pi}}(x_{it})^\top\right)(\bm{\vartheta}_i-\bm{\vartheta}_{0i})\right]^{1/2}\left[\sum_{t=1}^TR_{it}^2\right]^{1/2}\\
			= & Cd(n,T)TH^{-d}.
		\end{split}
	\end{equation}
	
	By Lemma \ref{lemma:3}, we have
	\begin{equation}
		\sum_{t=1}^T\left(\widetilde{\bm{\Pi}}(x_{it})^\top(\bm{\vartheta}_{i}-\bm{\vartheta}_{0i})\right)^2
		\asymp T\|\bm{\vartheta}_{i}-\bm{\vartheta}_{0i}\|_2^2=Td^2(n,T).
	\end{equation}
	Since $f_k(0|x_{it})\geq \underline{f}$, we have that
	\begin{equation}
		\sum_{t=1}^T|(R_{it}+\widetilde{\bm{\Pi}}(x_{it})^\top(\bm{\vartheta}_{i}-\bm{\vartheta}_{0i}))^3-R_{it}^3|
		= O_p\left(T[\sqrt{H}d(n,T)]^3\right)=o_p(d^2(n,T)),
	\end{equation}
	and with probability approaching one,
	\begin{equation}
		\begin{split}
			\sum_{t=1}^T\mathbb{E}\left[\rho_\tau(y_{it}-\widetilde{\bm{\Pi}}(x_{it})^\top\bm{\vartheta}_{i})\Big|x_{it}\right]-\mathbb{E}\left[\rho_\tau(y_{it}-\widetilde{\bm{\Pi}}(x_{it})^\top\bm{\vartheta}_{0i})\Big|x_{it}\right]\geq CTd^2(n,T).
		\end{split}
	\end{equation}
	
\end{proof}

\begin{lemma}
	\label{lemma:5}
	Under the conditions of Theorem \ref{thm:1}, for any constant $L>0$ and any sequence $d(n,T)$ such that $d(n,T)\geq C\sqrt{H/T}$,
	\begin{equation}
		\sup_{\|\bm{\vartheta}_i-\bm{\vartheta}_{0i}\|_2=Ld(n,T)}\sum_{t=1}^T\widetilde{\bm{\Pi}}(x_{it})^\top(\bm{\vartheta}_{i}-\bm{\vartheta}_{0i})(\tau-I\{e_{it}\leq0\})=L\cdot O_p(Td^2(n,T)).
	\end{equation}
\end{lemma}

\begin{proof}[\textbf{Proof of Lemma \ref{lemma:5}}]
	
	The proof is straightforward using that
	\begin{equation}
		\mathbb{E}\left[\left\|\sum_{t=1}^T\widetilde{\Pi}(x_{it})(\tau-I\{e_{it}\leq0\})\right\|_2^2\right]=O_p(TH).
	\end{equation}
	By Markov's inequality, it is easy to check that
	\begin{equation}
		\begin{split}
			& \sup_{\|\bm{\vartheta}_{i}-\bm{\vartheta}_{0i}\|_2=Ld(n,T)}\sum_{t=1}^T\widetilde{\bm{\Pi}}(x_{it})^\top(\bm{\vartheta}_i-\bm{\vartheta}_{0i})(\tau-I\{e_{it}\leq0\})\\
			=&L\cdot O_p(\sqrt{TH}d(n,T))=L\cdot O_p(Td^2(n,T)).
		\end{split}
	\end{equation}
	
\end{proof}

\begin{lemma}\label{lemma:asymp_represent}
	
	Under the conditions of Theorem \ref{thm:1}, we have the following asymptotic representations of the oracle estimator
	\begin{equation}
		\begin{split}
			&\widehat{\mu}_i-\mu_{0i} + o_p(|\widehat{\mu}_i-\mu_{0i}|)\\
			= & -\bm{\gamma}_i^\top(\widehat{\bm{\theta}}-\bm{\theta}_0)+f_i(0)^{-1}\left\{\mathbb{H}_i^{(1)}(\widehat{\mu}_i,\widehat{\bm{\theta}}) - \mathbb{H}_i^{(1)}(\mu_{0i},\bm{\theta}_0) - H_i^{(1)}(\widehat{\mu}_i,\widehat{\bm{\theta}})\right\}\\
			+&f_i(0)^{-1}\mathbb{H}_i^{(1)}(\mu_{0i},\bm{\theta}_0)+O_p(T^{-1}\vee H^{-d}\vee\|\widehat{\bm{\theta}}-\bm{\theta}_0\|_2^2),
		\end{split}
	\end{equation}
	for all $i=1,2,\dots,n$, and
	\begin{equation}
		\begin{split}
			& \widehat{\bm{\theta}}-\bm{\theta}_0 + o_p(\|\widehat{\bm{\theta}}-\bm{\theta}_0\|_2)\\
			= & \bm{\Gamma}^{-1}\left[-\frac{1}{n}\sum_{i=1}^n\mathbb{H}_i^{(1)}(\mu_{0i},\bm{\theta}_0)\bm{\gamma}_i + \mathbb{H}^{(2)}(\bm{\mu}_0,\bm{\theta}_0)\right]\\
			- & \bm{\Gamma}^{-1} \left[\frac{1}{n}\sum_{i=1}^n\left\{\mathbb{H}_i^{(1)}(\widehat{\mu}_i,\widehat{\bm{\theta}}) - \mathbb{H}_i^{(1)}(\mu_{0i},\bm{\theta}_0) - H_i^{(1)}(\widehat{\mu}_i,\widehat{\bm{\theta}})\right\}\bm{\gamma}_i\right]\\
			+ & \bm{\Gamma}^{-1}\left[\mathbb{H}^{(2)}(\widehat{\bm{\mu}},\widehat{\bm{\theta}}) - \mathbb{H}^{(2)}(\bm{\mu}_0,\bm{\theta}_0) - H^{(2)}(\widehat{\bm{\mu}},\widehat{\bm{\theta}})\right]\\
			+ & O_p\left(T^{-1}H^{1/2}\vee H^{-d}\vee \max_{1\leq i\leq n}|\widehat{\mu}_i-\mu_{0i}|^2\right).
		\end{split}
	\end{equation}
	
\end{lemma}

\begin{proof}[\textbf{Proof of Lemma \ref{lemma:asymp_represent}}]
	
	By the computational property of the QR estimator \citep{kato2012asymptotics}, it is shown that $\max_{1\leq i\leq n}|\mathbb{H}_i^{(1)}(\widehat{\mu}_i,\widehat{\bm{\theta}})|=O_p(T^{-1})$. Thus, uniformly over $1\leq i\leq n$, we have
	\begin{equation}
		O_p(T^{-1}) = \mathbb{H}_i^{(1)}(\mu_{0i},\bm{\theta}_0) + H_i^{(1)}(\widehat{\mu}_i,\widehat{\bm{\theta}})+ \left\{\mathbb{H}_i^{(1)}(\widehat{\mu}_i,\widehat{\bm{\theta}}) - \mathbb{H}_i^{(1)}(\mu_{0i},\bm{\theta}_0) - H_i^{(1)}(\widehat{\mu}_i,\widehat{\bm{\theta}})\right\}.
	\end{equation}
	Expanding $H_i^{(1)}(\widehat{\mu}_i,\widehat{\bm{\theta}})$ around $(\mu_{0i},\bm{\theta}_0)$, we have
	\begin{equation}
		\begin{split}
			H_i^{(1)}(\widehat{\mu}_i,\widehat{\bm{\theta}}) & = -f_i(0)(\widehat{\mu}_i-\mu_{0i}) - f_i(0)\bm{\gamma}_i^\top(\widehat{\bm{\theta}}-\bm{\theta}_{0})\\
			& + O_p(H^{-d}\vee\max_{1\leq i\leq n}|\widehat{\mu}_i-\mu_{0i}|^2\vee\|\widehat{\bm{\theta}}-\bm{\theta}_0\|_2^2),
		\end{split}
	\end{equation}
	and hence, for all $1\leq i\leq n$,
	\begin{equation}\label{eq:mu_dif}
		\begin{split}
			\widehat{\mu}_i-\mu_{0i} & = -\bm{\gamma}_i^\top(\widehat{\bm{\theta}}-\bm{\theta}_0)+f_i(0)^{-1}\left\{\mathbb{H}_i^{(1)}(\widehat{\mu}_i,\widehat{\bm{\theta}}) - \mathbb{H}_i^{(1)}(\mu_{0i},\bm{\theta}_0) - H_i^{(1)}(\widehat{\mu}_i,\widehat{\bm{\theta}})\right\}\\
			&+f_i(0)^{-1}\mathbb{H}_i^{(1)}(\mu_{0i},\bm{\theta}_0)+O_p(T^{-1}\vee H^{-d}\vee\max_{1\leq i\leq n}|\widehat{\mu}_i-\mu_{0i}|^2\vee\|\widehat{\bm{\theta}}-\bm{\theta}_0\|_2^2).
		\end{split}
	\end{equation}
	Similarly, we have $\|\mathbb{H}^{(2)}(\widehat{\bm{\mu}},\widehat{\bm{\theta}})\|_2=O_p(T^{-1}\max_{1\leq i\leq n,1\leq t\leq T}\|\bm{\Pi}(x_{it})\|_2)=O_p(T^{-1}H^{1/2})$, and
	\begin{equation}\label{eq:H2_O}
		\begin{split}
			O_p(T^{-1}H^{1/2}) & = \mathbb{H}^{(2)}(\bm{\mu}_0,\bm{\theta}_0) + H^{(2)}(\widehat{\bm{\mu}},\widehat{\bm{\theta}})+\left\{\mathbb{H}^{(2)}(\widehat{\bm{\mu}},\widehat{\bm{\theta}}) - \mathbb{H}^{(2)}(\bm{\mu}_0,\bm{\theta}_0) - H^{(2)}(\widehat{\bm{\mu}},\widehat{\bm{\theta}})\right\}.
		\end{split}
	\end{equation}
	Expanding $H^{(2)}(\widehat{\bm{\mu}},\widehat{\bm{\theta}})$ around $(\bm{\mu}_{0},\bm{\theta}_0)$, we have
	\begin{equation}\label{eq:H2_exp}
		\begin{split}
			H^{(2)}&(\widehat{\bm{\mu}},\widehat{\bm{\theta}}) =-\frac{1}{n}\sum_{i=1}^n\mathbb{E}[f_i(0|x_{it})\bm{\Pi}(x_{it})\bm{\Pi}(x_{it})^\top](\widehat{\bm{\theta}}-\bm{\theta}_0)\\
			&-\frac{1}{n} \sum_{i=1}^n\mathbb{E}[f_i(0|x_{it})\bm{\Pi}(x_{it})](\widehat{\mu}_i-\mu_{0i})+o_p(\|\widehat{\bm{\theta}}-\bm{\theta}_0\|_2)+O_p(\max_{1\leq i\leq n}|\widehat{\mu}_i-\mu_{0i}|^2).
		\end{split}
	\end{equation}
	By plugging \eqref{eq:mu_dif} into \eqref{eq:H2_exp}, we have
	\begin{equation}\label{eq:H2_result}
		\begin{split}
			H^{(2)}(\widehat{\bm{\mu}},\widehat{\bm{\theta}}) = & -\bm{\Gamma}(\widehat{\bm{\theta}}-\bm{\theta}_0)-\frac{1}{n}\sum_{i=1}^n\mathbb{H}_i^{(1)}(\mu_{0i},\bm{\theta}_0)\bm{\gamma}_i\\
			&-\frac{1}{n}\sum_{i=1}^n\left\{\mathbb{H}_i^{(1)}(\widehat{\mu}_i,\widehat{\bm{\theta}}) - \mathbb{H}_i^{(1)}(\mu_{0i},\bm{\theta}_0) - H_i^{(1)}(\widehat{\mu}_i,\widehat{\bm{\theta}})\right\}\bm{\gamma}_i\\
			&+o_p(\|\widehat{\bm{\theta}}-\bm{\theta}_0\|_2)
			+O_p(T^{-1}\vee H^{-d}\vee\max_{1\leq i\leq n}|\widehat{\mu}_i-\mu_{0i}|^2).
		\end{split}
	\end{equation}
	Combining \eqref{eq:H2_O} and \eqref{eq:H2_result}, we can obtain
	\begin{equation}
		\begin{split}
			\bm{\Gamma}(\widehat{\bm{\theta}}-\bm{\theta}_0)
			= & -\frac{1}{n}\sum_{i=1}^n\mathbb{H}_i^{(1)}(\mu_{0i},\bm{\theta}_0)\bm{\gamma}_i + \mathbb{H}^{(2)}(\bm{\mu}_0,\bm{\theta}_0)\\
			& - \frac{1}{n}\sum_{i=1}^n\left\{\mathbb{H}_i^{(1)}(\widehat{\mu}_i,\widehat{\bm{\theta}}) - \mathbb{H}_i^{(1)}(\mu_{0i},\bm{\theta}_0) - H_i^{(1)}(\widehat{\mu}_i,\widehat{\bm{\theta}})\right\}\bm{\gamma}_i\\
			& +\left\{\mathbb{H}^{(2)}(\widehat{\bm{\mu}},\widehat{\bm{\theta}}) - \mathbb{H}^{(2)}(\bm{\mu}_0,\bm{\theta}_0) - H^{(2)}(\widehat{\bm{\mu}},\widehat{\bm{\theta}})\right\}\\
			& + O_p(T^{-1}H^{1/2}\vee H^{-d}\vee \max_{1\leq i\leq n}|\widehat{\mu}_i-\mu_{0i}|^2)+o_p(\|\widehat{\bm{\theta}}-\bm{\theta}_0\|_2),
		\end{split}
	\end{equation}
	which completes the proof.
	
\end{proof}

\begin{lemma}\label{lemma:ep_rates}
	Take $\delta$ such that $\delta\sqrt{H}\to0$ and $\max_{1\leq i\leq n}|\widehat{\mu}_i-\mu_{0i}|\vee\|\widehat{\bm{\theta}}-\bm{\theta}_0\|_2=O_p(\delta)$. We have
	\begin{equation}
		\left\|\frac{1}{n}\sum_{i=1}^n\bm{\gamma}_i\left\{\mathbb{H}_i^{(1)}(\widehat{\mu}_i,\widehat{\bm{\theta}})-H_i^{(1)}(\widehat{\mu}_i,\widehat{\bm{\theta}})-\mathbb{H}_i^{(1)}(\mu_{0i},\bm{\theta}_0)\right\}\right\|_2=O_p(\sqrt{H}d(T,\delta)\vee H^{-d})
	\end{equation}
	and
	\begin{equation}
		\left\|\mathbb{H}^{(2)}(\widehat{\mu}_i,\widehat{\bm{\theta}})-H^{(2)}(\widehat{\mu}_i,\widehat{\bm{\theta}})-\mathbb{H}^{(2)}(\mu_{0i},\bm{\theta}_0)\right\|_2=O_p(\sqrt{H}d(T,\delta)\vee H^{-d}),
	\end{equation}
	where $d(T,\delta):=[T^{-1}|\log(\sqrt{H}\delta)|]\vee [T^{-1/2}H^{1/4}\delta^{1/2}|\log(\sqrt{H}\delta)|^{1/2}]$.
\end{lemma}

\begin{proof}[\textbf{Proof of Lemma \ref{lemma:ep_rates}}]
	
	We focus on the proof of the first statement since the proof of the second one is analogous. Without loss of generality, we assume that $\mu_{0i}=0$ and $\bm{\theta}_0=\bm{0}$. Let $g_{\mu,\bm{\theta}}(u,\bm{x}):=I(u\leq\mu+\bm{x}^\top\bm{\theta})-I(u\leq0)$ and $\mathcal{G}_\delta:=\{g_{\mu,\bm{\theta}}:|\mu|\leq\delta,\|\bm{\theta}\|_2\leq\delta\}$ and $\xi_{it}=(u_{it},\bm{\Pi}(x_{it}))$. 
	
	As $|m(x_{it})-\bm{\Pi}(x_{it})^\top\bm{\theta}_{0i}|\leq H^{-d}$ and $\|\bm{\gamma}_i\|_2\leq\sqrt{H}$ over $i=1,\dots,n$, it suffices to show
	\begin{equation}
		\max_{1\leq i\leq n}\mathbb{E}\left[\frac{1}{T}\sup_{g\in\mathcal{G}_\delta}\left|\sum_{t=1}^T\{g(\xi_{it})-\mathbb{E}g(\xi_{it})\}\right|\right]=O(d(T,\delta)).
	\end{equation}
	Denote $\widetilde{\mathcal{G}}_{i,\delta}:=\{g-\mathbb{E}[g(\xi_{it})]:g\in\mathcal{G}_\delta\}$. Note that $\widetilde{\mathcal{G}}_{i,\delta}$ is pointwise measurable and each element is bounded by 2. By Lemmas 2.6.15 and 2.6.18 of \citet{vanweak}, the class $\mathcal{G}_\infty$ is a VC subgraph class. By Theorem 2.6.7 of \citet{vanweak}, there exists a constant $v>1$ such that the covering number satisfies $N(\widetilde{\mathcal{G}}_{i,\delta},L_2(Q),2\epsilon)\leq C\epsilon^{-v}$ for any $0<\epsilon<1$ and any probability measure $Q$ on $\mathbb{R}^{H}$. In addition, as $\mathbb{E}[g_{\mu,\bm{\theta}}(\xi_{it})^2]=\mathbb{E}[|F_i(\mu+\bm{\Pi}(x_{it})^\top\bm{\theta}|x_{it})-F_i(0|x_{it})]\leq C(|\mu|+\sqrt{H}\|\bm{\theta}\|_2)\leq C\sqrt{H}\delta$. By the Bernstein-type inequality for bounded empirical process, e.g., Proposition B.1 in \citet{kato2012asymptotics}, we obtain the desired result.
	
\end{proof}

\subsection{Proofs of Theorems 2 and 3}\label{append:B}

In this subsection, we present the proofs of Theorems \ref{thm:2} and \ref{thm:3} and relegate some auxiliary lemmas to the end of this subsection.
For the brevity of notation, we simplity $\sum_{i=1}^n\sum_{t=1}^T$ to $\sum_{i,t}$.

\begin{proof}[\textbf{Proof of Theorem 2}]
	
	We define the oracle estimator to be that obtained from \eqref{eq:oracle} assuming the groups are known and thus $\bm{\widehat{\theta}}^\text{o}_{(k)}$ is obtained from only observations in $G_k$, separately for different groups. Similarly to the proof of Theorem \ref{thm:2}, we denote $\xi(n,T)=\sqrt{H/(nT)}+H^{-d}$. It suffices to show that with probability approaching one, the oracle estimator is a local minimizer of the SCAD-penalized quantile regression \eqref{eq:objective}.
	
	Considering any $\bm{\theta}_i$ with $\|\bm{\theta}_i-\bm{\widehat{\theta}}_i^\text{o}\|_2\leq c$ for all $1\leq i\leq n$, with $c$ sufficiently small, specifically $c=o(\lambda)$, and $(\mu_1,\dots,\mu_n)$ with $\max_{1\leq i\leq n}|\mu_{i}-\widehat{\mu}_{i}^\text{o}|\leq d$, with $d$ sufficiently small. We only need to show that uniformly over $\bm{\theta}_c:=\{\bm{\theta}=(\bm{\theta}_1^\top,\dots,\bm{\theta}_n^\top)^\top:\|\bm{\theta}_i-\bm{\widehat{\theta}}_i^\text{o}\|_2\leq c,~\forall i\}$ and $\bm{\mu}_d:=\{(\mu_1,\dots,\mu_n)^\top:\max_{1\leq i\leq n}|\mu_i-\widehat{\mu}_i|\leq d\}$
	\begin{equation}
		\begin{split}
			&\frac{1}{nT}\sum_{i,t}\rho_\tau(y_{it}-\mu_i-\bm{\Pi}(x_{it})^\top\bm{\theta}_i) + \binom{n}{2}^{-1}\sum_{i<j}p_\lambda(\|\bm{\theta}_i-\bm{\theta}_j\|_2) \\
			\geq& \frac{1}{nT}\sum_{i,t}\rho_\tau(y_{it}-\widehat{\mu}_i^\text{o}-\bm{\Pi}(x_{it})^\top\bm{\widehat{\theta}}_i^\text{o})+\binom{n}{2}^{-1}\sum_{i<j}p_\lambda(\|\bm{\widehat{\theta}}_i^\text{o}-\bm{\widehat{\theta}}_j^\text{o}\|_2).
		\end{split}
	\end{equation}
	
	Let $g_i=k$ if $i\in G_k$. That is, $g_i$ is an indicator on the individual $i$'s group identity. Let $\bm{O}:=\{\bm{\theta}=(\bm{\theta}_1^\top,\dots,\bm{\theta}_n^\top)^\top:\bm{\theta}_i=\bm{\theta}_j~\text{if}~g_i=g_j\}$. That is, $\bm{O}$ consists of all coefficients that satisfy the group partition structure. For ease of presentation, define the mapping $\bm{\Gamma}:\mathbb{R}^{Hn}\to\bm{O}$ with $\Gamma(\bm{\theta})=(\bm{\theta}_1^*,\dots,\bm{\theta}_n^*)$, where $\bm{\theta}_i^*=\sum_{j:g_j=g_i}\bm{\theta}_i/|G_{g_i}|$.
	In other words, $\bm{\Gamma}$ can be the projected value of $\bm{\theta}$ to the space $\bm{O}$.
	
	The proof of the displayed equation above can be achieved by the following two steps.\\
	(a)
	\begin{equation}
		\begin{split}
			&\inf_{\bm{\theta}^*=\bm{\Gamma}(\bm{\theta}),\bm{\theta}\in\bm{\theta}_c,\bm{\mu}\in\bm{\mu}_d}
			\frac{1}{nT}\sum_{i,t}\rho_\tau(y_{it}-\mu_i-\bm{\Pi}(x_{it})^\top\bm{\theta}_i^*)+\binom{n}{2}^{-1}\sum_{i<j}p_\lambda(\|\bm{\theta}_i^*-\bm{\theta}_j^*\|_2)\\
			\geq&\frac{1}{nT}\sum_{i,t}\rho_\tau(y_{it}-\widehat{\mu}_i^\text{o}-\bm{\Pi}(x_{it})^\top\bm{\widehat{\theta}}_i^\text{o}) +\binom{n}{2}^{-1} \sum_{i<j}p_\lambda(\|\bm{\widehat{\theta}}_i^\text{o}-\bm{\widehat{\theta}}_j^\text{o}\|_2).
		\end{split}
	\end{equation}
	(b)
	\begin{equation}
		\begin{split}
			&\inf_{\bm{\theta}^*=\bm{\Gamma}(\bm{\theta}),\bm{\theta}\in\bm{\theta}_c,\bm{\mu}\in\bm{\mu}_d}
			\frac{1}{nT}\sum_{i,t}\rho_\tau(y_{it}-\mu_i-\bm{\Pi}(x_{it})^\top\bm{\theta}_i)+\binom{n}{2}^{-1}\sum_{i<j}p_\lambda(\|\bm{\theta}_i-\bm{\theta}_j\|_2)\\
			-&
			\frac{1}{nT}\sum_{i,t}\rho_\tau(y_{it}-\mu_i-\bm{\Pi}(x_{it})^\top\bm{\theta}_i^*)-\binom{n}{2}^{-1}\sum_{i<j}p_\lambda(\|\bm{\theta}_i^*-\bm{\theta}_j^*\|_2)\geq0.
		\end{split}
	\end{equation}
	
	For (a), by the definition of the local minimizer which minimizes the check loss subject to the grouping constraint, we have
	\begin{equation}\label{eq:minimizer}
		\inf_{\bm{\theta}^*=\bm{\Gamma}(\bm{\theta}),\bm{\theta}\in\bm{\theta}_c,\bm{\mu}\in\bm{\mu}_d}\frac{1}{nT}\sum_{i,t}\rho_\tau(y_{i,t}-\mu_i-\bm{\Pi}(x_{it})^\top\bm{\theta}_i^*) \geq \frac{1}{nT}\sum_{i,t}\rho_\tau(y_{it}-\widehat{\mu}_i^\text{o}-\bm{\Pi}(x_{it})^\top\bm{\widehat{\theta}}_i^\text{o}).
	\end{equation}
	If $g_i\neq g_j$, by our assumptions, we have $\lambda=o(\|\bm{\theta}_{0i}-\bm{\theta}_{0j}\|_2)$ and $\xi(n,T)=o(\lambda)$.
	Thus,
	\begin{equation}
		\label{eq:oracle_dif}
		\begin{split}
			&\|\bm{\widehat{\theta}}_i^\text{o}-\bm{\widehat{\theta}}_j^\text{o}\|_2\geq\|\bm{\theta}_{0i}-\bm{\theta}_{0j}\|_2-\|\bm{\widehat{\theta}}_i^\text{o}-\bm{\theta}_{0i}\|_2 - \|\bm{\widehat{\theta}}_j^\text{o}-\bm{\theta}_{0j}\|_2\\
			\geq & 3a\lambda-o_p(\lambda)\geq 2a\lambda.
		\end{split}
	\end{equation}
	In addition,
	\begin{equation}
		\|\bm{\theta}_i^*-\bm{\widehat{\theta}}_i^\text{o}\|_2 = \left\| \sum_{k:g_k=g_i} \bm{\theta}_k/|G_k|-\bm{\widehat{\theta}}_i^\text{o}\right\|_2\leq \max_{k:g_k=g_i}\|\bm{\theta}_k-\bm{\widehat{\theta}}_i^\text{o}\|_2\leq c,
	\end{equation}
	which implies that
	\begin{equation}
		\begin{split}
			&\|\bm{\theta}_i^*-\bm{\theta}_j^*\|_2\geq \|\bm{\widehat{\theta}}_i^\text{o}-\bm{\widehat{\theta}}_j^\text{o}\|_2 - \|\bm{\theta}^*_i- \bm{\widehat{\theta}}_i^\text{o}\|_2 - \|\bm{\theta}_j^*-\bm{\widehat{\theta}}_j^\text{o}\|_2 \geq 2a\lambda-2c \geq a\lambda.
		\end{split}
	\end{equation}
	
	Thus, by the definition of SCAD penalty function,
	\begin{equation}
		\label{eq:notequal}
		p_\lambda(\|\bm{\theta}_i^*-\bm{\theta}_j^*\|_2) = p_\lambda(\|\bm{\widehat{\theta}}_i^\text{o}-\bm{\widehat{\theta}}_j^\text{o}\|_2)=\frac{(a+1)\lambda^2}{2},~~\text{if}~g_i\neq g_j.
	\end{equation}
	On the other hand, if $g_i=g_j$, then $\bm{\widehat{\theta}}_i^\text{o}=\bm{\widehat{\theta}}_j^\text{o}$ and $\bm{\theta}_i^*=\bm{\theta}_j^*$ and thus we have
	\begin{equation}
		\label{eq:equal}
		p_\lambda(\|\bm{\theta}_i^*-\bm{\theta}_j^*\|_2) = p_\lambda(\|\bm{\widehat{\theta}}_i^\text{o}-\bm{\widehat{\theta}}_j^\text{o}\|_2)=0,~~\text{if}~g_i=g_j.
	\end{equation}
	Combining these two cases \eqref{eq:notequal} and \eqref{eq:equal}, as well as \eqref{eq:minimizer}, we proved (a).
	
	In the rest of the proof we will show (b). Using the convexity of the check loss function, we have $\rho_\tau(x)-\rho_\tau(y)\geq(\tau-I\{y\leq0\})(x-y)$. Thus for the difference of the loss terms, we have
	\begin{equation}
		\label{eq:dif_loss}
		\begin{split}
			&\sum_{t=1}^T\rho_\tau(y_{it}-\mu_i-\bm{\Pi}(x_{it})^\top\bm{\theta}_i)-\sum_{t=1}^T\rho_\tau(y_{it}-\mu_i-\bm{\Pi}(x_{it})^\top\bm{\theta}_i^*)\\
			\geq&-\sum_{t=1}^T(\tau-1\{y_{it}\leq\bm{\Pi}(x_{it})^\top\bm{\theta}_i+\mu_i\})\bm{\Pi}(x_{it})^\top(\bm{\theta}_i-\bm{\theta}_i^*)\\
			=&-\sum_{t=1}^T(\tau-1\{e_{it}\leq0\})\bm{\Pi}(x_{it})^\top(\bm{\theta}_i-\bm{\theta}_i^*)\\
			&-\sum_{t=1}^T(1\{e_{it}\leq0\}-1\{e_{it}\leq\bm{\Pi}(x_{it})^\top\bm{\theta}_i-m_{it}+\mu_i-\mu_{0i}\})\bm{\Pi}(x_{it})^\top(\bm{\theta}_i-\bm{\theta}_i^*).
		\end{split}
	\end{equation}
	For the first term, using Bernstein's inequality in Lemma \ref{lemma:1} in subsection \ref{append:A}, we have
	\begin{equation}
		\max_{1\leq h\leq H,1\leq i\leq n}\sum_{t=1}^T(\tau-1\{e_{it}\leq0\})\bm{\Pi}_h(x_{it})=O_p(\sqrt{T\log(T)\log(nH\log(T))})
	\end{equation}
	\begin{equation}
		\text{ and thus, }\max_{1\leq i\leq n}\left\|\sum_{t=1}^T(\tau-1\{e_{it}\leq0\})\bm{\Pi}(x_{it})\right\|_2=O_p(\sqrt{TH\log(T)\log(nH\log(T))}).
	\end{equation}
	By Lemma \ref{lemma:6}, for sufficiently small $c$ and $d$, we have
	\begin{equation}
		\begin{split}
			&\sup_{\substack{1\leq i\leq n,\|\bm{\theta}_i-\bm{\widehat{\theta}}_i^\text{o}\|_2\leq c\\|\mu_i-\widehat{\mu}_i^\text{o}|\leq d}}\left\|\sum_{t=1}^T(1\{e_{it}\leq0\}-1\{e_{it}\leq \bm{\Pi}(x_{it})^\top\bm{\theta}_i-m_{it}+\mu_i-\mu_{0i}\})\bm{\Pi}(x_{it})\right\|_2\\
			\leq&\sup_{\substack{1\leq i\leq n,\|\bm{\theta}_i-\bm{\widehat{\theta}}_i^\text{o}\|_2\leq c\\|\mu_i-\widehat{\mu}_i^\text{o}|\leq d}}\Bigg\|\sum_{t=1}^T(1\{e_{it}\leq0\}-1\{e_{it}\leq \bm{\Pi}(x_{it})^\top\bm{\theta}_i-m_{it}+\mu_i-\mu_{0i}\}\\
			&~~~~~~~~~+F(\bm{\Pi}(x_{it})^\top\bm{\theta}_i-m_{it}+\mu_i-\mu_{0i})-F(0))\bm{\Pi}(x_{it})\Bigg\|_2\\
			+&\sup_{\substack{1\leq i\leq n,\|\bm{\theta}_i-\bm{\widehat{\theta}}_i^\text{o}\|_2\leq c\\|\mu_i-\widehat{\mu}_i^\text{o}|\leq d}}\left\|\sum_{t=1}^T(F(\bm{\Pi}(x_{it})^\top\bm{\theta}_i-m_{it}+\mu_i-\mu_{0i})-F(0))\bm{\Pi}(x_{it})\right\|_2\\
			=&O_p(H^{3/2}T^{1/2}\log(T)\log(nT))+O_p(T\sqrt{H}\xi(n,T))=O_p(T\sqrt{H}\xi(n,T)).
		\end{split}
	\end{equation}
	We denote
	\begin{equation}
		\begin{split}
			\bm{w}_i=&-\sum_{t=1}^T(\tau-1\{e_{it}\leq0\})\bm{\Pi}(x_{it})\\
			&-\sum_{t=1}^T(1\{e_{it}\leq0\}-1\{e_{it}\leq\bm{\Pi}(x_{it})^\top\bm{\theta}_i-m_{it}+\mu_i-\mu_{0i}\})\bm{\Pi}(x_{it}).
		\end{split}
	\end{equation}
	Then, the last line in \eqref{eq:dif_loss}, after summing over $i$, can be written as
	\begin{equation}
		\begin{split}
			&\frac{1}{nT}\sum_{i=1}^n\bm{w}_i^\top(\bm{\theta}_i-\bm{\theta}_i^*)=\frac{1}{nT}\sum_{i=1}^n\bm{w}_i^\top(\bm{\theta}_i-\sum_{j:g_j=g_i}\bm{\theta}_j/|G_{g_i}|)\\
			=&\frac{1}{nT}\sum_{i=1}^n\sum_{j:g_j=g_i}\frac{\bm{w}_i^\top(\bm{\theta}_i-\bm{\theta}_j)}{|G_{g_i}|}=\frac{1}{nT}\sum_{(i,j):i<j\text{ and }g_i=g_j}\frac{(\bm{w}_i-\bm{w}_j)^\top(\bm{\theta}_i-\bm{\theta}_j)}{|G_{g_i}|}\\
			=&O_p(n^{-2}\sqrt{H}\xi(n,T)+n^{-2}\sqrt{(H/T)\log(T)\log(nH\log(T))})\times\left(\sum_{i<j\text{ and }g_i=g_j}\|\bm{\theta}_i-\bm{\theta}_j\|\right).
		\end{split}
	\end{equation}
	
	When $g_i\neq g_j$, by \eqref{eq:oracle_dif},
	\begin{equation}
		\begin{split}
			\|\bm{\theta}_i-\bm{\theta}_j\|_2
			\geq&\|\bm{\widehat{\theta}}_i^\text{o}-\bm{\widehat{\theta}}_j^\text{o}\|_2
			-\|\bm{\widehat{\theta}}_i-\bm{\widehat{\theta}}_i^\text{o}\|_2
			-\|\bm{\widehat{\theta}}_j-\bm{\widehat{\theta}}_j^\text{o}\|_2\geq\|\bm{\widehat{\theta}}_i^\text{o}-\bm{\widehat{\theta}}_j^\text{o}\|_2-2c\geq a\lambda.
		\end{split}
	\end{equation}
	Also, $\|\bm{\theta}_i^*-\bm{\theta}_j^*\|_2\geq a\lambda$, when $g_i\neq g_j$. So the difference of penalty terms of (b) is
	\begin{equation}
		\begin{split}
			&\binom{n}{2}^{-1}\sum_{i<j}p_\lambda(\|\bm{\theta}_i-\bm{\theta}_j\|_2)-\binom{n}{2}^{-1}\sum_{i<j}p_\lambda(\|\bm{\theta}_i^*-\bm{\theta}_j^*\|_2)\\
			=&\binom{n}{2}^{-1}\sum_{i<j,g_i\neq g_j}[p_\lambda(\|\bm{\theta}_i-\bm{\theta}_j\|_2)-p_\lambda(\|\bm{\theta}_i^*-\bm{\theta}_j^*\|_2)]\\
			+&\binom{n}{2}^{-1}\sum_{i<j,g_i= g_j}[p_\lambda(\|\bm{\theta}_i-\bm{\theta}_j\|_2)-p_\lambda(\|\bm{\theta}_i^*-\bm{\theta}_j^*\|_2)]\\
			=&\binom{n}{2}^{-1}\sum_{i<j,g_i=g_j}[p_\lambda(\|\bm{\theta}_i-\bm{\theta}_j\|_2)-p_\lambda(\|\bm{\theta}_i^*-\bm{\theta}_j^*\|_2)].
		\end{split}
	\end{equation}
	
	When $g_i=g_j$, we have $\bm{\theta}_i^*=\bm{\theta}_j^*$. Furthermore,
	\begin{equation}
		\begin{split}
			\|\bm{\theta}_i-\bm{\theta}_j\|_2\leq\|\bm{\widehat{\theta}}_i^
			\text{o}-\bm{\widehat{\theta}}_j^\text{o}\|_2+\|\bm{\theta}_i-\bm{\widehat{\theta}}_i^\text{o}\|_2+\|\bm{\theta}_j-\bm{\widehat{\theta}}_j^\text{o}\|_2\leq 2c\leq \lambda,
		\end{split}
	\end{equation}
	and since the SCAD penalty $p_\lambda(x)=\lambda x$ when $x\in[0,\lambda]$, we have
	\begin{equation}
		\begin{split}
			&\binom{n}{2}^{-1}\sum_{i<j}p_\lambda(\|\bm{\theta}_i-\bm{\theta}_j\|_2)-\binom{n}{2}^{-1}\sum_{i<j}p_\lambda(\|\bm{\theta}_i^*-\bm{\theta}_j^*\|_2)\\
			\geq&\binom{n}{2}^{-1}\sum_{i<j,g_i=g_j}p_\lambda(\|\bm{\theta}_i-\bm{\theta}_j\|_2)=\binom{n}{2}^{-1}\lambda\sum_{i<j,g_i=g_j}\|\bm{\theta}_i-\bm{\theta}_j\|_2.
		\end{split}
	\end{equation}
	Thus, by our assumption, the difference of the penalties is positive and dominant in the left hand side of (b), which implies that (b) holds.
	
\end{proof}

\begin{proof}[\textbf{Proof of Theorem \ref{thm:3}}]
	
	In the proof, we denote the true number of groups by $K_0$ and the true partition is $\mathcal{G}_0=(G_{01},\dots,G_{0n})$ with true group indicators $g_{0i}=k$ if $i\in G_k$. Let $\mathcal{G}=\{G_1,\dots,G_K\}$ be any partition for $\{1,\dots,n\}$ with $K$ groups, with group indicators $g_i,i=1,\dots,n$. Define
	\begin{equation}
		\bm{\theta}^{\mathcal{G}_0}=\{\bm{\theta}^{\mathcal{G}_0}_{01},\dots,\bm{\theta}^{\mathcal{G}_0}_{0n}\}=\min_{\bm{\theta}_i=\bm{\theta}_j\text{ if }g_i=g_j}\mathbb{E}\left[\sum_{i=1}^n\sum_{t=1}^T\rho_\tau(y_{it}-\mu_{0i}-\bm{\Pi}(x_{it})^\top\bm{\theta}_i)\right].
	\end{equation}
	
	We note that $\bm{\theta}_{0i}^{\mathcal{G}_0}$ under the true partition is different from the $\bm{\theta}_{0i}$ we defined previously, as the minimizer of $\mathbb{E}[f_{(k)}(0|x_{it})(\bm{\Pi}(x_{it})^\top\bm{\theta}_i-m_{it})^2]$, where $f_{(k)}(\cdot|x_{it})$ is the average conditional density function for the group $G_k$. However, we first show that they are close enough. By Knight's identity, for any $i\in G_k$,
	\begin{equation}\begin{split}
			&\rho_\tau(y_{it}-\mu_{0i}-\bm{\Pi}(x_{it})^\top\bm{\theta}_i)-\rho_\tau(y_{it}-\mu_{0i}-m_{it})\\
			=&(\bm{\Pi}(x_{it})^\top\bm{\theta}_i-m_{it})[I\{e_{it}\leq0\}-\tau]+\int_0^{\bm{\Pi}(x_{it})^\top\bm{\theta}_i-m_{it}}[I\{e_{it}\leq u\}-I\{e_{it}\leq0\}]du.
	\end{split}\end{equation}
	
	Hence, as $f_i(0|x_{it})=f_{(k)}(0|x_{it})$ for all $i\in G_k$,
	\begin{equation}\begin{split}
			&\mathbb{E}[\rho_\tau(y_{it}-\mu_{0i}-\bm{\Pi}(x_{it})^\top\bm{\theta}_i)]-\mathbb{E}[\rho_\tau(y_{it}-\mu_{0i}-m_{it})]\\
			=&\mathbb{E}\left\{\int_0^{\bm{\Pi}(x_{it})^\top\bm{\theta}-m_{it}}F_i[u|x_{it}]-F_{}[0|x_{it}]du\right\}\\
			=&\mathbb{E}\left[\frac{1}{2}f_{(k)}(0|x_{it})[\bm{\Pi}(x_{it})^\top\bm{\theta}_i-m_{it}]^2\right]+O(\mathbb{E}[|\bm{\Pi}(x_{it})^\top\bm{\theta}_i-m_{it}|^3]).
	\end{split}\end{equation}
	
	So for $\|\bm{\theta}_i-\bm{\theta}_{0i}\|_2\leq M_nH^{-(2/3)d}$ with $M_n\to\infty$ arbitrarily slowly, we have $\mathbb{E}[|\bm{\Pi}(x_{it})^\top\bm{\theta}_i-m_{it}|^3]\leq\mathbb{E}[(|\bm{\Pi}(x_{it})^\top(\bm{\theta}_i-\bm{\theta}_{0i})|+|\bm{\Pi}(x_{it})^\top\bm{\theta}_{0i}-m_{it}|)^3]=O(M_n^3H^{-(9/2)d+3/2}+H^{-3d})=O(H^{-3d})$ and thus
	\begin{equation}\begin{split}
			&\mathbb{E}[\rho_\tau(y_{it}-\mu_{0i}-\bm{\Pi}(x_{it})^\top\bm{\theta}_i)]-\mathbb{E}[\rho_\tau(y_{it}-\mu_{0i}-m_{it})]\\
			=&\mathbb{E}\left[\frac{1}{2}f_{(k)}(0|x_{it})[\bm{\Pi}(x_{it})^\top\bm{\theta}_i-m_{it}]^2\right]+O(H^{-3d}).
	\end{split}\end{equation}
	
	For $\|\bm{\theta}_i-\bm{\theta}_{0i}\|_2=M_nH^{-(2/3)d}$, we have
	\begin{equation}\begin{split}
			&\mathbb{E}[\rho_\tau(y_{it}-\mu_{0i}-\bm{\Pi}(x_{it})^\top\bm{\theta}_{(k)})]-\mathbb{E}[\rho_\tau(y_{it}-\mu_{0i}-\bm{\Pi}(x_{it})^\top\bm{\theta}_{0i})]\\
			=&\mathbb{E}\left[\frac{1}{2}f_{(k)}(0|x_{it})[\bm{\Pi}(x_{it})^\top\bm{\theta}_i-m_{it}]^2\right]-\mathbb{E}\left[\frac{1}{2}f_{(k)}(0|x_{it})[\bm{\Pi}(x_{it})^\top\bm{\theta}_{0i}-m_{it}^2]\right]+O(H^{-3d})\\
			=&\mathbb{E}\left[\frac{1}{2}f_{(k)}(0|x_{it})[\bm{\Pi}(x_{it})^\top(\bm{\theta}_i-\bm{\theta}_{0i})]^2\right]+\mathbb{E}\left[\frac{1}{2}f_{(k)}(0|x_{it})\bm{\Pi}(x_{it})^\top(\bm{\theta}_i-\bm{\theta}_{0i})(\bm{\Pi}(x_{it})^\top\bm{\theta}_{0i}-m_{it})\right]\\
			&+O(H^{-3d})\\
			=&\mathbb{E}[\frac{1}{2}f_{(k)}(0|x_{it})[\bm{\Pi}(x_{it})^\top(\bm{\theta}_i-\bm{\theta}_{0i})]^2]+O(H^{-3d})\\
			\geq & CM_nH^{-3d}-O(H^{-3d})>0,
	\end{split}\end{equation}
	where the third equality above results from $\bm{\theta}_{0i}$ minimizes $\mathbb{E}[\frac{1}{2}f_{(k)}(0|x_{it})(\bm{\Pi}(x_{it})^\top\bm{\theta}_i-m_{it})^2]$, which means $\mathbb{E}[f_{(k)}(0|x_{it})\bm{\Pi}(x_{it})^\top(\bm{\theta}_i-\bm{\theta}_{0i})(\bm{\Pi}(x_{it})^\top\bm{\theta}_{0i}-m_{it})]=0$. This means $\|\bm{\theta}_{0i}^{\mathcal{G}_0}-\bm{\theta}_{0i}\|_2\leq M_nH^{-3d}\leq H^{-d-1/2}$ and thus $\bm{\theta}_{0i}^{\mathcal{G}_0}$ still satisfies the approximation property $\sup_x|\bm{\Pi}(x)^\top\bm{\theta}_{0i}^{\mathcal{G}_0}-m_i(x)|\leq CH^{-d}$.\\
	
	\noindent\textbf{Case 1.} ($K<K_0$, under-fitted model)
	
	In this case, let $\mathcal{G}$ be the partition that minimizes $$\min_{\bm{\theta}_i=\bm{\theta}_j,\text{ if }G_i=G_j, |\mathcal{G}|=K}\mathbb{E}\left[\sum_{i=1}^n\sum_{t=1}^T\rho_\tau(y_{it}-\mu_{0i}-\bm{\Pi}(x_{it})^\top\bm{\theta}_i)\right].$$
	By definition, it is obvious that if $i, j$ belongs to the same group in the true partition $\mathcal{G}_0$ so that the distribution of $(y_{it}, x_{it})$ and $(y_{jt}, x_{jt})$ are the same, they are still in the same group in the partition $\mathcal{G}$. In other words, $\mathcal{G}$ is formed by combining some groups in $\mathcal{G}_0$. In particular, given $K_0$ is fixed, there are only a fixed number of such possible partitions $\mathcal{G}$.
	
	Suppose $G_{0k}$, $G_{0k'}$ are combined into $G_{k''}$, then
	\begin{equation}\begin{split}
			&\sum_{i\in G_{k''}}\mathbb{E}[\rho_\tau(y_{it}-\mu_{0i}-\bm{\Pi}(x_{it})^\top\bm{\theta}_{0i}^{\mathcal{G}})]-\sum_{i\in G_{k''}}\mathbb{E}[\rho_\tau(y_{it}-\mu_{0i}-\bm{\Pi}(x_{it})^\top\bm{\theta}_{0i}^{\mathcal{G}_0})]\\
			=&\sum_{i\in G_{k''}}\mathbb{E}\int_{\bm{\Pi}(x_{it})^\top\bm{\theta}_{0i}^{\mathcal{G}_0}-m_{it}}^{\bm{\Pi}(x_{it})^\top\bm{\theta}_{0i}^{\mathcal{G}}-m_{it}}[I\{e_{it}\leq u\}-I\{e_{it}\leq 0\}]du\geq\sum_{i\in G_{k''}}\|\bm{\theta}_{0i}^{\mathcal{G}}-\bm{\theta}_{0i}^{\mathcal{G}_0}\|_2^2\geq Cn_{k''}\rho
	\end{split}\end{equation}
	where at least one of the distance $\|\bm{\theta}_{0i}^{\mathcal{G}}-\bm{\theta}_{0i}^{\mathcal{G}_0}\|_2$ for $i\in G_{0k}$ and $\|\bm{\theta}_{0i}^{\mathcal{G}}-\bm{\theta}_{0i}^{\mathcal{G}_0}\|_2$ for $i\in G_{0k'}$ is larger than, say $\rho/2$. By summing over different groups, we get $$\sum_{i=1}^n\mathbb{E}[\rho_\tau(y_{it}-\mu_{0i}-\bm{\Pi}(x_{it})^\top\bm{\theta}_{0i}^{\mathcal{G}})]-\sum_{i=1}^n\mathbb{E}[\rho_\tau(y_{it}-\mu_{0i}-\bm{\Pi}(x_{it})^\top\bm{\theta}_{0i}^{\mathcal{G}_0})]\geq Cn\rho.$$
	
	By following the proof of Theorem \ref{thm:1}, in particular Lemma \ref{lemma:T1_concentration}, we can show that $\|\bm{\widehat{\vartheta}}_{i}^{\mathcal{G}}-\bm{\vartheta}_{0i}^{\mathcal{G}_0}\|_2=O_p(\sqrt{H/(nT)}+H^{-d})$. Similarly to Lemma \ref{lemma:4}, we have $\mathbb{E}[\sum_{t=1}^T\rho_\tau(y_{it}-\widetilde{\bm{\Pi}}(x_{it})^\top\bm{\widehat{\vartheta}}_{i}^{\mathcal{G}})]-\mathbb{E}[\sum_t\rho_\tau(y_{it}-\bm{\Pi}(x_{it})^\top\bm{\vartheta}_{0i}^{\mathcal{G}_0})]=O_p(T\|\bm{\widehat{\vartheta}}_{i}^{\mathcal{G}}-\bm{\vartheta}_{0i}^{\mathcal{G}_0}\|_2^2)=O_p(nT\xi^2(n,T))$. By the definition of $\bm{\vartheta}_{0i}^{\mathcal{G}}$, we have $\sum_{i,t}\widetilde{\bm{\Pi}}(x_{it})^\top(\bm{\widehat{\vartheta}}_{i}^{\mathcal{G}}-\bm{\vartheta}_{0i}^{\mathcal{G}_0})[I\{e_{it}\leq \widetilde{\bm{\Pi}}(x_{it})^\top\bm{\theta}_{0i}^{\mathcal{G}}-\mu_{0i}-m_{it}\}-\tau]$ has mean zero and thus of order $O_p(nT\xi^2(n,T))$. Thus, $$\Big|\sum_{i,t}\rho_\tau(y_{it}-\bm{\Pi}(x_{it})^\top\bm{\widehat{\theta}}_{i}^{\mathcal{G}})-\sum_{i,t}\rho_\tau(y_{it}-\bm{\Pi}(x_{it})^\top\bm{\theta}_{0i}^{\mathcal{G}})\Big|=O_p(nT\xi^2(n,T)).$$
	
	Note that
	$$\left|\frac{1}{nT}\sum_{i,t}\rho_\tau(y_{it}-\widetilde{\bm{\Pi}}(x_{it})^\top\bm{\vartheta}_{0i}^{\mathcal{G}})-\mathbb{E}\left[\frac{1}{nT}\sum_{i,t}\rho_\tau(y_{it}-\widetilde{\bm{\Pi}}(x_{it})^\top\bm{\vartheta}_{0i}^{\mathcal{G}})\right]\right|=O_p((nT)^{-1/2}),$$
	and
	$$\left|\frac{1}{nT}\sum_{i,t}\rho_\tau(y_{it}-\widetilde{\bm{\Pi}}(x_{it})^\top\bm{\vartheta}_{0i}^{\mathcal{G}_0})-\mathbb{E}\left[\frac{1}{nT}\sum_{i,t}\rho_\tau(y_{it}-\widetilde{\bm{\Pi}}(x_{it})^\top\bm{\vartheta}_{0i}^{\mathcal{G}_0})\right]\right|=O_p((nT)^{-1/2}).$$ 
	So for the SIC, denote $\eta=KH\log(nT)/(nT)$ and we can write
	\begin{equation}\begin{split}
			&\text{SIC}(K)-\text{SIC}(K_0)\\
			=&\log\left(1+\frac{\sum_{i,t}\rho_\tau(y_{it}-\widetilde{\bm{\Pi}}(x_{it})^\top\bm{\widehat{\vartheta}}_{i}^{\mathcal{G}})/(nT)-\sum_{i,t}\rho_\tau(y_{it}-\widetilde{\bm{\Pi}}(x_{it})^\top\bm{\widehat{\vartheta}}_{i}^{\mathcal{G}_0})/(nT)}{\sum_{i,t}\rho_\tau(y_{it}-\widetilde{\bm{\Pi}}(x_{it})^\top\bm{\widehat{\vartheta}}_{i}^{\mathcal{G}_0})/(nT)}\right)+O(\eta)\\
			=&\log\left(1+\frac{\sum_{i,t}\rho_\tau(y_{it}-\widetilde{\bm{\Pi}}(x_{it})^\top\bm{\vartheta}_{0i}^{\mathcal{G}})/(nT)-\sum_{i,t}\rho_\tau(y_{it}-\widetilde{\bm{\Pi}}(x_{it})^\top\bm{\vartheta}_{0i}^{\mathcal{G}_0})/(nT)+O_p(\xi^2)}{\sum_{i,t}\rho_\tau(y_{it}-\widetilde{\bm{\Pi}}(x_{it})^\top\bm{\vartheta}_{0i}^{\mathcal{G}_0})/(nT)+O_p(\xi^2)}\right)+O(\eta)\\
			=&\log\left(1+\frac{\mathbb{E}[\sum_{i,t}\rho_\tau(y_{it}-\widetilde{\bm{\Pi}}(x_{it})^\top\bm{\vartheta}_{0i}^{\mathcal{G}})]-\mathbb{E}[\sum_{i,t}\rho_\tau(y_{it}-\widetilde{\bm{\Pi}}(x_{it})^\top\bm{\vartheta}_{0i}^{\mathcal{G}_0})]+O_p(nT\xi^2+\sqrt{nT})}{\mathbb{E}[\sum_{i,t}\rho_\tau(y_{it}-\widetilde{\bm{\Pi}}(x_{it})^\top\bm{\vartheta}_{0i}^{\mathcal{G}_0})]+O_p(nT\xi^2+\sqrt{nT})}\right)\\
			+&O(\eta)\geq\log(1+C\rho)+O\left(\eta\right)>0.
	\end{split}\end{equation}
	
	\noindent\textbf{Case 2.} ($K_0<K\leq K_{max}$, over-fitted model)
	
	Again, let $\mathcal{G}$ be the partition that minimizes $$\min_{\bm{\theta}_i=\bm{\theta}_j\text{ if }G_i=G_j, |\mathcal{G}|=K}\mathbb{E}\left[\sum_{i=1}^n\sum_{t=1}^T\rho_\tau(y_{it}-\mu_{0i}-\bm{\Pi}(x_{it})^\top\bm{\theta}_i)\right].$$ Obviously, we will have $\bm{\theta}_{0i}^{\mathcal{G}}=\bm{\theta}_{0i}^{\mathcal{G}_0}$. By the same argument in case 1, we have
	\begin{equation}
		\begin{split}
			&\Big|\sum_{i,t}\rho_\tau(y_{it}-\widetilde{\bm{\Pi}}(x_{it})^\top\bm{\widehat{\vartheta}}_{i}^{\mathcal{G}})-\sum_{i,t}\rho_\tau(y_{it}-\widetilde{\bm{\Pi}}(x_{it})^\top\bm{\vartheta}_{0i}^{\mathcal{G}})\Big|=O_p(H+nTH^{-2d}),\\
			\text{and }&\Big|\frac{1}{nT}\sum_{i,t}\rho_\tau(y_{it}-\widetilde{\bm{\Pi}}(x_{it})^\top\bm{\vartheta}_{0i}^{\mathcal{G}_0})-\mathbb{E}[\frac{1}{nT}\sum_{i,t}\rho_\tau(y_{it}-\widetilde{\bm{\Pi}}(x_{it})^\top\bm{\vartheta}_{0i}^{\mathcal{G}_0})]\Big|=O_p((nT)^{-1/2}).
		\end{split}
	\end{equation}
	Thus,
	\begin{equation}\begin{split}
			&\text{SIC}(K)-\text{SIC}(K_0)\\
			=&\log\left(1+\frac{\sum_{i,t}\rho_\tau(y_{it}-\widetilde{\bm{\Pi}}(x_{it})^\top\bm{\widehat{\vartheta}}_{i}^{\mathcal{G}})/(nT)-\sum_{i,t}\rho_\tau(y_{it}-\widetilde{\bm{\Pi}}(x_{it})^\top\bm{\widehat{\vartheta}}_{i}^{\mathcal{G}_0})/(nT)}{\sum_{i,t}\rho_\tau(y_{it}-\widetilde{\bm{\Pi}}(x_{it})^\top\bm{\widehat{\vartheta}}_{i}^{\mathcal{G}_0})/(nT)}\right)\\
			&+\frac{(K-K_0)H\log(nT)}{nT}\\
			=&\log\left(1+\frac{O_p(H+nTH^{-2d})}{\sum_{i,t}\rho_\tau(y_{it}-\widetilde{\bm{\Pi}}(x_{it})^\top\bm{\widehat{\vartheta}}_{i}^{\mathcal{G}_0})/(nT)}\right)+\frac{(K-K_0)H\log(nT)}{nT}\\
			=&O_p(H+nTH^{-2d})+\frac{(K-K_0)H\log(nT)}{nT}>0.
	\end{split}\end{equation}
	
\end{proof}

Finally, we present an auxiliary lemma used in the proof of Theorem \ref{thm:2}.

\begin{lemma}
	\label{lemma:6}For $c>0$ and $d>0$ sufficiently small,
	\begin{equation}
		\begin{split}
			\sup_{\substack{1\leq i\leq n,1\leq h\leq H\\\|\bm{\theta}_i-\bm{\widehat{\theta}}_i^\textup{o}\|\leq c,|\mu_i-\widehat{\mu}_i^\text{o}|\leq d}}&\Bigg|\frac{1}{T}\sum_{t=1}^T\bm{\Pi}_h(x_{it})[I\{e_{it}\leq0\}-I\{e_{it}\leq\bm{\Pi}(x_{it})^\top\bm{\theta}_i-m_{it}+\mu_i-\mu_{0i}\}\\
			&-F(0|x_{it})
			+F(\bm{\Pi}(x_{it})^\top\bm{\theta}_i-m_{it}+\mu_i-\mu_{0i}|x_{it})]\Bigg|\\
			=&O_p(H^{3/2}T^{-1/2}\log(T)\log(nT)).
		\end{split}
	\end{equation}
\end{lemma}

\begin{proof}[\textbf{Proof of Lemma \ref{lemma:6}}]
	
	We consider the upper bound for $\sum_{t=1}^T\bm{\Pi}_h(x_{it})[I\{e_{it}\leq0\}-I\{e_{it}\leq\bm{\Pi}(x_{it})^\top\bm{\theta}_i-m_{it}+\mu_i-\mu_{0i}\}-F_i(0|x_{it})+F_i(\bm{\Pi}(x_{it})^\top\bm{\theta}_i-m_{it}+\mu_i-\mu_{0i}|x_{it})$ only since the lower bound can be derived similarly. Letting $m_{it}(\bm{\theta}_i)=\bm{\Pi}(x_{it})^\top\bm{\theta}_i$ and $t_n$ satisfy that $|m_{it}(\bm{\theta}_i)-m_{it}(\bm{\widehat{\theta}}_i^\text{o})+\mu_i-\widehat{\mu}_{i}^\text{o}|\leq t_n$, we have
	\begin{equation}
		\begin{split}
			\sup_{\substack{1\leq i\leq n,1\leq h\leq H\\\|\bm{\theta}_i-\bm{\widehat{\theta}}_i^\text{o}\|\leq c,|\mu_i-\widehat{\mu}_i^\text{o}|\leq d}}&\frac{1}{T}\sum_{t=1}^T\bm{\Pi}_h(x_{it})\left[I\{e_{it}\leq m_{it}(\bm{\theta}_i)-m_{it}+\mu_i-\mu_{0i}\}-I\{e_{it}\leq0\}\right.\\
			&\left.+F(0|x_{it})-F(m_{it}(\bm{\theta}_i)-m_{it}+\mu_i-\mu_{0i}|x_{it})\right]\\
			\leq\sup_{\substack{1\leq i\leq n,1\leq h\leq H\\\|\bm{\theta}_i-\bm{\widehat{\theta}}_i^\text{o}\|\leq c,|\mu_i-\widehat{\mu}_i^\text{o}|\leq d}}&\frac{1}{T}\sum_{t=1}^T\bm{\Pi}_h(x_{it})\left[I\{e_{it}\leq m_{it}(\bm{\widehat{\theta}}_i^\text{o})-m_{it}+\widehat{\mu}_{i}^\text{o}-\mu_{0i}+t_n\}-I\{e_{it}\leq0\}\right.\\
			&\left.+F(0|x_{it})-F(m_{it}(\bm{\theta}_i)-m_{it}+\mu_i-\mu_{0i}|x_{it})\right]\\
			\leq\sup_{1\leq i\leq n,1\leq h\leq H}&\frac{1}{T}\sum_{t=1}^T\bm{\Pi}_h(x_{it})\left[I\{e_{it}\leq m_{it}(\bm{\widehat{\theta}}_i^\text{o})-m_{it}+\widehat{\mu}_{i}^\text{o}-\mu_{0i}+t_n\}-I\{e_{it}\leq0\}\right.\\
			&\left.+F(0|x_{it})-F(m_{it}(\bm{\widehat{\theta}}^\text{o}_i)-m_{it}+\widehat{\mu}_{i}^\text{o}-\mu_{0i}+t_n|x_{it})\right]\\
			+\sup_{\substack{1\leq i\leq n,1\leq h\leq H\\\|\bm{\theta}_i-\bm{\widehat{\theta}}_i^\text{o}\|\leq c,|\mu_i-\widehat{\mu}_i^\text{o}|\leq d}}&\frac{1}{T}\sum_{t=1}^T\bm{\Pi}(x_{it})[F(m_{it}(\bm{\widehat{\theta}}^\text{o}_i)-m_{it}+\widehat{\mu}_{i}^\text{o}-\mu_{0i}+t_n|x_{it})\\
			&-F(m_{it}(\bm{\theta}_i)-m_{it}+\mu_{i}-\mu_{0i})],
		\end{split}
	\end{equation}
	where the first inequality stems from the increasing monotonicity of the indicator function. The second term in the last line can be arbitrarily small since $|m_{it}(\bm{\theta}_i)-m_{it}(\bm{\widehat{\theta}}_i^\text{o})+\mu_i-\widehat{\mu}_i^\text{o}|\leq t_n$ while $t_n$ is arbitrarily small when we choose $c$ and $d$ to be sufficiently small.
	
	Note that $\mathbb{E}[|\bm{\Pi}_h(x_{it})(1\{e_{it}\leq a_n+\delta_n\}-1\{e_{it}\leq a_n\})|^q]\leq (C\sqrt{H})^{q-2}\delta_n$, for $q=3,4,\dots$
	By Lemma \ref{lemma:1}, for any non-negative sequences $a_n\to0$, $\delta_n\to0$, we have that for any $u>0$,
	\begin{equation}
		\label{eq:bernstein}
		\begin{split}
			&\mathbb{P}\left[\left|\frac{1}{T}\sum_{t=1}^T\bm{\Pi}_h(x_{it})[1\{e_{it}\leq a_n+\delta_n\}-1\{e_{it}\leq a_n\}-F(a_n+\delta_n|x_{it})+F(a_n|x_{it})]\right|>u\right]\\
			\leq&C\log(T)\exp\left[-C\frac{Tu^2}{\log(T)(u\sqrt{H}+\delta_n)}\right].
		\end{split}
	\end{equation}
	
	Denote $\bm{\vartheta}_i=(\mu_i,\bm{\theta}_i^\top)^\top$, $\bm{\vartheta}_{0i}=(\mu_{0i},\theta_{0i}^\top)^\top$ and $\widetilde{m}(\bm{\vartheta}_i)=m_{it}(\bm{\theta}_i)+\mu_i$. Let $A_i=\{\bm{\vartheta}_i:\|\bm{\vartheta}_i-\bm{\vartheta}_{0i}\|_2\leq C\xi(n,T)\}$. Similarly to Lemma \ref{lemma:T1_concentration}, we construct an $(nT)^{-\delta}$ covering of $A_i$ with size $R=O((nT)^{CH})$, with elements denoted by $\{\bm{\vartheta}_i^{(1)},\dots,\bm{\vartheta}_i^{(R)}\}$. Then, we have
	\begin{equation}
		\begin{split}
			\sup_{\substack{1\leq i\leq n,1\leq h\leq H\\\bm{\vartheta}_i\in A_i}}&\frac{1}{T}\sum_{t=1}^T\bm{\Pi}_h(x_{it})\Big[1\{e_{it}\leq \widetilde{m}_{it}(\bm{\vartheta}_i)-m_{it}-\mu_{0i}+t_n\}-1\{e_{it}\leq0\}\\
			&-F(\widetilde{m}_{it}(\bm{\vartheta}_i)-m_{it}-\mu_{0i}+t_n|x_{it})+F(0|x_{it})\Big]\\
			\leq\max_{\substack{1\leq i\leq n,1\leq h\leq H\\1\leq r\leq R}}&\frac{1}{T}\sum_{t=1}^T\bm{\Pi}_h(x_{it})\Big[1\{e_{it}\leq \widetilde{m}_{it}(\bm{\vartheta}_i^{(r)})-m_{it}-\mu_{0i}+t_n\}-1\{e_{it}\leq0\}\\
			&-F(\widetilde{m}_{it}(\bm{\vartheta}_i^{(r)})-m_{it}-\mu_{0i}+t_n|x_{it})+F(0|x_{it})\Big]\\
			+\sup_{\substack{1\leq i\leq n,1\leq h\leq H\\1\leq r\leq R,\|\bm{\theta}_i-\bm{\theta}_i^{(r)}\|\leq C(nT)^{-\delta}}}&\frac{1}{T}\sum_{t=1}^T\bm{\Pi}_h(x_{it})\Big[1\{e_{it}\leq \widetilde{m}_{it}(\bm{\vartheta}_i)-m_{it}-\mu_{0i}+t_n\}\\
			&-1\{e_{it}\leq \widetilde{m}_{it}(\bm{\vartheta}_i^{(r)})-m_{it}-\mu_{0i}+t_n\}\\
			&-F(\widetilde{m}_{it}(\bm{\vartheta}_i)-m_{it}-\mu_{0i}+t_n|x_{it})+F(\widetilde{m}_{it}(\bm{\vartheta}_i^{(r)})-m_{it}+t_n|x_{it})\Big]\\
			&:=I_1+I_2.
		\end{split}
	\end{equation}
	
	By \eqref{eq:bernstein}, using the union bound and that $|\widetilde{m}_{it}(\bm{\vartheta}_i^{(r)})-m_{it}-\mu_{0i}+t_n|\leq C\sqrt{H}\xi(n,T)$, we have
	\begin{equation}
		I_1=O_p((\sqrt{H}/T)\log(T)(H\log(nT))).
	\end{equation}
	
	For $I_2$, using the monotonicity of the indicator function and define $t_n'$ such that $|\widetilde{m}_{it}(\bm{\vartheta}_i)-\widetilde{m}_{it}(\bm{\vartheta}_i^{(r)})|\leq t_n'$ for all $\|\bm{\vartheta}_i-\bm{\vartheta}_i^{(r)}\|_2\leq C(nT)^{-\delta}$, we have
	\begin{equation}
		\begin{split}
			I_2\leq&\sup_{\substack{1\leq r\leq R,1\leq i\leq n,1\leq h\leq H\\\bm{\theta}_i\in A_i,\|\bm{\vartheta}_i-\bm{\vartheta}_i^{(r)}\|_2\leq C(nT)^{-\delta}}}\frac{1}{T}\sum_{t}\bm{\Pi}_h(x_{it})\Big[ 1\{e_{it}\leq \widetilde{m}_{it}(\bm{\vartheta}_i^{(r)})-m_{it}-\mu_{0i}+t_n'+t_n\}\\
			&-1\{e_{it}\leq \widetilde{m}_{it}(\bm{\vartheta}_i^{(r)})-m_{it}-\mu_{0i}+t_n\}-F(\widetilde{m}_{it}(\bm{\vartheta}_i)-m_{it}-\mu_{0i}+t_n|x_{it})\\
			&+F(\widetilde{m}_{it}(\bm{\vartheta}_i^{(r)})-m_{it}-\mu_{0i}+t_n|x_{it})\Big]\\
			\leq&\sup_{1\leq r\leq R,1\leq i\leq n,1\leq h\leq H}\frac{1}{T}\sum_{t}\bm{\Pi}_h(x_{it})\Big[ 1\{e_{it}\leq \widetilde{m}_{it}(\bm{\vartheta}_i^{(r)})-m_{it}-\mu_{0i}+t_n'+t_n\}\\
			&-1\{e_{it}\leq \widetilde{m}_{it}(\bm{\vartheta}_i^{(r)})-m_{it}-\mu_{0i}+t_n\}-F(\widetilde{m}_{it}(\bm{\vartheta}_i)-m_{it}-\mu_{0i}+t_n'+t_n|x_{it})\\
			&+F(\widetilde{m}_{it}(\bm{\vartheta}_i^{(r)})-m_{it}-\mu_{0i}+t_n|x_{it})\Big]\\
			+&\sup_{\substack{1\leq r\leq R,1\leq i\leq n,1\leq h\leq H\\\bm{\vartheta}_i\in A_i,\|\bm{\vartheta}_i-\bm{\vartheta}_i^{(r)}\|_2\leq C(nT)^{-\delta}}}\frac{1}{T}\sum_{t}\bm{\Pi}_h(x_{it})\Big[F(\widetilde{m}_{it}(\bm{\vartheta}_i)-m_{it}-\mu_{0i}+t_n'+t_n|x_{it})\\
			-&F(\widetilde{m}_{it}(\bm{\vartheta}_i)-m_{it}-\mu_{0i}+t_n|x_{it})\Big]\\
			=:&I_{21}+I_{22}.
		\end{split}
	\end{equation}
	Again by \eqref{eq:bernstein} for $I_{21}$ with union bound, and that $I_{22}$ is arbitrarily small by the smoothness of $F(\cdot)$, we obtain
	\begin{equation}
		I_2=O_p((\sqrt{H}/T)\log(T)(H\log(nT))).
	\end{equation}
	
\end{proof}

\end{document}